\documentclass[a4paper,reqno]{amsart}
\usepackage[T1]{fontenc}
\usepackage[english]{babel}
\usepackage{amssymb,bbm,enumerate}
\usepackage{bm}
\usepackage{xcolor}
\usepackage[linktocpage=true,colorlinks=true, linkcolor=blue, citecolor=red, urlcolor=green]{hyperref}

\usepackage{multirow} 
\usepackage{graphicx}

\usepackage{float}
\restylefloat{table}
\usepackage{tabu}
\usepackage{braket}
\usepackage{ytableau}
\usepackage{mathtools}
\usepackage[makeroom]{cancel}
\usepackage{caption}
\usepackage{subcaption}

\renewcommand{\Re}{\operatorname{Re}}
\renewcommand{\Im}{\operatorname{Im}}

\DeclareMathOperator{\im}{Im}

\newcommand{\lev}{\mathrm{lev}}
\newcommand{\sym}{\rho}
\newcommand{\fun}{\tau}
\newcommand{\vecmu}{\underline{\mu}}
\newcommand{\vectmu}{\underline{\tilde\mu}}

\newcommand{\snu}{\scriptscriptstyle{[\nu]}}

\newcommand{\ceili}[1]{ \left\lceil #1 \right\rceil}
\newcommand{\floori}[1]{ \left\lfloor #1 \right\rfloor}
\newcommand{\afr}[1]{ \left\langle #1 \right\rangle}
\renewcommand{\Re}{\operatorname{Re}}
\renewcommand{\Im}{\operatorname{Im}}

\newcommand{\bb}[1]{\mathbb{#1}}
\newcommand{\mc}[1]{\mathcal{#1}}

\newcommand{\e}{\varepsilon}
\newcommand{\la}{\lambda}

\newcommand{\C}{\mathbb{C}}
\newcommand{\R}{\mathbb{R}}

 \DeclareMathOperator{\res}{Res}

\theoremstyle{plain}
\newtheorem{theorem}{Theorem}[section]
\newtheorem{lemma}[theorem]{Lemma}

\newtheorem{proposition}[theorem]{Proposition}

\theoremstyle{definition}
\newtheorem{definition}[theorem]{Definition}

\theoremstyle{remark}
\newtheorem{remark}[theorem]{Remark}

\numberwithin{equation}{section}

\definecolor{light}{gray}{.9}

\title{On solutions of the Bethe Ansatz for the Quantum KdV model}

\author{R. Conti, D. Masoero}

\address{Grupo de F\'isica Matem\'atica da Universidade de Lisboa, Edif\'icio C6
Campo Grande, Lisboa, Portugal.}
\email{dmasoero@gmail.com}
\email{riccardo.conti92@gmail.com }

\begin{document}

\pagestyle{plain}

\begin{abstract}
We study the Bethe Ansatz Equations
for the Quantum KdV model, which are also known to be solved by the spectral determinants of a specific
family of anharmonic oscillators called monster potentials (ODE/IM correspondence).

These Bethe Ansatz Equations depend on two parameters 
 identified with the momentum and the degree at infinity of the anharmonic oscillators.
We provide
a complete classification of the solutions with only real and positive roots -- when the degree is greater than $2$ --
in terms of admissible sequences of holes. In particular, we prove
that admissible sequences of holes are naturally parameterised by integer partitions, and
we prove that they are in one-to-one correspondence with solutions of the Bethe Ansatz Equations, if the momentum is large enough.

Consequently, we deduce that the monster potentials are complete, in the sense that every solution of the Bethe Ansatz Equations
coincides with the spectrum of a unique monster potential. This
essentially (i.e. up to gaps in the previous literature)
proves the ODE/IM correspondence for the Quantum KdV model/monster potentials -- which was conjectured by Dorey-Tateo and
Bazhanov-Lukyanov-Zamolodchikov -- when the degree is greater than $2$.\\
Our approach is based on the transformation of the Bethe Ansatz Equations into a
free-boundary nonlinear integral equation -- akin to the equations known in the physics literature as
DDV or KBP or NLIE -- of which we develop the mathematical theory from the beginning.

\end{abstract}

\maketitle

\tableofcontents

\section{Introduction}
The Bethe Ansatz Equations (BAE) are arguably among the most important equations in Mathematical/Theoretical/Condensed-Matter Physics,
as well as the cornerstone of quantum integrability. They originated in the famous paper \cite{bethe31} of H. Bethe, who proposed his Ansatz to solve the XXX model.
Since then, and especially after the works of R. Baxter and of the Leningrad school led by L. Faddeev
(see \cite{baxter16,takhtadzhyan79}),
a wealth of quantum models have been found to be integrable by means of the Bethe Ansatz.\footnote{We refer the reader to \cite{baxter16,faddeev95}
for some historical context.}

In this paper we study the solutions of a two-parameters
family of BAE for 
an infinite dimensional system -- the parameters being the `momentum' $p \geq 0$ and the `degree' $2\alpha >0$.
The unknown is an entire function $Q$
of order $\frac{1+\alpha}{2\alpha}$ and with simple zeroes only, and the 
BAE consist of the following system of infinitely many identities for the zeroes of $Q$, which are referred to as Bethe roots
\begin{equation}
\label{eq:BAEQ} 
e^{- 4 \pi i p } \frac{Q\left(xe^{-\frac{2 \pi i}{1+\alpha}}\right)}{Q\left(xe^{\frac{2 \pi i}{1+\alpha}}\right)} =
-1 \;,\quad \forall x \in \bb{C} \mbox{ such that } Q(x)=0.
\end{equation}
In theoretical physics literature, solutions of the above BAE are shown to provide
(see \cite{doreyreview}, for more details and references)
\begin{itemize}
 \item the eigenvalues of the operator-valued-function $\mc{Q}_+(x)$ of a Conformal Field Theory known
 as Quantum KdV model \cite{baluzaI,baluzaII};
\item the leading asymptotics -- in the thermodynamic limit -- of the edge distribution of Bethe roots (i.e. the scaling limit)
of the XXZ model (with anisotropy $\Delta \in (-1,1)$)  \cite{klumper91}.
\end{itemize}
More precisely, we focus on solutions of the BAE (\ref{eq:BAEQ}) under the conditions that the Bethe roots are all real and positive.
To every such solution one associates a finite set of integers, namely the set of hole-numbers,
which are quantum numbers of the vacancies (holes) in the distribution of Bethe roots.
The admissible sets of hole-numbers are parameterised by a non-negative integer
$N$ and a partition of $N$.
The main result of the present paper is Theorem \ref{thm:mainIntro} below, in which
\begin{itemize}
 \item we prove that -- provided $\alpha>1$ and $p$ 
is sufficiently large --
for every admissible set of hole-numbers, i.e. for every partition, 
there exists a unique solution of the BAE (\ref{eq:BAEQ});
\item we obtain a uniform estimate for the Bethe roots in the large $p$ limit.
\end{itemize}
Our interest in the BAE is motivated by the discovery of Dorey-Tateo \cite{dorey98}, later generalised by
Bazhanov-Lukyanov-Zamolodchikov
\cite{bazhanov01,BLZ04}, that the spectral determinant of a family of anharmonic oscillators, known as
\textit{monster potentials}, also
fulfils the BAE (\ref{eq:BAEQ}).
In this case, it was more precisely conjectured by Bazhanov-Lukyanov-Zamolodchikov in \cite{BLZ04}
that there exists a bijection $\mc{O}$
between the Bethe states of the Quantum KdV model
and the monster potentials. In other words,
the Bethe roots of solutions of the BAE (\ref{eq:BAEQ}) should be the eigenvalues of a self-adjoint (if $p$ is large)
Schr\"odinger operator!\footnote{The ODE/IM correspondence can be thus thought of as a Hilbert-P\'olya conjecture for the
BAE. The distribution of Bethe roots turns out to be much more regular than the distribution of non trivial zeroes of the
zeta function.}
\textit{The fact that the BAE, encoding the integrability of a quantum field theory, provide the spectrum of 
a linear differential operator is the first and most fundamental instance of a
striking phenomenon, known as the ODE/IM correspondence \cite{doreyreview}}.

With the aim of understanding the origin
of such a fascinating phenomenon, we decided to tackle the proof of the ODE/IM correspondence conjecture for the Quantum KdV model/monster potentials. In fact, in our previous paper \cite{coma20} we proved -- assuming the existence of a certain Puiseaux series
as per \cite[Conjecture 5.9]{coma20} --
that monster potentials with $N$ apparent singularities,
are parameterised by partitions of $N$ and we computed the large momentum asymptotics of the eigenvalues,
that turn out to be real and positive in this regime.
As a consequence, combining the result of \cite{coma20} with the results of the present paper
(in particular, comparing the two large momentum asymptotics that we have found),
we prove -- up to \cite[Conjecture 5.9]{coma20} -- that \textit{the monster potentials are complete}:
For every solution of the BAE with real and positive roots, there exists a unique
monster potentials whose spectrum coincide with the given solution, if
$\alpha>1$ and $p$ is large.\footnote{This essentially proves
the ODE/IM correspondence conjecture for the Quantum KdV model/monster potentials, for $\alpha>1$.
However, to complete the definition of the bijection $\mc{O}$, one needs to parameterise the set of Bethe states of Quantum KdV
in terms of partitions and then compute the hole-numbers of the corresponding solutions of the BAE. This has not yet been
computed in the literature. See Section \ref{sec:ODEIMcorr} for a more detailed discussion about this point.}

\subsection{Notation}
\begin{itemize}
\item $\bb{R}^+= \lbrace x \in \bb{R}\,|\, x \geq 0\rbrace$.
\item $p,\alpha$ are real parameters with $\alpha>1$ and $p\geq0$.
\item $n,N,H$ are non-negative integers.
\item We denote by $[\nu]=(\nu_1,\dots,\nu_H)$ a partition of $N$ into $H\leq N$
integers, assumed to be decreasingly ordered.
\item $\ceili{ \cdot }$ is the ceiling function and $\floori{\cdot}$ is the floor function.
\item If $J\subset \bb{R}^+$ is the disjoint union of open, closed or half-closed intervals, we denote by
$\mc{C}^k(J)$ the space of $k-$times differentiable functions with domain $J$. Moreover, we denote by $\chi_J$ the characteristic
function of $J$.
\item $D_x=x \frac{d}{d x}$ is the Euler operator.
In general, we let $D_x^n=\left(x \frac{d}{d x}\right)^n$, where $D_x^{n=0}$ is the identity operator.
We also write $D$ for $D_x$ whenever we do not have to specify the co-ordinate with respect to which we take the derivative.
 \item We use the notation $A \lesssim B$ (resp. $A \gtrsim B $ ) to indicate that $A \leq C\, B$ (resp. $A \geq C \, B $ ),
where $C >0$ is an absolute constant that only depends on fixed parameters.
We also write $A \lesssim_{k} B$ to indicate that the implicit constant depends on a parameter $k$.
\item In long proofs with several estimates, the various estimates have been assigned the tags $(a)$,$(b)$,$(c)$, etc... to
facilitate the reading of the proof.
\end{itemize}

\subsection{The Bethe Ansatz Equations}
In this paper we address the following problems:
\begin{itemize}
    \item The classification of all solutions $Q$ of the BAE (\ref{eq:BAEQ}) such that\\
    (1) $Q$ is \textit{purely real}, that is all its roots are real and positive;\\
    (2) $Q$ is \textit{normalised},\footnote{The choice of the normalisation is made to Gauge-away the symmetry $Q(x) \to C_1 Q(C_2 x)$, with $C_1,C_2 \neq 0$, of
the BAE (\ref{eq:BAEQ}).} namely
\begin{align} \nonumber 
& Q(1)=1 \mbox{ and }
\lim_{x \to +\infty} x^{-\frac{1+\alpha}{2\alpha}}n_Q(x)=1 \;, \\ \label{eq:counting}
& \mbox{with } n_Q(x)=\left\lbrace y\in\R^+\,|\, Q(y)=0 \mbox{ and } y\leq x \right\rbrace \;.
 \end{align}
 (3) The set of holes -- which will be introduced below -- is finite.
 \item The computation of the asymptotic expansions of the above solutions when $p$ is large and positive.\footnote{The reader may wonder why we choose $p \in \bb{R}^+$ and we claim to study the large momentum limit, even though the
BAE are invariant under the transformation $p \to p + \frac{k}{2}$ with $k \in \bb{Z}$. We reassure the reader that
this will be explained in due course in this introduction.}
\end{itemize}
Following the literature \cite{hulthen38,orbach58,yangyang,destri92}, in order to gain further insight into the BAE
we take their logarithm.
To this aim, we introduce the \textit{associated $z$ function}.\footnote{The function $z$ is sometimes called the 
\textit{counting function}, but we choose to
reserve this name for the function defined previously in (\ref{eq:counting}).}
\begin{lemma}\label{lem:z}
To any purely real and normalised solution $Q$ of the BAE (\ref{eq:BAEQ}), we associate the following function
\begin{equation}\label{eq:defzanyalpha}
 z:[0,+\infty) \to \bb{R}:x\mapsto z(x)=-2p 
 + \frac{1}{2 \pi i} \log \frac{Q\left(xe^{-\frac{2 \pi i}{1+\alpha}}\right)}{Q\left(x e^{\frac{2 \pi i}{1+\alpha}}\right)} \;,
\end{equation}
where the branch of the $\log$ is chosen so that $z(0)=-2p$, i.e. $\log 1=0$.
The function $z$ satisfies the following properties:\\
(1) $z$ is real analytic;\\
(2) $z$ admits the following representation
 \begin{equation}\label{eq:defz}
 z(x)=-2p + \sum_{y\in\mc{R}} F_{\alpha}\left(\frac{x}{y}\right) \;,
\end{equation}
where $\mc{R}:=\lbrace y\in\R^+\,|\,Q(y)=0\rbrace $ is the set of roots and
\begin{align}\label{eq:Ffunza}
& F_{\alpha}: [0,+\infty) \to \left[0,\frac{\alpha-1}{1+\alpha}\right] \;, \\ \nonumber
& F_{\alpha}(x)= \frac{1}{2 \pi i} \log \left(\frac{1-x e^{-\frac{2 \pi i}{1+\alpha}}}
{1-x e^{\frac{2 \pi i}{1+\alpha}}} \right) 
    :=\frac{\sin\left(\frac{2 \pi}{1+\alpha}\right)}{\pi}\int_0^x \frac{dy}{1+y^2-2y\cos\left(\frac{2 \pi }{1+\alpha}\right)}\;.
\end{align}\\
(3) $z'(x)>0$ for all $x\in\R^+$;\\ 
(4) $z$ has the asymptotic behaviour
\begin{equation}\label{eq:xaz1}
 \lim_{x\to +\infty} x^{-\frac{1+\alpha}{2\alpha}}z(x)= 1\;.
\end{equation}
(5) The roots of $Q$ satisfy the identities
\begin{equation}\label{eq:logBA}
z(x)-\frac12 \in \bb{Z}\;,\quad \forall x \in \mc{R}\;.
% -2p + \sum_{y\in \mc{R}} F_{\alpha}\left(\frac{x}{y}\right)-\frac12 \in \bb{Z} \;,\quad \forall x\in\mc{R}.
\end{equation}
\end{lemma}
\begin{proof}
(1) Since $Q$ is a real entire function, then
 $\frac{Q\left(xe^{-\frac{2 \pi i}{1+\alpha}}\right)}{Q\left(x e^{\frac{2 \pi i}{1+\alpha}}\right)} $ is a meromorphic function that, restricted to
 the positive real semi axis, has modulus one. Hence its logarithm is analytic on the real semi axis and it is purely imaginary.\\
 (2) Since $\alpha>1$, the order of $Q$ is $\frac{\alpha+1}{2\alpha}\in\left(\frac{1}{2},1\right)$. Therefore
 after Hadamard Factorisation Theorem we have that
$Q(x)= \prod_{y\in\mc{R}} \left( 1-\frac{x}{y} \right)$, from which (\ref{eq:defz}) follows.\\
(3) From (\ref{eq:Ffunza}), we have that $F'_{\alpha}(x)$ is positive for all $x\in\R^+$. Then, the thesis follows
 from (\ref{eq:defz}).\\
 (4) See \cite[Lecture 12, Theorem 1]{levin96}.\\
 (5) Indeed $\frac{\log(-1)}{2\pi i}-\frac12 \in \bb{Z}$ for every branch of the logarithm.
\end{proof}

\subsection{Roots and Holes}
As we will show later in Proposition \ref{prop:strongz},
the system of equations (\ref{eq:logBA}) are equivalent to the BAE equations (\ref{eq:BAEQ}) if the function $z$
is strictly monotone.
This system is not exactly a system of real equations.
In fact, it states that if $x$ is a root of $Q$ then $z(x)-\frac12$ is an integer,
but it does not specify the value of this integer.
In other words, the BAE are quantisation conditions which do not encode the
information about which numbers are occupied or empty.

However, we can circumvent this problem reasoning as follows.
After Lemma \ref{lem:z} (3), if all roots of $Q$ are real and positive then $z'(x)>0$ for all $x\in\R^+$.
Therefore, the set of admissible quantum numbers is
\begin{align}\label{eq:Zp}
 \bb{Z}_p=\left\lbrace k \in \bb{Z}\,|\, k \geq -2p-\frac12 \right\rbrace \;,
\end{align}
and for each admissible quantum number $k\in \bb{Z}_p$,
we can uniquely associate the real and positive number $x_k$
\begin{equation}\label{eq:xk}
 x_k=z^{-1}\left(k+\frac12\right)\;.
\end{equation}
We name $x_k$ a (Bethe) root if it is a zero of $Q$, a hole otherwise.
Furthermore, we denote by $\bb{H}$ the subset of $\bb{Z}_p$ of hole-numbers
and by $\overline{\bb{H}}$ the set of roots
\begin{equation}
\label{eq:setholenum}
 \bb{H}=\left\lbrace k \in \bb{Z}_p \,|\, Q(x_k) \neq 0 \right\rbrace\;,\quad \overline{\bb{H}}=\bb{Z}_p\setminus\bb{H}\;.
\end{equation}
Once the set of holes is fixed, the system (\ref{eq:logBA}) becomes a well-defined system of equations for the roots
$\{x_k\}_{k\in \overline{\bb{H}}}$. In fact, it reads
\begin{equation*}
z(x_k)-\frac12=k \;,\quad \forall k \in\overline{\bb{H}}\;,
\end{equation*}
or, using the representation (\ref{eq:defz}) for $z$,
\begin{equation}\label{eq:logBAH}
-2p + \sum_{j\in\overline{\bb{H}}} F_{\alpha}\left(\frac{x_k}{x_j}\right)-\frac12 = k\;,\quad \forall k \in\overline{\bb{H}}\;.
\end{equation}
We call the latter system the logarithmic BAE.
In order to proceed further with our analysis we need to understand the structure of the set of holes.
The first coarse characterisation is given by two integer invariants, which we name \textit{sector} and \textit{level}:
Denoting by $\kappa$ the quantum number corresponding to the smallest root,
\begin{equation}\label{eq:kappa}
 \kappa = \min\{\overline{\bb{H}}\}\;,
\end{equation}
and by $\bb{H}_\kappa$ the set of hole-numbers larger than $\kappa$
\begin{equation}
\label{eq:Hkappa}
    \bb{H}_\kappa=\lbrace k \in \bb{H}\,|\, k > \kappa \rbrace \;,\quad H=|\bb{H}_\kappa|\;,
\end{equation}
we define the sector and level as
  \begin{align}\label{eq:secH}
 & \sec_{\bb{H}}=  \kappa + H \;, \\ \label{eq:levH}
 & \lev_{\bb{H}} =-  \frac{\left(\sec_{\bb{H}}-\kappa\right)(\sec_{\bb{H}}+\kappa-1)}{2}+\sum_{k \in \bb{H}_\kappa} k \geq 0\;.
 %& \lev_{\bb{H}} =-  \frac{\left(\sec_{\bb{H}}-\kappa\right)(\sec_{\bb{H}}-3\kappa-1)}{2}+\sum_{k \in \bb{H}_\kappa} k \geq 0\;.
\end{align}
In the above definition, $\sec_{\bb{H}}$ and $\lev_{\bb{H}}$ are defined to be $\infty$ if $\bb{H}$ is not a finite set.
\begin{remark}
We notice here that, since we assume 
that the set of holes is finite, then \textit{we can restrict,
without any loss in generality, to the case of a vanishing sector} $\sec_{\bb{H}}=0$, if $p$ is large enough. In fact, 
the transformation $$(z,p,x_k)\to\left(z+l,p-\frac{l}{2},x_{k+l}\right) \;,\quad l\in\bb{Z}\;,$$ %$z\toz+l$, $p \to p-\frac{l}{2}$ and $x_{k} \to x_{k+l}$, with $l$ integer,
is a symmetry of the logarithmic BAE (\ref{eq:logBAH}) which leaves the roots and the level unchanged but
induces a shift of the sector $\sec_{\bb{H}} \to \sec_{\bb{H}}+l$.\footnote{We could otherwise
assume that the parameter $p$ that appears in the BAE belongs to $\left[0,\frac12\right)$, admit any non-negative sector,
and define the effective momentum as $\hat{p}=p+\sec_{\bb{H}}$.}
\end{remark}

Fixed the sector to be $0$, the distinct sets of holes of level $N$ are naturally
parameterised by integer partitions of $N$. In fact, we have the following Lemma,
which was already proven in \cite[Appendix A]{BLZ04}.
\begin{lemma}\label{lem:holespartitions}
 Fix $N\in\bb{N}$ and assume that
 \begin{equation}\label{eq:seclevinequality}
   2p \geq N + \frac12 \;.
 \end{equation}
 Let $\mc{H}_N =\lbrace \bb{H} \subset \bb{Z}_p\,|\,\sec_{\bb{H}}=0 \mbox{ and } \lev_{\bb{H}} =N \rbrace$.
 
 The cardinality
 of $\mc{H}_N$ is the number of integer partitions of $N$. If
 $[\nu]=(\nu_1,\dots,\nu_H)$ with $H\leq N$ is a partition of $N$, the corresponding holes-subset of sector $0$
 and level $N$ is defined as
 \begin{equation}\label{eq:partitiontoholes}
  \bb{H}_{[\nu]}=\left\lbrace k \in \bb{Z}_p\,|\, k < -H \right\rbrace \bigcup_{l=1}^H \left\lbrace -H + \nu_l+(l-1) \right\rbrace \;.
  \end{equation}
\begin{proof}
   %Since the level is invariant under a shift, we can fix the sector to be $S=0$.
   
   Let $\bb{H}$ be a finite subset of $\bb{Z}_p$ as per (\ref{eq:setholenum}), and
   let $\kappa$, $\bb{H}_\kappa$ and $H$ be as per (\ref{eq:kappa}) and (\ref{eq:Hkappa}). Since 
   the sector vanishes then $\kappa=-H$ and we can write $\bb{H}_\kappa = \lbrace -H+m_k \rbrace_{k=1}^{H}$ for some
   $m_k\in\bb{Z}_p$ such that $m_{k+1}\geq m_k+1$ for any $k\in\{1,\dots,H-1\}$.
   Setting $\nu_k=m_k-(k-1)$, we have $\bb{H}_\kappa = \lbrace -H+\nu_k+(k-1)
   \rbrace_{k=1}^{H}$ with $\nu_{k+1} \geq \nu_k$ for any $k\in\{1,\dots,H-1\}$.
   By definition $\text{lev}_{\bb{H}}=N$, therefore formula (\ref{eq:levH}) yields
   $$N=-\frac{H(-H-1)}{2}+\sum_{k=1}^H \left(-H+\nu_k+(k-1)\right) \quad\Longleftrightarrow\quad \sum_{k=1}^H \nu_k =N\;. $$
   
   Therefore, we deduce that any finite subset of level $N$ and sector $0$ is
   characterised by a unique partition $[\nu]$ of $N$ in such a way that (\ref{eq:partitiontoholes}) holds.
  \end{proof}

\end{lemma}

\begin{figure}
\includegraphics[width=10cm]{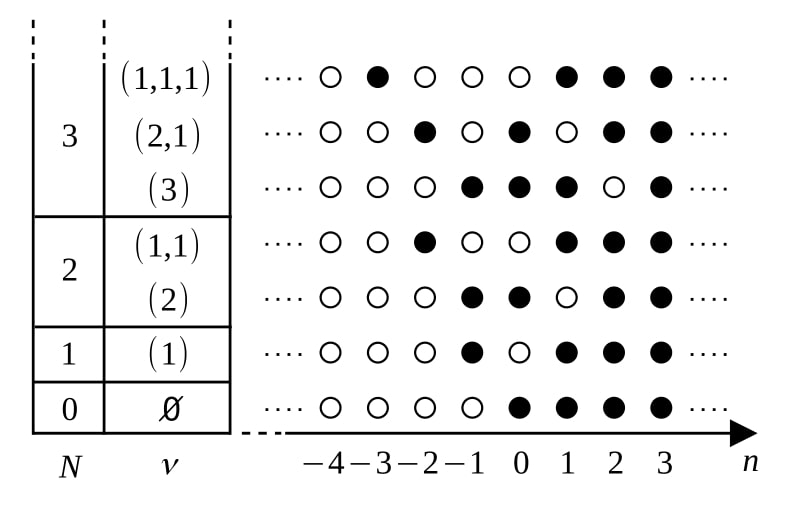}
\caption{Hole subsets of sector $S=0$ and level $N \leq 3$. Black dots are root-numbers, white dots are hole-numbers.}
\end{figure}

We can now state the main result of the present paper.
\begin{theorem}\label{thm:mainIntro}
 Let $N$ be a non-negative integer number and $[\nu]$ a partition of $N$.
 If $p$ is large enough there exists a unique purely real and normalised solution $Q$ of the BAE 
 (\ref{eq:BAEQ}) such that $\bb{H}=\bb{H}_{[\nu]}$, where $\bb{H}_{[\nu]}$ is as per definition (\ref{eq:partitiontoholes}).
 
 Moreover, the Bethe roots fulfil the following estimate
  \begin{align}\label{eq:xkintro}
\left| \frac{x_k(p)}{p^{\frac{2\alpha}{1+\alpha}}}-
A\left[ 1+
\left(\frac{\sqrt{2}\alpha}{(1+\alpha)^{\frac32}}\left(k+\frac12\right)\frac{1}{p} \right)\right] \right| \lesssim_k
\frac{1}{p^2}\;,\quad \forall k\in\bb{Z}_{p}\setminus\bb{}\bb{H}_{[\nu]}\;,
\end{align}
 where $A=(1+\alpha) \left(\frac{1}{\sqrt{\pi\alpha}}
 \frac{\Gamma \left(\frac{1}{2 \alpha }\right)}{\Gamma \left(\frac{1+\alpha}{2 \alpha }\right)}\right)^{\frac{2 \alpha }{\alpha +1}} $.
\end{theorem}
It must be noted that the implicit constant
$C_k$ in the estimate (\ref{eq:xkintro}) diverges as $k \to +\infty$.
In Section \ref{sec:mainthm}, we will actually derive much stronger estimates on $Q$ and $z$, that include a uniform estimate on the
position of the roots (see Theorem \ref{cor:main}) in terms of WKB integrals.

\subsection{Organisation of the paper}
The paper is organised as follows.

In Section \ref{sec:IBA}, we show that the logarithmic BAE (\ref{eq:logBAH}) is equivalent to a
\textit{free boundary} nonlinear integral equation for the
function $z$, that we call Integral Bethe Ansatz (IBA). Conversely, we prove that any solution of the IBA equation provides
a solution of the BAE if $z$ is strictly monotone.
We also show how the IBA can be conveniently split into a linearised integral equation (linear IBA) and a
nonlinear perturbation (perturbed IBA). Moreover, we solve the linear IBA by means of WKB integrals.
Finally, we compare the IBA equation with the nonlinear
integral equation studied in the physics literature and known as DDV (Destri-De Vega), KBP (Klumper-Batchelor-Pearce) or NLIE (Non-Linear Integral Equation) equation.

In Section \ref{sec:convoper}, we study many properties of the convolution operator, which appear in the IBA equation.

In Section \ref{sec:osc}, we study the oscillatory integrals which governs the nonlinear perturbation and we prove
that if $Q$ is a normalised solution of the BAE and $z$ its associated function, then
\begin{equation*}
  z(x)=x^{\frac{\alpha+1}{2\alpha}}-p \left(1+\alpha\right)-\frac{\sec_{\bb{H}}}{2}\left(\alpha-1\right)
 + O\left(x^{-\frac{\alpha+1}{2\alpha}}\right) \mbox{ as } x\to+\infty\;.
\end{equation*}
In Section \ref{sec:mainthm}, we use every bit of theory developed to prove our main Theorem.

In Section \ref{sec:ODEIMcorr}, we recall some notion of the theory of the Quantum KdV model and of the monster potentials. We
discuss the state-of-the-art of the ODE/IM correspondence and we establish a bijection among
purely real normalised solutions of the BAE and monster potentials.

\subsection{Note on the literature}
The BAE equations for the XXZ spin chain \cite{orbach58} / six vertex model \cite{lieb1967} 
have a long history and originated in the seminal
paper by H. Bethe \cite{bethe31} (see \cite{baxter16} and \cite{gaudin14} for some historical context).
The BAE (\ref{eq:BAEQ}) we study appeared in the physics literature 
in the context of the scaling
limit of the six vertex model/XXZ chain (see \cite{doreyreview} and references therein), then in the context of the Quantum KdV in the works \cite{baluzaI,baluzaII},
finally in the context of the ODE/IM correspondence in the works \cite{dorey98,bazhanov01}. In this paper,
we follow what seems to be, after the seminal works \cite{klumper91,destri92}, the approach more common in physics,
the one of transforming the BAE into a \textit{free-boundary} nonlinear integral equation, see e.g.
\cite{dunning00,fioravanti97,doreyreview}, among many. Of this equation,
we develop the mathematical theory from the very beginning.

Other approaches are possible. One is the variational approach which was pioneered by \cite{hulthen38,yangyang} and more recently used
in \cite{kozlowski18} to study the bulk asymptotic of roots of the XXZ limit in the thermodynamic limit;
the variational approach could certainly be used to study the edge asymptotic behaviour or, directly,
the logarithmic BAE (\ref{eq:logBAH}). 

A second alternative is the so-called TBA equation, which was studied recently in the mathematical
literature \cite{hilfiker20}; such an approach
has however the drawback that it is limited to the case of $2\alpha \in \bb{N}$ and $\bb{H}=\emptyset$ (ground-state).

Finally, the logarithmic BAE with $\bb{H}=\emptyset$ and $p=\pm \frac{1}{2\alpha+2}$
-- albeit written in a different
form -- was studied in \cite{avila04},\footnote{See \cite{eremenko20} for a recent application of this method.} after \cite{voros97},
using the tools of dynamical systems
and proven to admit a unique solution for every $\alpha>1$, which coincides with the spectrum of
the anharmonic oscillator $V(x)=x^{2\alpha+2}$.
It would be certainly interesting to extend the analysis
of \cite{avila04} to the case of a general $\bb{H}$ and a general $p$.

\subsection*{Acknowledgements}
We are deeply indebted to Stefano Negro, who generously
shared his great expertise on the DDV equation with us and allowed us to start our investigation \cite{stefanoprivate}, and
to Roberto Tateo, who is always available to clarify our doubts.
We are also grateful to Giordano Cotti for many stimulating conversations.
% , to Mattia Cafasso for checking Section 4.2

Finally, D.M. thanks Chiara Moneta for continuous support in these difficult times.

R.C. is supported by the FCT Project UIDB/00208/2020.
D.M. is partially supported by the FCT Project PTDC/MAT-PUR/ 30234/2017 ``Irregular
connections on algebraic curves and Quantum Field Theory'', by the
FCT Investigator grant IF/00069/2015 ``A mathematical framework for the ODE/IM correspondence'',
and by the FCT CEEC grant 2021.00091.CEECIND ``The Nonlinear Stokes Phenomenon. A unifying perspective on Integrable Models, Enumerative Geometry, and Special Functions''.

\section{Integral Bethe Ansatz}
\label{sec:IBA}
Here we transform the logarithmic BAE (\ref{eq:logBAH}) into an integral equation for the 
function $z$. To this aim we introduce the integral kernel $K_{\alpha}$
\cite{orbach58}
 \begin{equation} \label{eq:Kfunzal}
K_{\alpha}(x):= x F_{\alpha}'(x)=\frac{\sin\left(\frac{2 \pi}{1+\alpha}\right)}{\pi}
\frac{x}{1+x^2- 2 x \cos\left(\frac{2 \pi}{1+\alpha}\right)} \;,
\end{equation}
which enjoys the symmetry
\begin{equation}\label{eq:Ksymm}
 K_{\alpha}(x)=K_{\alpha}\left(\frac{1}{x}\right)\;.
\end{equation}
% \begin{equation}
%  F_{\alpha}(x)+F_{\alpha}(1/x)=\frac{2}{1+\alpha}
If $J\subset \bb{R}^+$ is measurable then we denote by $\bb{K}_J$ the convolution operator
with kernel $K_{\alpha}$
\begin{equation}\label{eq:convop}
 \bb{K}_J[f](x):= \int_J K_{\alpha}\left(\frac{x}{y}\right)\frac{f(y)}{y}dy \;.
\end{equation}
In the case $J=[\omega,+\infty)$ for some $\omega \in \bb{R}^+$, we write $\bb{K}_\omega:=\bb{K}_{[\omega,+\infty)}$.

We also introduce a class of closed subsets of $\bb{R}^+$.
\begin{definition}
\label{def:admsubset}
 A closed subset $I \subset \bb{R}^+$ is said admissible if there exist $2n+1$ with 
 $n\in\bb{N}$ positive real numbers $a_0<b_0<a_1<...<b_{n-1}<\omega$ such that
\begin{equation*}
 I= \bigsqcup_{i=0}^{n-1} \left[a_i,b_i\right] \bigsqcup [\omega,+\infty)\;.
%  , \mbox{   equivalently } \partial I= -a_{j}+\sum_{l=0}^{j-1} \left( b_l-a_l \right).
\end{equation*}
\end{definition}

\begin{proposition}\label{prop:intzI}
Let $Q$ be a purely real and normalised solution of the BAE (\ref{eq:BAEQ}), $z$ the associated 
function as per (\ref{eq:defz}), and $\bb{H}$ the set of hole-numbers as per (\ref{eq:setholenum}).
For every admissible subset $I$, the following identity holds
\begin{align}\label{eq:Iidentity}
& z(x)= - 2p + \bb{K}_I\left[\ceili{z-\frac12}\right](x)  +
   \sum_{k \in \overline{\bb{H}}_{\bar{I}} } F_{\alpha}\left(\frac{x}{x_k}\right)-
  \sum_{k \in \bb{H}_I } F_{\alpha}\left(\frac{x}{x_k}\right) \\ \nonumber
&- F_{\alpha}\left(\frac{x}{\omega}\right)\ceili{z(\omega)-\frac12} +
 \sum_{i=0}^{n-1} \left( F_{\alpha}\left(\frac{x}{b_i}\right)\floori{z(b_i)+\frac12}-
 F_{\alpha}\left(\frac{x}{a_i}\right)\ceili{z(a_i)-\frac12}\right)\;,
\end{align}
where $\bb{K}_I$ is as per (\ref{eq:convop}) while $\bb{H}_I$ and $\overline{\bb{H}}_{\bar{I}}$ are as follows
\begin{align}
    &\bb{H}_I= \lbrace k \in \bb{H}\,|\, x_k \in I  \rbrace\;, \quad \overline{\bb{H}}_{\bar{I}}=  \lbrace k \in \overline{\bb{H}}\,|\, x_k \in \R^+\setminus I  \rbrace \;.
\label{eq:HIdef}
\end{align}

\end{proposition}
\begin{proof}
(1) First we prove (\ref{eq:Iidentity}), assuming that $z(\omega)+\frac12 \notin \bb{Z}$,
$z(a_i)+\frac12 \notin \bb{Z}$ and $z(b_i)+\frac12 \notin \bb{Z}$ for all $i\in\{0,\dots,n-1\}$. In particular, this implies
that
$\ceili{z(b_i)-\frac12}=\floori{z(b_i)+\frac12}$ for all $i\in\{0,\dots,n-1\}$.

By definition of $z$, we have that
\begin{equation}\label{eq:zproof} \tag{a}
z(x)= - 2p + \sum_{k \in \overline{\bb{H}}}
F_\alpha\left(\frac{x}{x_k}\right)=- 2p +
\sum_{k \in \overline{\bb{H}}_I} F_\alpha\left(\frac{x}{x_k}\right)+
\sum_{k \in \overline{\bb{H}}_{\bar{I}}} F_\alpha\left(\frac{x}{x_k}\right)\;.
\end{equation}
with $\overline{\bb{H}}_{I}=  \lbrace k \in \overline{\bb{H}}\,|\, x_k \in I  \rbrace$.
Let us study the series
\begin{align}\label{eq:zdelta} \tag{b}
 \sum_{k \in\overline{\bb{H}}_I } 
 F_\alpha\left(\frac{x}{x_k}\right)+\sum_{k \in\bb{H}_I }
 F_\alpha\left(\frac{x}{x_k}\right) = \int_I F_\alpha\left(\frac{x}{y}\right)
 \sum_{k\in\bb{H}_I\cup\overline{\bb{H}}_I} \delta\left(y-x_k\right)   dy \;.
\end{align}
Considering the right-hand-side of the above expression as a Riemann-Stieltjes integral, we write
\begin{align}\label{eq:sumI} \tag{c}
 \sum_{k \in\overline{\bb{H}}_I } 
 F_\alpha\left(\frac{x}{x_k}\right)+\sum_{k \in\bb{H}_I }
 F_\alpha\left(\frac{x}{x_k}\right) =\int_I F_\alpha\left(\frac{x}{y}\right) d \ceili{z(y)-\frac12}\;,
\end{align}
since $d \ceili{z(x)-\frac12}=  \sum_{k \in \bb{Z}_p} \delta\left(x-x_k\right) dx$.
\footnote{The latter formula stems from the fact that $\ceili{z(x)-\frac12}$ is
constant on every interval of the form $(x_{k-1},x_k)$ while
$\lim_{\e \to 0^+}\ceili{z(x_k+\e)-\frac12}-\ceili{z(x_k-\e)-\frac12}=1$.}

Integrating (\ref{eq:sumI}) by parts we obtain
\begin{align}\label{eq:sumFproof} \tag{d}
 & \sum_{k \in\overline{\bb{H}}_I}  F_\alpha\left(\frac{x}{x_k}\right)=-\sum_{k \in\bb{H}_I}
 F_\alpha\left(\frac{x}{x_k}\right) -
 \int_I \left(\frac{d}{dy} F_\alpha\left(\frac{x}{y}\right) \right) \ceili{z(y)-\frac12} dy \\ \nonumber
 &-F_\alpha\left(\frac{x}{\omega}\right)\ceili{z(\omega)-\frac12}+
 \sum_{i=0}^{n-1} \left( F_\alpha\left(\frac{x}{b_i}\right)\ceili{z(b_i)-\frac12}
 -F_\alpha\left(\frac{x}{a_i}\right)\ceili{z(a_i)-\frac12}\right)\;.
\end{align}
In the last passage, we used the fact that 
$\lim_{y \to+\infty} F_\alpha\left(\frac{x}{y}\right)\ceili{z(y)-\frac12}=0$, since
$F_\alpha\left(\frac{x}{y}\right)=O(y^{-1})$ as $y\to+\infty$ while
$\ceili{z(y)-\frac12}=O\left(y^{\frac{\alpha+1}{2\alpha}}\right)$ as $y\to+\infty$, due to Lemma \ref{lem:z}(4).
Finally, combining identities (\ref{eq:zproof}),(\ref{eq:zdelta}) and (\ref{eq:sumFproof}), and using the identity
$\frac{d}{dy} F_{\alpha}\left(\frac{x}{y}\right) =-\frac{1}{y}K_{\alpha}\left(\frac{y}{x}\right)$
that follows from (\ref{eq:Kfunzal}) and (\ref{eq:Ksymm}), we obtain (\ref{eq:Iidentity}).\\
(2) Assume on the contrary that $z(\omega)-\frac12 \in \bb{Z}$ or there exists 
a $j \in \{ 1,\dots, n\}$ such that $z(a_j)-\frac12 \in \bb{Z}$ or
$z(b_j)-\frac12 \in \bb{Z}$.
Now, for any $\e$ sufficiently small we let $I^{\e}$ be the admissible set with $\omega^{\e}=\omega+\e$,
$a_i^{\e}=a_i+\e$, and
$b_i^{\e}=b_i + \e$.
If $\e$ is sufficiently small then $z(\omega^\e)-\frac12\notin \bb{Z}$ as well as
$z(a_i^{\e})-\frac12\notin\bb{Z}$ and $z(b_i^{\e})-\frac12 \notin \bb{Z}$ for all $i\in\{0,\dots,n-1\}$.
Hence by (1), identity (\ref{eq:Iidentity}) holds for the set $I^{\e}$. Now taking the limit  $\e \to 0$ of (\ref{eq:Iidentity}),
we deduce the thesis after some simple computations.
\end{proof}

\subsection{Standard form of the IBA equation}
While the identity (\ref{eq:Iidentity}) holds for any admissible set $I$, in our analysis we
stick to admissible sets of the form $I=[\omega,+\infty)$ with $\omega>0$. For such sets we shall prove the converse of Proposition \ref{prop:intzI}:
if $z$ solves the IBA equation (\ref{eq:Iidentity}) and it is strictly monotone
then it solves the logarithmic BAE (\ref{eq:logBAH}), hence it is the function associated to a purely real normalised solution
of the BAE (\ref{eq:BAEQ}).

As a first step, we write (\ref{eq:Iidentity}) in a more convenient form. 
After Lemma \ref{lem:holespartitions},
if $z$ is a solution of the logarithmic BAE (\ref{eq:logBAH}) whose set of hole-numbers $\bb{H}$ is of zero sector, then
the lowest root number is $\kappa=-H $ for some $H\geq 0$, and $\bb{H}$
coincides with the set $\bb{H}_{[\nu]}$ for some partition $[\nu]=(\nu_1,\dots,\nu_H)$ (see (\ref{eq:partitiontoholes})).
Moreover, if we choose $\omega>0$ such that
$\ceili{z(\omega)-\frac12}=-H$, the subset
of hole-numbers such that the corresponding hole belongs to $I$, i.e. the set $\bb{H}_I$ defined in Proposition \ref{prop:intzI},
can be represented as $\bb{H}_I=\lbrace -H+\sigma(k)\rbrace_{k=1}^H$ where $\sigma(k)=\nu_{H+1-k}-(H+1-k)$.
Therefore after Proposition \ref{prop:intzI},
the following identities hold
\begin{subequations}
  \label{eq:IBAStGen}
 \begin{alignat}{2}
     &z(x)= -2p + \bb{K}_\omega\left[\ceili{z-\frac12}\right](x)+H \, F_\alpha\left(\frac{x}{\omega}\right)-\sum_{k=1}^{H} F_\alpha\left(\frac{x}{h_k}\right)\;,\quad \forall x\in[\omega,+\infty) \;,
     \label{eq:IBASt}\\
     &z(h_k)= \sigma(k) +\frac12 \;,\quad \forall k \in\{1,\dots,H\} .\label{eq:IBASth}
     \end{alignat}
 \end{subequations}

Now, we shall prove the converse of Proposition \ref{prop:intzI}.
\begin{definition}
Fix a $H \geq 0$ and a strictly monotone function $\sigma: \lbrace 1,\dots, H \rbrace \to \bb{Z}$
such that $\sigma(1)>-H$.
 We say that the $(H+2)$-tuple $\{z,\underline{h},\omega\}$, with
 \begin{itemize}
  \item $\omega $ a positive real number;
  \item $z:[\omega,+\infty) \to \bb{R}^+$ a continuous function such that $\lim_{x\to+\infty}x^{-\frac{1+\alpha}{2\alpha}}z(x)=1$;
  \item $\underline{h}=(h_1,\dots,h_H) \in \bb{R}^H$ a vector whose components satisfy the inequalities $0<\omega<h_1<\dots<h_H$;
 \end{itemize}
 is a solution of the standard IBA equation if the system (\ref{eq:IBAStGen}) holds,
  together with the constraint
\begin{align}\label{eq:IBAStka}
& \ceili{z(\omega)-\frac12} =-H\;.
\end{align}
We say moreover that the solution is strictly monotone if $z$ is differentiable and
\begin{equation}
 \label{eq:IBAStSt}
 z'(x)>0\;,\quad \forall x\in[\omega,+\infty)\;.
\end{equation}
Finally, we say that two solutions $\{z,\underline{h},\omega\} $ and $\{\tilde{z},\underline{\tilde{h}},\tilde{\omega}\}$
are equivalent
if $z(x)=\tilde{z}(x)$ for all $x\geq \max\{\omega,\tilde\omega\}$.
\end{definition}

\begin{proposition}\label{prop:strongz}
Fix a $H \geq 0$ and a strictly monotone function $\sigma: \lbrace 1,\dots, H \rbrace \to \bb{Z}$
such that $\sigma(1)>-H$.\\
(1) If $\{z,\underline{h},\omega\}$ is a non-necessarily strictly monotone solution of the standard IBA equation, then
 $z$ extends to a holomorphic function on the domain 
$\mc{D}_{\omega}=\bb{C}
\setminus \left\lbrace e^{\pm \frac{2 \pi i}{\alpha+1}} x \,
|\, x \in[\omega,+\infty)\right\rbrace$, with $z(0)=-2p$.\\
(2) If $\{z,\underline{h},\omega\}$ is a strictly monotone solution of the standard IBA equation,
then the extension of $z$ to the positive real line is a solution of the logarithmic BAE (\ref{eq:logBAH}) with
set of hole-numbers
\begin{equation}
\label{eq:bbHsIBA}
 \bb{H}= \left\lbrace k\in\bb{Z}_p\,|\,k<-H \right\rbrace \cup \sigma\left(\{1,\dots,H\}\right)\;.
\end{equation}
In particular, $\bb{H}=\bb{H}_{[\nu]}$ with $[\nu]=(\nu_1,\dots,\nu_H)$ and $\nu_l=\sigma_{H+1-l}+l$.\\
(3) If $\{z,\underline{h},\omega\}$ is a strictly monotone solution of the standard IBA equation, then
$z$ coincides with the function associated to a purely real and normalised solution of the BAE (\ref{eq:BAEQ}).
  \begin{proof}
(1) (1i) The kernel $K_{\alpha}$ is a rational function of degree two, with simple poles at
 $x=e^{\pm\frac{2 \pi i}{\alpha+1} }$ and simple zeroes
 at $x=0$ and $x=\infty$.
 It is therefore analytic in the domain $\mc{D}_1$ and it is easily seen that
 $\bb{K}_{\omega}\left[\ceili{z-\frac12}\right](x)$ can be analytically continued on the domain
  $\mc{D}_{\omega}$ and that $\lim_{x\to 0}\bb{K}_{\omega}\left[\ceili{z-\frac12}\right](x)=0$.\\
  (1ii) After (\ref{eq:Kfunzal}), $K_{\alpha}(x)=D_x F_{\alpha}(x)$, hence
  the function $F_{\alpha}\left(\frac{x}{y}\right)$ -- with $y=\omega$ or $y=h_k$ -- is analytic in
  the domain $\mc{D}_{\omega}$ for all $y\geq \omega$. Moreover, after formula (\ref{eq:Kfunzal}),
  $F_{\alpha}(0)=0$.
  
  Combining (1i) with (1ii) and using (\ref{eq:IBASt}),
  we deduce the thesis. \\
  (2) After (1),  $z$ extends to a real analytic function on $\bb{R}^+$.
 Moreover, by hypothesis $z'(x)>0$ for all $x \in [\omega,+\infty)$.
 Reverting all the steps of the proof of Proposition \ref{prop:intzI} and using the constraint (\ref{eq:IBAStka}), we get
  \begin{align*}
 \bb{K}_\omega\left[\ceili{z-\frac12}\right]+H\,F_\alpha\left(\frac{x}{\omega}\right)=\sum_{y\in\mc{T}}  F_\alpha\left(\frac{x}{y}\right) \;,
\end{align*}
with $\mc{T}= \left\lbrace x \in [\omega,+\infty)\,|\, z(x)-\frac12 \in \bb{Z}\right\rbrace$. Whence
\begin{equation}\label{eq:propstrongz}
 z(x)=  -2p + \sum_{y\in \mc{R}} F_{\alpha}\left(\frac{x}{y}\right)\;,\quad \forall x \in \bb{R}^+\;,
\end{equation}
with $\mc{R}=\mc{T}\setminus \lbrace h_1,\dots,h_H\rbrace$.
From the above expression we have that $z'(x)>0$ for all $x\in\R^+$.

Since $z$ is strictly monotone, we can define the points $x_k=z^{-1}\left(k+\frac12\right)$ for all
$k \in \bb{Z}_p$. We observe that $x_k\in\mc{T}$ if and only if
$k\in\{k\in\bb{Z}_p\,|\,k\geq H\}$ and, consequently, $x_k\in\mc{R}$ if and only if
$k\in \bb{Z}_p\setminus\bb{H}$ with $\bb{H}$ as per (\ref{eq:bbHsIBA}). Therefore,
we deduce from (\ref{eq:propstrongz}) that $z$ satisfies the logarithmic BAE (\ref{eq:logBAH})
\begin{equation}\label{eq:proplogBAH}
 -2p + \sum_{j\in\bb{Z}_p\setminus\bb{H}} F_{\alpha}\left(\frac{x_k}{x_j}\right)-\frac12=k \;,\quad \forall k \in\bb{Z}_p\setminus\bb{H}\;.
\end{equation}
(3) Let $Q(x)=\prod_{k \in \bb{Z}_p\setminus \bb{H}} \left(1-\frac{x}{x_k} \right)$. Since
$z(x)=x^{\frac{1+\alpha}{2\alpha}}+o\left(x^{\frac{1+\alpha}{2\alpha}}\right)$ as $x\to+\infty$, then
$x_k=k^{\frac{2\alpha}{1+\alpha}}+ o\left(k^{\frac{2\alpha}{1+\alpha}}\right)$ as $k\to+\infty$, hence
the product converges to an entire function of order $\frac{1+\alpha}{2\alpha}$, and
due to (\ref{eq:propstrongz}), $z$ coincides with the function associated to $Q$
by formula (\ref{eq:defz}), namely
$$z(x)=-2p+
\log\frac{Q\left(xe^{-\frac{2 \pi i}{1+\alpha}}\right)}{Q\left(x e^{\frac{2 \pi i}{1+\alpha}}\right)}\;. $$
Moreover, from (\ref{eq:propstrongz}), we have that
\begin{equation*}
 0 \leq z(x)+2p- n_Q(x) \leq |\bb{H}|\;,
\end{equation*}
which implies that $Q(x)$ is normalised.
Finally, if $x_k$ is a root of $Q$, then
$$e^{- 4 \pi i p } \frac{Q\left(x_ke^{-\frac{2 \pi i}{1+\alpha}}\right)}{Q\left(x_k e^{\frac{2 \pi i}{1+\alpha}}\right)} = e^{2 \pi i
z(x_k)}=-1\;, $$
where in the last step we have used equation (\ref{eq:proplogBAH}).
 \end{proof}
\end{proposition}

\begin{remark}
 We notice here that, due to the constraint (\ref{eq:IBAStka}), the boundary of integration $\omega$ of the standard IBA
 is not fixed but depends on the solution of the equation. This is the reason why we say that the standard IBA equation
 is a free-boundary nonlinear integral equation. This fact makes its study particularly challenging. 
\end{remark}

\subsection{Linearised IBA equation}
It is useful to consider $\ceili{z(x)-\frac12}$ as a function oscillating about $z(x)$. For this reason, we define
\begin{equation}\label{eq:afr}
 \afr{\cdot}:\bb{R} \to \left(-\frac{1}{2},\frac12\right]:x\mapsto\afr{x}:=x-\ceili{x-\frac12} \;.
\end{equation}
Whence, in the standard IBA we shall write
\begin{equation*}
 \bb{K}_\omega \left[\ceili{z-\frac12}\right](x)=\bb{K}_\omega [z](x)-
\bb{K}_\omega [\afr{z}](x) \;.
\end{equation*}
Neglecting the nonlinear term $\afr{z}$ and setting $h_1=\dots=h_H=\omega$, from the standard IBA (\ref{eq:IBASt})
we obtain the following linear(ised) standard IBA equation for a function $l:[\omega,+\infty)\to\R^+$
\begin{equation}
\label{eq:IBAStL}
l(x)= -2p + \bb{K}_\omega[l](x)\;,\quad \forall x\in[\omega,+\infty).
\end{equation}

The above equation is a Wiener-Hopf integral equation. It was solved by means of the Mellin transform
(or the Fourier transform if one introduces the variable $\theta=\ln x$) in the physics literature,
see e.g. \cite{baluzaII,fioravanti03}. It was shown to have a unique solution once the normalisation
such that $l(x)=x^{\frac{1+\alpha}{2\alpha}}+o\left(x^{\frac{1+\alpha}{2\alpha}}\right)$ as $x\to+\infty$, and it was shown to admit
an explicit integral representation as the anti-Mellin transform of a ratio of Gamma functions.

Here we follow a different, more direct and convenient path. We provide the explicit solution of the linear IBA equation
in terms of a WKB integral.
Let us first consider the following function 
\begin{equation}
V(\cdot;x): \R^+ \to \R^+\;,\quad V(u;x):=x u^2-u^{2\alpha+2}- 1 \;,\quad x\in\R^+\;.
\label{eq:Vfun}
\end{equation}
In the next Lemma (whose proof is left to the reader), we discuss the properties of the real zeroes of $V(\cdot;x)$ as functions
of $x\in\R^+$.
\begin{lemma}
\label{lem:u12prop}
Let $x\in\R^+$ and define 
\begin{equation}
\label{eq:cutoff}
\sym:=\frac{\alpha+1}{\alpha} \alpha^{\frac{1}{\alpha+1}}\in(1,2]\;.
\end{equation}
The equation $V(u;x)=0$ admits
\begin{itemize}
\item two real positive and distinct roots at $u=u_1\in\left(0,\alpha^{-\frac{1}{2\left(1+\alpha\right)}}\right)$ and $u=u_2\in\left(\alpha^{-\frac{1}{2\left(1+\alpha\right)}},+\infty\right)$ if $x>\sym$;
\item two real positive and coincident roots at $u=\alpha^{-\frac{1}{2\left(1+\alpha\right)}}$ if $x=\sym$;
\item no real roots if $0\leq x<\sym$.
\end{itemize}
Let $u_1$ and $u_2$ be the real functions that fulfil the equation $V(u_k(x);x)=0$ with $k\in\{1,2\}$ for all $x\in[\sym,+\infty)$. Then,\\
(1) $u_1:[\sym,+\infty)\to\left(0,\alpha^{-\frac{1}{2\left(1+\alpha\right)}}\right]$ and $u_2:[\sym,+\infty)\to\left[\alpha^{-\frac{1}{2\left(1+\alpha\right)}},+\infty\right)$ are continuous;\\
(2) $u_1$ and $u_2$ restricted to $(\sym,+\infty)$ are smooth;\\
(3) $u_1$ is strictly decreasing, while $u_2$ is strictly increasing for all $[\sym,+\infty)$;\\
(4) the following asymptotics hold at large positive $x$
\begin{align}
&u_1(x)=x^{-\frac12}+O\left(x^{-\frac32-\alpha}\right)%x^{-\frac12}\left(1+\frac12 x^{-1-\alpha}\right)
      \mbox{ as } x\to+\infty \;,\notag \\
&u_2(x)=x^{\frac{1}{2\alpha}}+O\left(x^{-\frac{1}{2\alpha}-1}\right)%x^{\frac{1}{2\alpha}}\left(1-\frac{1}{2\alpha}x^{-\frac{1+\alpha}{\alpha}}\right)
      \mbox{ as } x\to+\infty \;;
    \label{eq:u12asy}
\end{align}
(5) the following expansions hold in a right neighborhood of $x=\sym$
\begin{align}
&u_k(x)=\alpha^{-\frac{1}{2\left(1+\alpha\right)}}+\sum_{n\geq 1}(-1)^{kn}c_n\left(x-\sym\right)^{\frac{n}{2}} \mbox{ as } x\to \sym^+ \;,\notag\\
&\mbox{with } \{c_n\}_{n\geq 1}\in\R\;,\quad c_1 = \frac{\alpha ^{-\frac{1}{\alpha +1}}}{\sqrt{2\alpha +2}}\;.
\label{eq:u12lim}
\end{align}
\end{lemma}
Let us define
\begin{equation}
S(x):= \begin{cases}
 \displaystyle{\frac{1}{\pi} \int_{u_1(x)}^{u_2(x)} \sqrt{V(u;x)}\,\frac{du}{u}\;,\quad
 x>\frac{\alpha+1}{\alpha} \alpha^{\frac{1}{\alpha+1}}}\\
 \displaystyle{0 \;,\quad 0\leq x\leq\frac{\alpha+1}{\alpha} \alpha^{\frac{1}{\alpha+1}}}
 \end{cases}\label{eq:Jfun}\;.
 \end{equation}
 where $V(\cdot;x)$ is as per (\ref{eq:Vfun}), while $u_1(x)$ and $u_2(x)$ are the two zeroes of $V(\cdot;x)=0$  studied in the Lemma
 above. The linear IBA can be explicitly solved in terms of the function $S$.
 
 In fact, we have the following proposition
\begin{proposition}\label{prop:tauJ}
The function $\fun: \R^+ \to \R^+ $ defined by the formula
 \begin{equation}\label{eq:tauJ}
  \fun(\xi):=2\left(1+\alpha\right) S\left(\frac{\alpha+1}{\alpha} \alpha^{\frac{1}{\alpha+1}}\xi\right)\;,\quad \xi\in\R^+,
 \end{equation}
satisfies the following properties:\\
(1) $\fun:\R^+\to\R^+$ is continuous and $\fun(\xi)=0$ for all $\xi\in\left[0,1\right]$;\\
(2) $\fun$ restricted to $\left[1,+\infty\right)$ is smooth;\\
(3) $\fun$ fulfils the following inhomogeneous integral equation
 \begin{equation}\label{eq:inhomWH}
    \fun(\xi)=-2+\bb{K}_1[\fun](\xi)\;,\quad \forall \xi\in[1,+\infty)\;,
\end{equation}
where $\bb{K}_1:=\bb{K}_{[1,+\infty)}$ is as per (\ref{eq:convop}).

The function $D\fun$ -- with $D$ the Euler operator -- fulfils the associated homogeneous integral equation 
\begin{equation}
\label{eq:homWH}
D_\xi\fun(\xi)=\bb{K}_1[D\fun](\xi) \;,\quad \forall \xi\in[1,+\infty)\;;
\end{equation}
(4) the following asymptotic holds
\begin{equation}
        \fun(\xi)=(A\xi)^{\frac{1+\alpha}{2\alpha}}-\left(1+\alpha\right)+O\left(\xi^{-\frac{1+\alpha}{2\alpha}}\right) \mbox{ as } \xi\to+\infty\;,
        \label{eq:fasy}
    \end{equation}
    with
    \begin{equation}
    \label{eq:Adef}
    A=(1+\alpha) \left(\frac{1}{\sqrt{\pi\alpha}}
 \frac{\Gamma \left(\frac{1}{2 \alpha }\right)}{\Gamma \left(\frac{1+\alpha}{2 \alpha }\right)}\right)^{\frac{2 \alpha }{\alpha +1}}\;.
   \end{equation}
    More generally, for all $n\in\bb{N}^+$
    \begin{equation}
        D_\xi^n\fun(\xi)=\left(\frac{1+\alpha}{2\alpha}\right)^n(A\xi)^{\frac{1+\alpha}{2\alpha}}+
        O\left(\xi^{-\frac{1+\alpha}{2\alpha}}\right) \mbox{ as } \xi\to+\infty\;;
        \label{eq:Dnfasy}
    \end{equation}
(5) the derivative of $\fun$ has a well-defined right limit at $1$,
  \begin{equation}\label{eq:fp1}
    \lim_{\xi \to 1^+}\fun'(\xi)= \frac{\left(1+\alpha\right)^{\frac32}}{\sqrt{2}\alpha} \;;
   \end{equation}
(6) $\fun$ is strictly increasing for all $\xi\in[1,+\infty)$ and we have that
    \begin{equation}\label{eq:infderf}
       \underset{\xi\in[1,+\infty)}{\inf}\left(\xi^{-\frac{1+\alpha}{2\alpha}}D_\xi \fun(\xi)\right)>0\;;
    \end{equation}
\end{proposition}
\begin{proof}
    (1)-(2) Follow immediately from the definition (\ref{eq:Jfun}) and Lemma \ref{lem:u12prop};\\
    (3) Let us prove (\ref{eq:inhomWH}) first. From the definition (\ref{eq:tauJ}) we have 
    \begin{equation} \label{eq:item31}\tag{a}
    \bb{K}_1[\fun](\xi)=2\left(1+\alpha\right)\bb{K}_{\sym}[S]\left(\sym\xi\right),
    \end{equation}
    where $\sym$ is as per (\ref{eq:cutoff}).
    Let us consider the integral
 \begin{equation*}
     \bb{K}_{\sym}[S](x)=\int_{\sym}^{+\infty} K_\alpha\left(\frac{x}{y}\right)S(y)\frac{dy}{y}=\int_{\sym}^{+\infty}K_\alpha\left(\frac{x}{y}\right)\left(\int_{u_1(y)}^{u_2(y)} \sqrt{V(u;y)}\,\frac{du}{\pi u}\right)\frac{dy}{y}\;.
 \end{equation*}
 After Lemma \ref{lem:u12prop} we have
 \begin{equation*}
 u_1(\sym)=u_2(\sym)=u^*:=\alpha^{-\frac{1}{2\left(1+\alpha\right)}} \;,\quad\lim_{x\to+\infty}u_1(x)=0\;,\quad \lim_{x\to+\infty} u_2(x)=+\infty\;.
 \end{equation*}
 Then, exchanging the order of integration per Fubini Theorem, we get
 \begin{multline} \label{eq:item32} \tag{b}
    \bb{K}_{\sym}[S](x)=\frac{1}{\pi}\lim_{\delta\to0^+}\int_\delta^{+\infty}
  \left(\int_{y^*(u)}^{+\infty}K_\alpha\left(\frac{x}{y}\right)\sqrt{V(u;y)}\,\frac{dy}{y}\right)\frac{du}{u} \\
  =\frac{1}{2\pi}\lim_{\delta\to0^+}\int_\delta^{+\infty} \left[\sqrt{V\left(u;xe^{\frac{2 i \pi }{\alpha +1}}\right)}-\sqrt{V\left(u;xe^{-\frac{2 i \pi }{\alpha +1}}\right)}\right]\frac{du}{u}\\
  =\frac{1}{2\pi}\lim_{\delta\to0^+}\left[\int_{e^{\frac{i\pi}{1+\alpha}}\delta}^{e^{\frac{i\pi}{1+\alpha}}\infty}-\int_{e^{-\frac{i\pi}{1+\alpha}}\delta}^{e^{-\frac{i\pi}{1+\alpha}}\infty}\right]\left(\sqrt{V(u;x)}\,\frac{du}{u}\right)\;,
 \end{multline}
 where
 \begin{equation*}
 y^*(u)=\frac{u^{2\alpha+2}+1}{u^2}\;.
 \end{equation*}
 %where we used the fact that
 %\begin{multline*}
 %    \int K_\alpha\left(\frac{x}{y}\right)\sqrt{u^{2\alpha+2}-y u^2 + 1}\frac{dy}{y}=\\
 %    -\frac{2}{\pi}\im{\left[\sqrt{u^{2 \alpha +2}-xe^{\frac{2 i \pi }{\alpha +1}} u^2+1}\,\text{arctanh}\left(\sqrt{\frac{u^{2\alpha+2}-yu^2+1}{u^{2 \alpha +2}-xe^{\frac{2 i \pi }{\alpha +1}} u^2+1}}\right)\right]}
 %\end{multline*}
 In the latter expression, we used contour integration in the complex plane to show that
 \begin{equation*}
     \int_{y^*(u)}^{+\infty}K_\alpha\left(\frac{x}{y}\right)\sqrt{V(u;y)}\,\frac{dy}{y}=\frac12\left[\sqrt{V\left(u;xe^{\frac{2 i \pi }{\alpha +1}}\right)}-\sqrt{V\left(u;xe^{-\frac{2 i \pi }{\alpha +1}}\right)}\right]\;.
 \end{equation*}
 More precisely, we choose a closed anti-clockwise contour around the branch cut going from $y^*(u)$ to $+\infty$ (see Figure \ref{fig:contour1}). Then, the result follows from Cauchy's Theorem, taking into account the residues coming from the simple poles of the integrand at $y=e^{\pm\frac{2\pi i}{1+\alpha}}x$.

\begin{figure}[t]
\centering
     \begin{subfigure}{0.3\textwidth}
         \centering
         \includegraphics[width=1.75\textwidth]{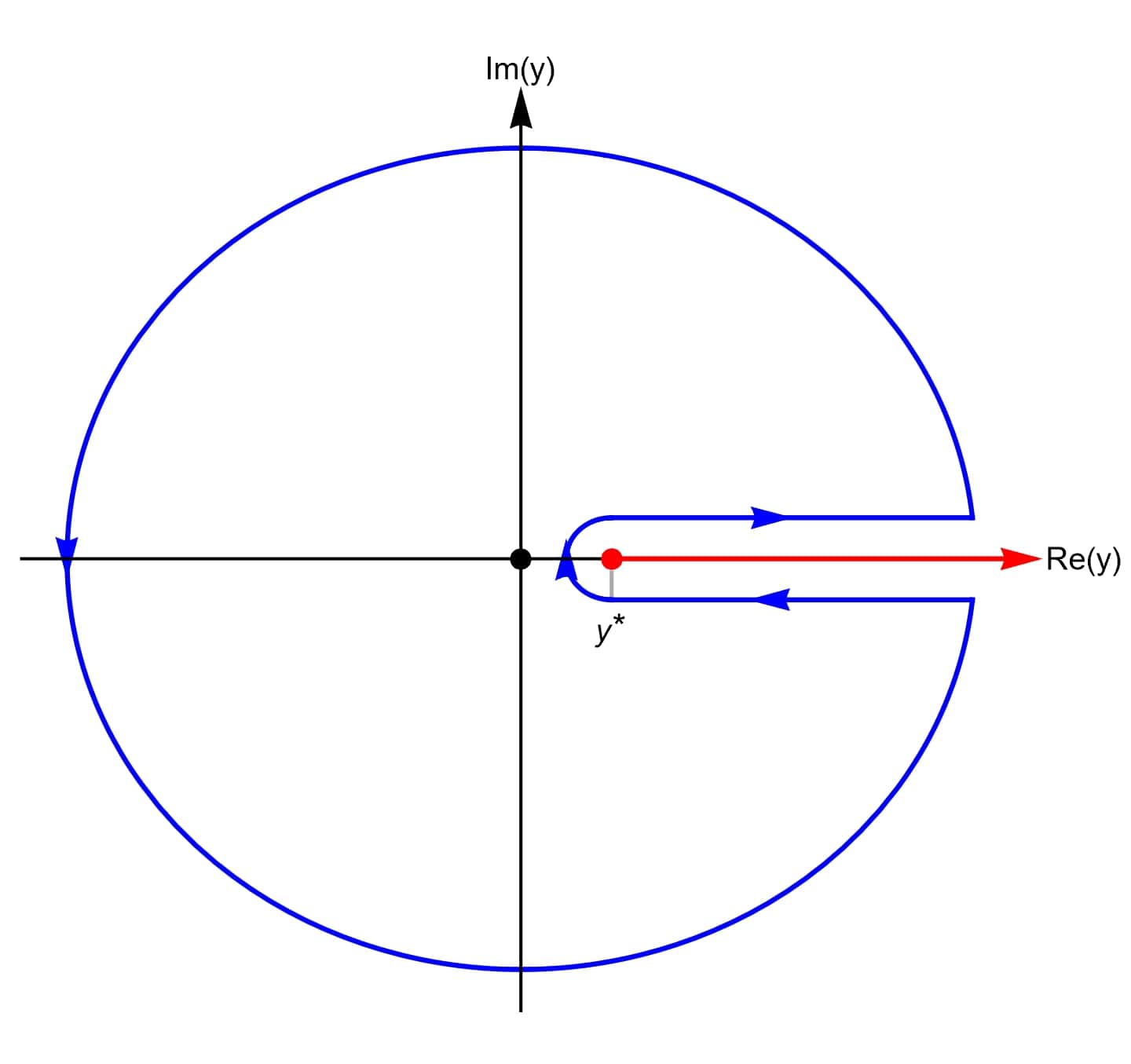}
         \caption{}
         \label{fig:contour1}
     \end{subfigure}
     \hfill
     \begin{subfigure}{0.3\textwidth}
         \centering
         \includegraphics[width=1.75\textwidth]{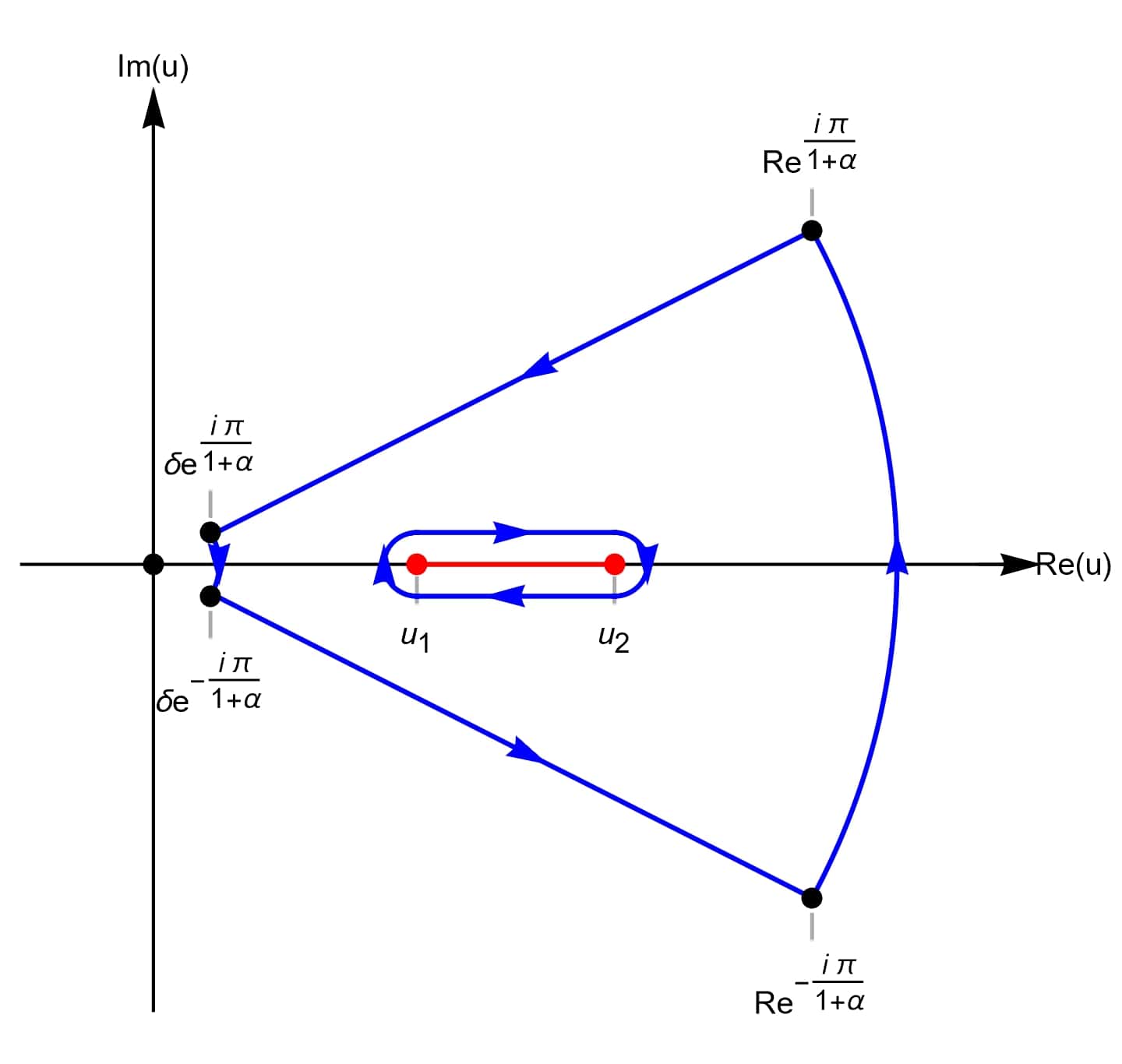}
         \caption{}
         \label{fig:contour2}
     \end{subfigure}
        \caption{Contours in the complex plane.}
        \label{fig:contours}
\end{figure}

To further manipulate the expression obtained in (\ref{eq:item32}), let us consider the integral of the function $\frac{1}{\pi u}\sqrt{V(u;x)}$ along a closed anti-clockwise contour made of (see Figure \ref{fig:contour2}):
\begin{itemize}
\item a clockwise circular arc of radius $\delta$ and amplitude $\frac{2\pi}{1+\alpha}$ around $u=0$, that we label as $\gamma_\delta(0)$;
\item a line going from $\delta e^{-\frac{i\pi}{1+\alpha}}$ to $\infty$ along the direction $-\frac{\pi}{1+\alpha}$;
\item a line going from $\infty$ to $\delta e^{\frac{i\pi}{1+\alpha}}$ along the direction $\frac{\pi}{1+\alpha}$;
\item a clockwise dog bone contour around the cut from $u_1$ to $u_2$.
\end{itemize}
Since all the singularities of the integrand are outside this contour, after Cauchy Theorem we have
\begin{equation*}
    2S(x)+\frac{1}{\pi}\lim_{\delta\to 0^+}\left[\int_{e^{\frac{i\pi}{1+\alpha}}\infty}^{e^{\frac{i\pi}{1+\alpha}}\delta}+
    \int_{\gamma_\delta(0)}+\int_{e^{-\frac{i\pi}{1+\alpha}}\delta}^{e^{-\frac{i\pi}{1+\alpha}}\infty}\right]\left(\sqrt{V(u;x)}\,\frac{du}{u}\right)=0 \;,
\end{equation*}
which implies that
\begin{multline} \label{eq:item33} \tag{c}
    \bb{K}_{\sym}[S](x)=S(x)+\frac{1}{2\pi}\lim_{\delta\to0^+}\int_{\gamma_\delta(0)}\sqrt{V(u;x)}\,\frac{du}{u} \\
    = S(x)-\frac{i}{1+\alpha}\res\left(\frac{\sqrt{V(u;x)}}{u},u=0\right)=S(x)+\frac{1}{1+\alpha}\;,
\end{multline}
where we used the fact that the integral of a function over an anti-clockwise circular arc of amplitude $\theta$ -- and whose radius tends to zero -- centred in a simple pole yields the residue of the function in that pole multiplied by a factor $i\theta$. Finally, plugging (\ref{eq:item33}) in (\ref{eq:item31}) and using the definition (\ref{eq:tauJ}) we obtain (\ref{eq:inhomWH}).

Now, we shall briefly present the proof of (\ref{eq:homWH})
omitting the details of the computations, since it follows the same step of the proof of (\ref{eq:inhomWH}).
From the definition (\ref{eq:tauJ}) we have
\begin{equation}\label{eq:item34} \tag{d}
\bb{K}_1[D\fun](\xi)= 2\left(1+\alpha\right)\bb{K}_\sym[DS](\sym\xi)\;.
\end{equation}
Let us consider the integral
 \begin{multline*}
     \bb{K}_{\sym}[DS](x)=\int_{\sym}^{+\infty} K_\alpha\left(\frac{x}{y}\right)D_yS(y)\frac{dy}{y}=\int_{\sym}^{+\infty} K_\alpha\left(\frac{x}{y}\right)S'(y)\,dy\\
     =\frac{1}{2\pi}\int_{\sym}^{+\infty}K_\alpha\left(\frac{x}{y}\right)\left(\int_{u_1(y)}^{u_2(y)} \frac{u}{\sqrt{V(u;y)}}\,du\right)dy\;,
 \end{multline*}
 where we used Leibniz integral rule to evaluate
\begin{equation} \label{eq:item35} \tag{e}
    S'(x)=\frac{1}{2\pi}\int_{u_1(x)}^{u_2(x)}\frac{u}{\sqrt{V(u;x)}}\,du \;.
\end{equation}
We treat the integral $\bb{K}_{\sym}[DS](x)$ in the same way as $\bb{K}_{\sym}[S](x)$. Performing an analogous computation we arrive at
\begin{equation}\label{eq:item36} \tag{f}
\bb{K}_{\sym}[DS](x)= \frac{1}{4\pi}\left[\int_{0}^{e^{\frac{i\pi}{1+\alpha}}\infty}-\int_{0}^{e^{-\frac{i\pi}{1+\alpha}}\infty}\right]\left(\frac{u}{\sqrt{V(u;x)}}\,du\right)= S(x)\;,
\end{equation}
where in the last equality we used the Cauchy Theorem.  Finally, plugging (\ref{eq:item36}) in (\ref{eq:item34}) and using the definition (\ref{eq:tauJ}) we obtain (\ref{eq:homWH}).

(4) Let us perform the change of variables $u=x^{\frac{1}{2\alpha}}t$ in the integral appearing in the definition of $S(x)$
\begin{equation} \label{eq:item41} \tag{g}
S(x) = \frac{x^{\frac{1+\alpha}{2\alpha}}}{\pi}\int_{x^{-\frac{1}{2\alpha}}u_1(x)}^{x^{-\frac{1}{2\alpha}}u_2(x)} \sqrt{1-t^{2\alpha}}\left(1-\frac{x^{-\frac{1+\alpha}{\alpha}}}{t^2\left(1-t^{2\alpha}\right)}\right)^{\frac12}dt\;.
\end{equation}
After (\ref{eq:u12asy}), we have that $x^{-\frac{1}{2\alpha}}u_1(x)=x^{-\frac{1+\alpha}{2\alpha}}+O\left(x^{-\frac{1+\alpha}{2\alpha}-1-\alpha}\right)$ and $x^{-\frac{1}{2\alpha}}u_2(x)=1+O\left(x^{-\frac{1+\alpha}{\alpha}}\right)$ as $x\to+\infty$, whence
%To derive the asymptotic expansion of $J(x)$ at large positive $x$, it is convenient to rewrite $J$ as follows
%\begin{multline*}
%\left(1-\frac{x^{-\frac{1+\alpha}{\alpha}}}{t^2\left(1-t^{2\alpha}\right)}\right)^{\frac12} = 1+\sum_{n\geq 1} \frac{\left(-\frac{1}{2}\right)_n}{\Gamma (n+1)} \frac{x^{-2n\frac{1+\alpha}{2\alpha}}}{t^{2 n} \left(1-t^{2 \alpha }\right)^{n}}\\
%J(x)=\frac{x^{\frac{1+\alpha}{2\alpha}}}{\pi}\int_{u_1 x^{-\frac{1}{2\alpha}}}^{u_2 x^{-\frac{1}{2\alpha}}} \sqrt{1-t^{2\alpha}}\,dt \\+ \frac{1}{\pi}\sum_{n\geq 1}\frac{\left(-\frac{1}{2}\right)_n}{n!}x^{\left(1-2n\right)\frac{1+\alpha}{2\alpha}}\int_{u_1 x^{-\frac{1}{2\alpha}}}^{u_2 x^{-\frac{1}{2\alpha}}} t^{-2 n} \left(1-t^{2 \alpha }\right)^{\frac12-n}dt 
%\end{multline*}
%where we first expanded the square root around $x=+\infty$, then we interchanged summation and integration using the Fubini-Tonelli Theorem. From the latter expression we have
\begin{equation} \label{eq:item42} \tag{h}
    K:=\lim_{x\to+\infty}\left(x^{-\frac{\alpha+1}{2\alpha}}S(x)\right)=\frac{1}{\pi}\int_0^1\sqrt{1-t^{2\alpha}}\,dt=\frac{1}{2 \sqrt{\pi}} \frac{\Gamma\left(\frac{1+2\alpha}{2\alpha}\right)}{\Gamma\left(\frac{1+3\alpha}{2\alpha}\right)}
 \;.
\end{equation}
Using (\ref{eq:item42}) and the definition (\ref{eq:tauJ}), we find
\begin{equation*}
\lim_{\xi\to+\infty}\left(\xi^{-\frac{\alpha+1}{2\alpha}}\fun(\xi)\right) = 2\left(1+\alpha\right)\sym^{\frac{1+\alpha}{2\alpha}} K = A^{\frac{1+\alpha}{2\alpha}}\;,
\end{equation*}
where $A$ is as per (\ref{eq:Adef}), thus proving the leading asymptotic behavior in (\ref{eq:fasy}). Going further, the sub-leading asymptotic behavior yields
\begin{multline} \label{eq:item43} \tag{i}
    \lim_{x\to+\infty} \left(S(x)-Kx^{\frac{\alpha+1}{2\alpha}}\right) = \lim_{x\to+\infty}\left[ -\frac{x^{\frac{\alpha+1}{2\alpha}}}{\pi}\int_0^{x^{-\frac{\alpha+1}{2\alpha}}}\sqrt{1-t^{2\alpha}}\,dt\right.\\
    \left.+\frac{1}{\pi}\sum_{n\geq 1}\frac{\left(-\frac{1}{2}\right)_n}{n!}x^{\left(1-2n\right)\frac{1+\alpha}{2\alpha}}\int_{x^{-\frac{1+\alpha}{2\alpha}}}^1 t^{-2 n} \left(1-t^{2 \alpha }\right)^{\frac12-n}dt \right] = -\frac12\;,
\end{multline}
where we first expanded the square root in (\ref{eq:item41}) around $x=+\infty$, then we interchanged summation and integration.
In the latter formula we used the fact that
\begin{multline*}
\lim_{x\to+\infty} \left[x^{\left(1-2n\right)\frac{1+\alpha}{2\alpha}} \int_{x^{-\frac{1+\alpha}{2\alpha}}}^1 t^{-2 n} \left(1-t^{2 \alpha }\right)^{\frac12-n}dt\right] \\
= \frac{1}{2n-1}\lim_{x\to+\infty} \,_2F_1\left(\frac{2n-1}{2},\frac{1-2n}{2\alpha};1+\frac{1-2n}{2\alpha};x^{-1-\alpha}\right) = \frac{1}{2n-1} \;,\quad \forall n\in\bb{N}\setminus\{0\}\;,
\end{multline*}
and 
\begin{equation*}
\sum_{n\geq 1} \frac{\left(-\frac{1}{2}\right)_n }{\left(1-2n\right)n!} = \frac12\left(\pi-2\right)\;.
\end{equation*}
Using (\ref{eq:item43}) and the definition (\ref{eq:tauJ}), we find
\begin{equation*}
\lim_{\xi\to+\infty}\left(\fun(\xi)-(A\xi)^{\frac{1+\alpha}{2\alpha}}\right) = -\left(1+\alpha\right)\;,
\end{equation*}
which proves the sub-leading asymptotic behavior in (\ref{eq:fasy}). Finally, we have
\begin{equation*}
\lim_{x\to+\infty} \left[x^{\frac{1+\alpha}{2\alpha}}\left(S(x)-Kx^{\frac{1+\alpha}{2\alpha}}+\frac12\right)\right]=-\frac{1}{2\pi}\int_0^1\frac{dt}{t^2\sqrt{1-t^{2\alpha}}}=\frac{1}{2 \sqrt{\pi }}\frac{\Gamma \left(1-\frac{1}{2 \alpha }\right)}{\Gamma \left(\frac{\alpha -1}{2 \alpha }\right)}\;,
\end{equation*}
which yields the thesis.

We omit the detailed proof of (\ref{eq:Dnfasy}), since it is obtained by repeating the same steps of
(\ref{eq:fasy}).\footnote{Moreover, in Lemma \ref{lem:Knbounded} below, we will show how the case $n \geq1$ follows the case $n=0$,
via the study of a family of integral operators that we will introduce in the next section.}
% Applying $n$ times the Euler operator to both sides of (\ref{eq:inhomWH}) we get
% \begin{equation} \label{eq:item44} \tag{j}
%  D_\xi^n \fun(\xi)= \bb{K}_1^{(n)}[\fun](\xi)\;,\quad \forall n \geq 1\;,
% \end{equation}
% with $\bb{K}_1^{(n)}:=\bb{K}_{[1,+\infty)}^{(n)}$ and $\bb{K}_J^{(n)}$ the operators defined by formula
% (\ref{eq:linoperatordef}). After
% Lemma \ref{lem:Knbounded}, we know that these operators are bounded on $L^{\infty}_{\frac{\alpha+1}{2\alpha}}([1,+\infty))$.
% Therefore $\fun$ restricted to $[1,+\infty)$ is smooth and $D_\xi^n \fun(\xi) = O(\xi^{\frac{1+\alpha}{2\alpha}})$.
% Formula (\ref{eq:item44}) together with (\ref{eq:fasy}) yield immediately (\ref{eq:Dnfasy}), due to (\ref{eq:limitsrespectn}).
\\
(5) From the definition (\ref{eq:tauJ}), we have
\begin{equation} \label{eq:item51} \tag{j}
\lim_{\xi\to1^+}\fun'(\xi)=2\left(1+\alpha\right)\sym\lim_{x\to \sym^+} S'(x) \;,
\end{equation}
thus we are led to compute $\lim_{x\to \sym^+} S'(x)$. To evaluate the above mentioned limit, it is
convenient to rewrite (\ref{eq:item35}) as a contour integral as follows
\begin{equation*}
    S'(x)=\frac{1}{4\pi}\lim_{\delta\to 0}\left[\int_{-\gamma_\delta(u_1)}+\int_{\gamma_\delta(u_2)}+\int_{u_1+i\delta}^{u_2+i\delta}+\int_{u_2-i\delta}^{u_1-i\delta}\right]\left(\frac{u}{\sqrt{V(u;x)}}\,du\right)\;,
\end{equation*}
where $\gamma_{\delta}(t)=\left\{u\in\C\,|\,u=t+\delta e^{i\theta}\,,\,\theta\in\left[\frac{\pi}{2},\frac{3\pi}{2}\right]\right\}$.

In fact, after Lemma \ref{lem:u12prop}, we know that $u_1(\sym)=u_2(\sym)=u^*:=\alpha^{-\frac{1}{2\left(1+\alpha\right)}}$. Then, 
\begin{multline} \label{eq:item52} \tag{k}
    \lim_{x\to \sym^+}S'(x)=\frac{1}{4\pi}\lim_{\delta\to0^+}\int_{C_\delta(u^*)}\frac{u}{\sqrt{V(u;\sym)}}\,du=\\
    \frac{i}{2}\res\left(\frac{u}{\sqrt{V(u;\sym)}},u=u^*\right) = \frac{\alpha ^{-\frac{1}{\alpha +1}}}{2 \sqrt{2\alpha +2}} \;,
\end{multline}
where $C_\delta(t)=\left\{u\in\C\,|\,u=t+\delta e^{i\theta}\,,\,\theta\in[0,2\pi]\right\}$, and we used the fact that
\begin{equation*}
   \lim_{u\to u^*} \frac{V(u;\sym)}{\left(u-u^*\right)^2}= -2 \alpha ^{\frac{1}{\alpha +1}} (\alpha +1) \;.
\end{equation*}
Finally, plugging (\ref{eq:item52}) in (\ref{eq:item51}) we obtain the thesis.\\
(6) From the definition (\ref{eq:tauJ}), we have
\begin{equation*}
\underset{\xi\in[1,+\infty)}{\inf}\left(\xi^{-\frac{1+\alpha}{2\alpha}}D_\xi \fun(\xi)\right)=\underset{x\in[\sym,+\infty)}{\inf}\left(2\left(1+\alpha\right)\sym^{\frac{1+\alpha}{2\alpha}}x^{1-\frac{1+\alpha}{2\alpha}}S'(x)\right)\;,
\end{equation*}
thus we are led to prove that the right-hand-side is positive. %$\underset{x\in[\sym,+\infty)}{\inf}\left(x^{1-\frac{1+\alpha}{2\alpha}}S'(x)\right)>0$. 
From (\ref{eq:item35}), we obtain the following estimate
\begin{multline*}
S'(x)=\frac{1}{2\pi}\int_{u_1(x)}^{u_2(x)}\frac{u}{\sqrt{V(u;x)}}\,du \geq\frac{1}{2\pi}\frac{1}{\sqrt{V\left(u_{max}(x);x\right)}}\int_{u_1(x)}^{u_2(x)} u\,du \\
= \frac{1}{4\pi} \frac{u_2^2(x)-u_1^2(x)}{\sqrt{\left(\frac{x}{\sym}\right)^{\frac{1+\alpha}{\alpha}}-1}}:=f(x) \;,
\end{multline*}
where $u_{max}(x)=\left(\frac{x}{1+\alpha}\right)^{\frac{1}{2\alpha}}$ denotes the position of the maximum of $V(u;x)$ and $u_{max}(x)\in\left(u_1(x),u_2(x)\right)$ for all $x>\sym$. Using (\ref{eq:u12asy}) and (\ref{eq:u12lim}) we find
\begin{align*}
&2\left(1+\alpha\right)\sym^{\frac{1+\alpha}{2\alpha}}\lim_{x\to+\infty}\left(x^{1-\frac{1+\alpha}{2\alpha}}f(x)\right)= %\frac{(\alpha +1)^{\frac{\alpha +1}{2 \alpha }}}{4 \pi\sqrt{\alpha }}
\frac{(\alpha +1)^{\frac{1}{\alpha }+2}}{2 \pi  \alpha }>1\;,\\ 
&2\left(1+\alpha\right)\sym^{\frac{1+\alpha}{2\alpha}}\lim_{x\to \sym^+}\left(x^{1-\frac{1+\alpha}{2\alpha}}f(x)\right)= %\frac{(\alpha +1)^{-\frac{1}{2\alpha}}}{\pi\sqrt{2\alpha}}
\frac{\sqrt{2}\left(\alpha +1\right)^{3/2}}{\pi  \alpha }\geq\frac{3}{\pi}\sqrt{\frac32}\;,
\end{align*}
and, since $u_2(x)>u_1(x)$ for all $x>\sym$, we have
\begin{equation*}
\underset{x\in[\sym,+\infty)}{\inf}\left(2\left(1+\alpha\right)\sym^{\frac{1+\alpha}{2\alpha}}x^{1-\frac{1+\alpha}{2\alpha}}S'(x)\right)\geq\underset{x\in[\sym,+\infty)}{\inf}\left(2\left(1+\alpha\right)\sym^{\frac{1+\alpha}{2\alpha}}x^{1-\frac{1+\alpha}{2\alpha}}f(x)\right)>0\;,
\end{equation*}
thus proving the thesis.
\end{proof}

\begin{theorem}\label{prop:fproperty}
The linearised IBA equation (\ref{eq:IBAStL}) admits a unique solution
$l:=l(\cdot;\omega,p):[\omega,+\infty) \to \R$ normalised such that
$\lim_{x\to+\infty}x^{-\frac{1+\alpha}{2\alpha}}l(x;\omega,p)=1$.

It has the following explicit representation
\begin{align} \nonumber
&l(x;\omega,p)= \bar{l}\left(\frac{x}{\omega};\omega,p\right)\;,\\
    &\bar{l}(\xi;\omega,p)=p \fun(\xi)-\frac{2\alpha}{1+\alpha}\left[p-
    \left(\frac{\omega}{A}\right)^{\frac{1+\alpha}{2\alpha}}\right]D_\xi \fun(\xi)\;,
    \label{eq:lomegaground2} 
\end{align}
Moreover,
 \begin{equation}\label{eq:barl1}
    \bar{l}(1;\omega,p)=\sqrt{2\left(1+\alpha\right)}\left[\left(\frac{\omega}{A}\right)^{\frac{1+\alpha}{2\alpha}}-p\right]\;.
\end{equation}
In the above formulae, $\fun$ is the restriction to $[1,+\infty)$ of the function defined by formula (\ref{eq:tauJ}), and
$A$ is as per (\ref{eq:Adef}).
\end{theorem}
\begin{proof}
Uniqueness was already proven in the cited physics literature, see e.g. \cite[Section 3]{baluzaII}.

After Proposition \ref{prop:tauJ}, the function $l(x;\omega,p)$ satisfies the linearised equation and the normalisation.
Equation (\ref{eq:barl1}) follows from (\ref{eq:fp1}) and (\ref{eq:lomegaground2}).
\end{proof}

For the rest of our paper it is crucial that the solution of the linearised IBA equation is strictly monotone, in the
strong sense of the definition below.
\begin{definition}
 Let $f:[1,+\infty) \to \R^+$ be a differentiable function with the asymptotics $f(\xi)=\xi^{\frac{1+\alpha}{2\alpha}}+
 o\left(\xi^{\frac{1+\alpha}{2\alpha}}\right)$ as $\xi \to +\infty$.
 We say that $f$ is strongly strictly monotone if
  $$\underset{\xi\in[1,+\infty)}{\inf}\left(\xi^{-\frac{1+\alpha}{2\alpha}}D_\xi f(\xi)\right)>0\;.$$
\end{definition}
Here we show that
if $\omega$ is chosen in a way such that
$ \bar{l}(1;\omega,p)$ is bounded as $p \to \infty$, then $\bar{l}(\xi;\omega,p)$ is strongly strictly monotone.
\begin{lemma}
\label{lem:linearisedground}
Fix $C>0$ and $H\geq0$. Assume that  $p\geq 
\frac{H+C}{\sqrt{2\left(1+\alpha\right)}}$.\\
(1) The condition $|\bar{l}(1;\omega,p)+H|\leq C$ holds if and only if
\begin{equation}
\omega\in\Omega_{H,C}:=\left[
A\left(p-\frac{H+C}{\sqrt{2\left(1+\alpha\right)}}\right)^{\frac{2\alpha}{1+\alpha}},
A\left(p-\frac{H-C}{\sqrt{2\left(1+\alpha\right)}}\right)^{\frac{2\alpha}{1+\alpha}} \right]\;,
    \label{eq:omegaground}
\end{equation}
where $A$ is as per (\ref{eq:Adef}).\\
(2) There exists a $p_C> 0$ such that
\begin{enumerate}[(i)]
\item for all $(\omega,p)\in\Omega_{H,C}\times[p_C,+\infty)$, the solution
$\bar{l}(\xi;\omega,p)$ is strongly strictly monotone. More precisely,
\begin{equation}\label{eq:barlpunif}
    \underset{\xi\in[1,+\infty)}{\inf} \left(\xi^{-\frac{1+\alpha}{2\alpha}}D_\xi\bar{l}(\xi;\omega,p)\right) \gtrsim p\;;
\end{equation}
\item if $\delta \in [0,M]$ for some arbitrary but fixed $M \geq 0$, then
\begin{equation}\label{eq:1+deltap}
    \bar{l}\left(1+\delta p^{-1};\omega,p\right) -
    \bar{l}(1;\omega,p)=  \frac{\delta\left(1+\alpha\right)^{\frac32}}{\sqrt{2}\alpha}
    + O(p^{-1}) \mbox{ as } p\to+\infty \;,
\end{equation}  
uniformly with respect to $(\delta,\omega,p) \in [0,M] \times  \Omega_{H,C}\times [p_C,+\infty)$.
\end{enumerate}
(3) Let $\omega_H$ be such that $\bar{l}(1;\omega_H,p)=-H$, i.e.
\begin{equation}
\label{eq:omegagroundlim}
    \omega_H:=A\left(p-\frac{H}{\sqrt{2\left(1+\alpha\right)}}\right)^{\frac{2\alpha}{1+\alpha}}.
\end{equation}
For any $k\geq-H$, the equation $\bar{l}(\xi_k;\omega_H,p)=k+\frac12$ admits a unique solution
$\xi_k$ which has the following asymptotics
\begin{equation}
\xi_k(p)=1+\frac{\sqrt{2}\alpha}{\left(1+\alpha\right)^{\frac32}}\left(k+\frac12+H\right)p^{-1}+O\left(p^{-2}\right) \mbox{ as } p\to+\infty \;.
\label{eq:xkexpground}
\end{equation}
\end{lemma}
\begin{proof}
(1) Follows immediately from (\ref{eq:barl1}) in Theorem \ref{prop:fproperty}.\\
(2i) From equations (\ref{eq:lomegaground2},\ref{eq:barl1}) of Theorem \ref{prop:fproperty} it follows that
\begin{equation*}
 \bar{l}(\xi;\omega,p)= p \fun(\xi)+ \frac{\sqrt{2}\alpha}{\left(1+\alpha\right)^{\frac32}}\bar{l}(1;\omega,p) D_\xi \fun(\xi)\;.
\end{equation*}
Since $\bar{l}(1;\omega,p)$ is bounded if $\omega\in\Omega_{H,C}$ and $D_{\xi}^n \fun(\xi)= O\left(\xi^{\frac{1+\alpha}{2\alpha}}\right)$ as $\xi\to+\infty$ for any $n \geq 0$, then
$D_{\xi}  \bar{l}(\xi;\omega,p)$ is dominated by $pD_{\xi}\fun(\xi)$ if $p$ is large enough. Therefore,
(\ref{eq:barlpunif}) follows from (\ref{eq:infderf}).\\
(2ii) Equation (\ref{eq:1+deltap}) follows from (\ref{eq:fp1}) as proven in Proposition \ref{prop:tauJ} (4).\\
(3) Equation (\ref{eq:xkexpground}) follows immediately from (\ref{eq:1+deltap}).

\end{proof}

\begin{remark}
Solving the linearised IBA equation via the Wiener-Hopf method, see \cite{baluzaII,fioravanti03}, \footnote{A thorough discussion
 of this method can also be found in the first version of this paper on the arXiv.} one naturally
 expresses the solution $\fun$ as an inverse Mellin transform and one deduces the following
 alternative formula
 \begin{align}\nonumber
&\fun(\xi)=\frac{1}{2\pi i}\int_{\delta-i\infty}^{\delta+i\infty} 
\frac{\alpha^{\frac{\alpha  s}{1+\alpha}}}{ 2\sqrt{ \pi} \left(1+\alpha\right)^{s-1} }
\frac{\Gamma \left(-\frac12-\frac{\alpha s}{1+\alpha}\right) \Gamma \left(1-\frac{s}{1+\alpha}\right)}{ s^2  \, \Gamma (-s)}
\,\xi^{-s}ds \;,
\label{eq:fhatrem}
\end{align}
where $\delta$ is an arbitrary real number less than $-\frac{1+\alpha}{2\alpha}$.

That the above formula and the WKB integral (\ref{eq:tauJ}) provides the same function, it is one of the magic of the ODE/IM
correspondence. We will explain
the origin of the formula (\ref{eq:tauJ}) in Section \ref{sec:ODEIMcorr} below.
\end{remark}

\subsection{The perturbed IBA equation and the strategy of the proof of the main theorem}
In this paper we study the standard IBA equation (\ref{eq:IBASt}) as a perturbation of its linearisation, equation (\ref{eq:IBAStL}).

Let then $l(x;\omega,p)$ be the unique normalised solution of (\ref{eq:IBAStL}) studied in Theorem \ref{prop:fproperty}.
For convenience, we define the rescaled function
\begin{equation}
    \bar{l}:[1,+\infty)\to\R^+: \xi\mapsto\bar{l}(\xi):=l(\omega \xi;\omega,p)\;,
\end{equation}
and we introduce new unknowns $\la:[1,+\infty)\to\R$ and $\underline{\mu}=(\mu_1,\dots,\mu_H)$
\begin{align}
& \la(\xi)=z(\omega\xi)-\bar{l}(\xi)\;,\quad \forall \xi\in[1,+\infty)\;, \\
& \mu_k=\frac{h_k}{\omega}-1\;,\quad \forall k\in\{1,\dots,H\}\;,
\end{align}
With respect to these, the standard IBA equation (\ref{eq:IBAStGen}) reads
 \begin{subequations}
\label{eq:IBASteGen}
 \begin{alignat}{2}
     &\la(\xi)= \bb{K}_1\left[\la\right](\xi) - \bb{K}_1\left[\afr{\bar{l}+\la}\right](\xi) -
  \sum_{k=1}^{H} \left[ F_\alpha\left(\frac{\xi}{1+\mu_k}\right)- F_\alpha(\xi)\right],
     \label{eq:IBASte}\\
     &\la(1+\mu_k)  =  \sigma(k)+\frac12 -\bar{l}(1+\mu_k)\;,\quad \forall k \in\{1,\dots,H\}\;, \label{eq:IBAStmu}
     \end{alignat}
 \end{subequations}
 and the inequalities (\ref{eq:IBAStka},\ref{eq:IBAStSt}) read
 \begin{align} \label{eq:IBAekappa}
 &  -H-\frac{1}{2}-\bar{l}(1)<\la(1)\leq -H+\frac{1}{2}-\bar{l}(1) \;,\\
 &  \bar{l}'(\xi)+ \la'(\xi)>0 \;,\quad \forall \xi\in[1,+\infty)\;.
 \label{eq:IBAeine}
 \end{align}
We name (\ref{eq:IBASteGen}) the perturbed IBA equation.
We notice that in case $\bar{l}$ is monotone then (\ref{eq:IBAStmu}) can equivalently (and more conveniently) be rewritten as follows
\begin{equation}\label{eq:IBAStmu2}
 \mu_k= \left(\frac{\bar{l}(1+\mu_k)-\bar{l}(1)}{\mu_k} \right)^{-1} \left( 
  \sigma(k)+\frac12 -\bar{l}(1) - \la(1+\mu_k) 
 \right)\;,\quad k\in\{1,\dots, H\}\;.
\end{equation}
In fact, writing
 $\bar{l}(1+\mu_k)=\left(\frac{\bar{l}(1+\mu_k)-\bar{l}(1)}{\mu_k}\right)\mu_k+
 \bar{l}(1),$ and using the monotonicity of $\bar{l}(\xi)$ to invert the difference quotient,
 (\ref{eq:IBAStmu2}) and (\ref{eq:IBAStmu}) are shown to be equivalent. \\

After Proposition \ref{prop:strongz}, studying purely real and normalised solutions of the BAE is the same
as studying equivalence classes of strictly monotone solutions of the standard IBA equation,
provided that the set of hole-numbers and the parameters
$H,\sigma$ of the standard IBA are related by formula (\ref{eq:bbHsIBA}). 
Therefore, we prove the main Theorem by showing that for any $H$ and $\sigma$ the
standard IBA equation admits a unique solution, up to equivalence, if $p$ is large enough.
Our strategy is based on the separate analysis of the linearised IBA and of the perturbed IBA and follows the
steps briefly illustrated below.
\begin{itemize}
 \item In Section \ref{sec:convoper}, we show that the inequality (\ref{eq:IBAekappa}),
 implies that
 \begin{equation}\label{eq:Dxinst} \tag{S1}
 \|\la\|_{\infty} \leq \left(H+\frac{1}{4}\right)(\alpha-1)\;,\quad \| D^n \la \|_{\infty}\lesssim_n 1\;,\quad \forall n \geq 1\;.
 \end{equation}
 where $\|\cdot\|_\infty$ is the norm of $L^\infty\left([1,+\infty)\right)$. Combining this estimate with the value of $\bar{l}(1)$ that we computed in Theorem \ref{prop:fproperty},
 we deduce that the constraint (\ref{eq:IBAekappa}) induces the following a-priori inequality on the end-point $\omega$
 \begin{align}
 &\left|\frac{\omega}{\omega_H}-1\right| \lesssim p^{-1}\;,\notag\\
 &\omega_H= (1+\alpha) \left(\frac{p}{\sqrt{\pi\alpha}}
 \frac{\Gamma \left(\frac{1}{2 \alpha }\right)}
 {\Gamma \left(\frac{1+\alpha}{2 \alpha }\right)}\right)^{\frac{2 \alpha }{\alpha +1}}
\left[ 1- 
\left(\frac{\sqrt{2}\alpha}{(1+\alpha)^{\frac32}} H p^{-1} \right)\right] \;.\label{eq:omegaHst} \tag{S2}
 \end{align}
\item In Section \ref{sec:osc}, we estimate the oscillatory term $\bb{K}_{1}[\afr{\bar{l}+\e}]$,
where $\bar{l}$ is the rescaled solution of the linearised IBA equation with $\omega$ satisfying the bound (\ref{eq:omegaHst}), and
$\e$ any bounded differentiable function with bounded derivative. We prove that 
\begin{equation}\label{eq:bbKost} \tag{S3}
\left|\bb{K}_{1}[\afr{\bar{l}+\e}](x)\right| \lesssim p^{-1} x^{-\frac{1+\alpha}{2\alpha}} \| \e\|_{\infty}\;.
 \end{equation}
 \item In Section \ref{sec:mainthm}, using estimate (\ref{eq:bbKost}), we
 greatly strengthen estimates (\ref{eq:Dxinst},\ref{eq:omegaHst}) to obtain
 \begin{equation}\label{eq:Dxinlast} \tag{S4}
\| D^n \la \|_{\infty}\lesssim_n p^{-1}\;,\quad 
 \left|\frac{\omega}{\omega_H}-1\right| \lesssim p^{-2}\;.
 \end{equation}
 \item Finally, we show that the perturbed IBA equation (\ref{eq:IBASteGen})
admits a unique strictly monotone solution if $\omega$ and $\lambda$ satisfy the constraints (\ref{eq:Dxinlast}).
Moreover, the solutions that we obtain for different values of $\omega$ are equivalent.
\end{itemize}

\subsubsection{Comparison with the BKP-DDV-NLIE equation}
We briefly compare here the standard IBA equation with the nonlinear equation derived in the physics literature, which is known as Batchelor-Klumper-Pearce (BKP) or
Destri-De Vega (DDV) or simply Nonlinear Integral Equation (NLIE).
The comparison is based on the following identity (proven below)
\begin{equation}\label{eq:zetalim}
 \afr{z(x)}=\lim_{\e \to 0^+}\frac{1}{\pi} \im \log\left( 1+e^{2\pi i\,z(x+i \e)}\right)\;,
\end{equation}
where $\afr{x}=x-\ceili{x-\frac12}$ as per (\ref{eq:afr}), and the branch of the logarithm is chosen so that $\frac{1}{\pi}\im \log \left( e^{-2i \pi p} \cos\left(-2 \pi p\right) \right)
\in \left(-\frac{1}{2},\frac{1}{2}\right]$. From the above identity,
it follows that the standard IBA equation (\ref{eq:IBASt}) can be written as
\begin{align}\nonumber
& z(x)= -2p +\bb{K}_\omega[z](x)
+H \, F_{\alpha}\left(\frac{x}{\omega}\right)-
  \sum_{k=1}^H F_{\alpha}\left(\frac{x}{h_k}\right) 
\\ \label{eq:DDV}
& - \lim_{\e \to 0^+} \frac{1}{\pi} \im \bb{K}_\omega\left[\log\left( 1+e^{2\pi i\,z(\cdot +i \e)}\right)\right](x)  \;.
\end{align}
Fix $H=0$. Starting from the above equation and making some standard manipulations explained for example in
\cite{doreyreview}, one obtains the Destri-De Vega equations for
the ground-state of the Quantum KdV model.

We prove here formula (\ref{eq:zetalim}), assuming for simplicity that $-2p+\frac12 \notin \bb{Z}$.\\
(i) We start with the identity $ \log\left( 1+e^{2\pi i\,z(x)}\right)= \log \left(e^{i \pi z(x)} \cos \left(\pi z(x)\right) \right)
 +\log 2$ and let
 $$\zeta(x)=\lim_{\e \to 0^+}\frac{1}{\pi} \im \log\left( e^{i \pi z(x+i \e)} \cos \left(\pi z(x+i\e)\right)\right),$$
 where the branch of the logarithm is chosen so that
 $$ \zeta(0) =  \frac{1}{\pi}\im \log \left( e^{-2i \pi p} \cos\left(-2 \pi p\right) \right) \in  \left(-\frac{1}{2},\frac{1}{2}\right]
 \quad\Longrightarrow\quad \zeta(0)=\afr{z(0)}.$$
 It is straightforward to see that $\zeta$ has a point-wise limit on $\bb{R}^+ \setminus \lbrace x_k\rbrace_{k \in \bb{Z}_p}$ and
 that the convergence is uniform on any compactsubset of $\bb{R}^+ \setminus \lbrace x_k\rbrace_{k \in \bb{Z}_p}$.\\
 (ii) Now we compare $\zeta$ with $\afr{z}$. By definition of $\zeta$, we have that
 $\zeta(x) \equiv z(x) +\sigma \mod 2$, where $\sigma= 0$ if $\mbox{ sign} \cos (\pi z(x))>0$ and $\sigma=-1$ if
 $\mbox{ sign} \cos (\pi z(x))<0$. Since $z$ and $\zeta$ are continuous in each interval of the form 
 $(x_{k},x_{k+1})$, then for every $k \in \bb{Z}_p$ there exists a $c_k \in \bb{Z}$ such that
 $$
 \zeta(x) - \afr{z(x)} =c_k\;,\quad \forall x \in(x_{k},x_{k+1})\;.
 $$
 By construction $\zeta(0)=\afr{z(0)}$, therefore equation (\ref{eq:zetalim}) holds
 if and only if $\lim_{x \to x_k^-} \zeta(x) -\lim_{x \to x_k^+} \zeta(x)=1$ for all $k \in \bb{Z}_p$.
 Now, since $z'(x)>0$ by hypothesis, then the meromorphic function
 $\frac{d}{dx}\log\left( 1+e^{2\pi i\,z(x)}\right)$ has a simple pole at $x=x_k$. Therefore, using the residue formula we have that
 $$
 \lim_{x \to x_k^-} \zeta(x) -\lim_{x \to x_k^+} \zeta(x)=\frac{1}{2} 2 \pi i \res\left(\zeta'(x),x_k\right)=1\;.
 $$

\section{The convolution operator}
\label{sec:convoper}
In this section we study many properties of the convolution operators
\begin{equation}\label{eq:linoperatordef}
 \bb{K}^{(n)}_J[f](x):= \int_J D_x^n \, K_{\alpha}\left(\frac{x}{y}\right)\frac{f(y)}{y}dy := 
 \int_J  K_{\alpha}\left(\frac{x}{y};n\right)\frac{f(y)}{y}dy\;,
\end{equation}
where $D_x=x \frac{d}{dx}$, $J$ is either an admissible set or $\bb{R}^+$ (which is not admissible), and
$K_{\alpha}(\cdot,n)$ is obtained from the kernel $K_{\alpha}$ (\ref{eq:Kfunzal}), by repeated action of the operator $D_x$:
\begin{equation}\label{eq:Knal}
K_{\alpha}(x;n)= D_x^n  K_{\alpha}(x)\;,\quad
 K_{\alpha}(x)=\frac{\sin\left(\frac{2 \pi}{1+\alpha}\right)}{\pi} \frac{x}{1+x^2- 2 x \cos\left(\frac{2 \pi}{1+\alpha}\right)}\;.
\end{equation}
%The linear spaces on which $\bb{K}^{(n)}_J$ are weighted $L^{\infty}$ spaces, which are defined as follows.
Let us define the following weighted $L^{\infty}$ spaces.
 \begin{definition}
  For every admissible set $I$ and every $s \in \bb{R}$, we let
  $L^{\infty}_{s}(I)$ be the space of locally bounded functions on $I$ with finite weighted supremum norm
\begin{equation}\label{eq:weightedsup}
 \| f \|_{\infty,s}= \sup_{x \in I} \left| x^{-s} f(x) \right|\;.
\end{equation}
We let  $L^{\infty}_{s}(\bb{R}^+)$ be the space of locally bounded functions on $\R^+$ with finite norm
\begin{equation}\label{eq:weightedsupR}
 \| f \|^+_{\infty,s}=\sup_{x \in \R^+}\left| \left(w(x)\right)^{-s} f(x) \right|\;,\quad w(x)=  \chi_{[0,1)}(x)+ x \chi_{[1,+\infty)}(x) \;.
\end{equation}
 \end{definition}
 Before we enter into details of the properties of the operators $\bb{K}_{\alpha}^{(n)}$, we collect in the Lemma
 below some properties of their kernels, that are needed in the sequel of the paper.
 \begin{lemma}\label{lem:maximumkernel}
(1) The functions $K_{\alpha}(x;n)$ are rational functions, smooth on $\bb{R}^+$, vanishing linearly
at $x=0$ and $x=\infty$.\\
(2) The functions $K_{\alpha}(x;n)$ have the symmetry
\begin{equation}\label{eq:Kansymmetry}
 K_{\alpha}\left(\frac{x}{y};n\right)=(-1)^nK_{\alpha}\left(\frac{y}{x};n\right)\;,\quad\forall x,y \in \bb{R}^+\;.
\end{equation}
(3) Fix $\alpha>1$. The following estimates hold
 \begin{align} \label{eq:maxkernel2}
    & \sup_{x \in \bb{R}^+} \left| K_{\alpha}(x;n) \right|\lesssim_n 1\;, \\
  & \label{eq:maximumkernel}
 \sup_{y \in \bb{R}^+} \left| \frac{1}{y}K_{\alpha}\left(\frac{x}{y};n\right) \right|\lesssim_n    x^{-1} \;,\\ \label{eq:maxkernel3}
  &  \sup_{y \in \bb{R}^+} \left|\frac{d}{dy}\left[\frac{1}{y}K_{\alpha}\left(\frac{x}{y};n\right)\right]\right| \lesssim_n x^{-2} \;, \\ \label{eq:maxkernel4}
   &  \sup_{x \in \bb{R}^+} \left|\frac{d}{dy}\left[\frac{1}{y}K_{\alpha}\left(\frac{x}{y};n\right)\right]\right|\lesssim_n y^{-2}  \;.
 \end{align}
 \begin{proof}
 (1) By direct inspection, the property holds for $n=0$. Since the action of the Euler operator $D_x$ preserves
 the order of vanishing at $x=0$ and $x=+\infty$, then (1) holds for every $n\in\bb{N}$.\\
 (2) For $n=0$, (\ref{eq:Kansymmetry}) is the same as (\ref{eq:Ksymm}). Acting with the Euler operator
 on both sides of (\ref{eq:Ksymm}), we obtain the thesis.\\
 (3) The estimate (\ref{eq:maxkernel2}) follows immediately from (1).
 Since $K_\alpha(x;n)$ vanishes linearly at $x=0$, we have
 \begin{equation*}
 \sup_{y \in \bb{R}^+} \left| \frac{x}{y}K_{\alpha}\left(\frac{x}{y};n\right) \right|=\sup_{z \in \bb{R}^+}\left|\frac{1}{z}K_\alpha(z;n)\right|\lesssim_n 1 \;,
 \end{equation*}
 from which (\ref{eq:maximumkernel}) follows. 
 It can be easily shown that
 \begin{equation*}
 K_\alpha(x;n)=\frac{\sin \left(\frac{2 \pi }{\alpha +1}\right)}{\pi }x+\frac{2^n \sin \left(\frac{4 \pi }{\alpha +1}\right)}{\pi }x^2+O(x^3) \mbox{ as } x\to0^+ \;,
 \end{equation*}
 whence $K_\alpha(x;n+1)-K_\alpha(x;n)=\frac{2^n \sin \left(\frac{4 \pi }{\alpha +1}\right)}{\pi }x^2+O(x^3)$ as $x\to0^+$. Then,
 \begin{multline*}
 \sup_{y \in \bb{R}^+} \left|x\frac{\partial}{\partial y}\left[\frac{x}{y}K_{\alpha}\left(\frac{x}{y};n\right)\right]\right| = \sup_{z \in \bb{R}^+} \left|\frac{d}{dz}\left[\frac{1}{z}K_{\alpha}(z;n)\right]\right|\\
 =\sup_{z \in \bb{R}^+} \left|\frac{1}{z^2}\left[K_{\alpha}(z;n+1)-K_\alpha(z;n)\right]\right| \lesssim_n 1 \;,
 \end{multline*}
 from which (\ref{eq:maxkernel3}) follows.
 A simple computation gives
 \begin{equation*}
 \sup_{x \in \bb{R}^+} \left|y^2\frac{\partial}{\partial y}\left[\frac{1}{y}K_{\alpha}\left(\frac{x}{y};n\right)\right]\right| = \sup_{z \in \bb{R}^+} \left|K_{\alpha}(z;n+1)-K_{\alpha}(z;n)\right|\lesssim_n 1\;,
 \end{equation*}
 from which (\ref{eq:maxkernel4}) follows.
 \end{proof}

\end{lemma}

\subsection*{The operator $\bb{K}_J$.}
We analyse first the operator $\bb{K}_J:=\bb{K}_J^{(0)}$
which is the most important to our analysis. The following family of integrals will be often used:

For every $r\geq0$, $\alpha >1$ and $s \in \bb{C}$ with $ \Re s\in(-\infty,1)$, we define the incomplete integral
\begin{align}\label{eq:incompletedef}
 \Phi(\alpha,s;r)=\frac{\sin\left(\frac{2 \pi}{1+\alpha}\right)}{\pi}\int_r^{+\infty} \frac{t^{s}}{1+t^2-2t\cos\left(\frac{2 \pi}{1+\alpha}\right) } \,dt\;,
 \end{align}
 and, if $ \Re s\in (-1,1)$, we also define the complete integral
 \begin{align}\label{eq:completedef}
 \Phi(\alpha,s)=\frac{\sin\left(\frac{2 \pi}{1+\alpha}\right)}{\pi}\int_0^{+\infty} \frac{t^{s}}{1+t^2-2t\cos\left(\frac{2 \pi}{1+\alpha}\right)} \, dt\;.
 \end{align}
 The latter coincides with the Mellin transform of the integral kernel $K_{\alpha}(x)$.

The above integrals can be computed in closed form. In fact, we have the following Lemma.
 \begin{lemma}
We have that
\begin{equation}\label{eq:incompleteintegral}
  \Phi(\alpha,s;r)=  \frac{r^{s}}{2\pi i s} \left[ {}_2F_1\left(1,-s;-s+1;\frac{\kappa - i \sigma}{r}\right) -
  {}_2F_1\left(1,-s;-s+1;\frac{\kappa + i \sigma}{r}\right) \right]\;,
\end{equation}
where $\sigma=\sin\left(\frac{2 \pi}{1+\alpha}\right)$ and $\kappa=\cos\left(\frac{2 \pi}{1+\alpha}\right)$. In particular, when $s$ is real
\begin{equation}\label{eq:incompleteintegralreal}
  \Phi(\alpha,s;r)=  \frac{r^{s}}{\pi s} \Im\left[ {}_2F_1\left(1,-s;-s+1;\frac{\kappa - i \sigma}{r}\right)\right]\;.
\end{equation}
Moreover,
\begin{equation}\label{eq:completeintegral}
  \Phi(\alpha,s)=\lim_{r \to 0^+} \Phi(\alpha,s;r)= 
  \frac{\sin\left(\frac{\pi(\alpha-1)}{1+\alpha}s\right)}{\sin\left(\pi s\right)} \;
  \mbox{ and }\;\Phi(\alpha,0)=\frac{\alpha-1}{\alpha+1}\;.
\end{equation}

  \begin{proof}
  Let $\sigma=\sin\left(\frac{2 \pi}{1+\alpha}\right)$ and $\kappa=\cos\left(\frac{2 \pi}{1+\alpha}\right)$. Then,
  \begin{multline*}
\Phi(\alpha,s;r)= \frac{\sigma}{\pi}\int_r^{+\infty} \frac{t^{s}}{1+t^2-2 \kappa t}\,dt =
\frac{1}{2 \pi i}\int_r^{+\infty} \left( \frac{t^{s}}{t-\kappa-i \sigma} - \frac{t^{s}}{t-\kappa+i \sigma}\right) dt
\\ 
= \frac{r^s}{2 \pi i s}  \left[ {}_2F_1\left(1,-s;-s+1;\frac{\kappa - i \sigma}{r}\right) -
    {}_2F_1\left(1,-s;-s+1;\frac{\kappa + i \sigma}{r}\right) \right]\;,
  \end{multline*}
  where $ _2F_1\left(a,b;c;z\right)$ is the hypergeometric function with domain $|\arg{(-z)}|<\pi$. In order to compute
  $\Phi(\alpha,s)$ we need to evaluate the hypergeometric function when the last argument tends to infinity in the direction $\kappa-i \sigma= e^{-\frac{2 \pi i}{1+\alpha}}$.
  To this aim we use the following well-known formula \cite[Equation (17), Chapter 2]{bateman1}
\begin{multline*}
 _2F_1\left(a,b;c;-z^{-1}\right)  =  \frac{\Gamma(c)\Gamma(b-a)}{\Gamma(b)\Gamma(c-a)} z^a {}_2F_1\left(a,1-c+a;1-b+a;z\right)  \\
+ \frac{\Gamma(c)\Gamma(a-b)}{\Gamma(a)\Gamma(c-b)} z^b {}_2F_1\left(b,1-c+b;1-a+b;z\right) \;,\quad b-a \notin \bb{Z}\;,
\end{multline*}
to deduce that, for each $s\in\C\setminus\{0\}$ with $\Re s \in(-1,1)$
\begin{equation*}
\Phi(\alpha,s)=\frac{1}{2 \pi i s} \Gamma(-s+1)\Gamma(1+s) \left[ (-\kappa+ i \sigma)^{s} -(-\kappa- i \sigma)^{s}\right] = 
\frac{\sin \left( \frac{\pi (\alpha-1)}{1+\alpha} s\right)}{\sin\left( \pi s\right)}\;.
\end{equation*}
In order to obtain the latter identity we used well-known functional equations for the $\Gamma$ function \cite[Equations (1.6), Chapter 1.3]{bateman1}.
Finally, $\Phi(\alpha,0)$ is obtained by continuity.
  \end{proof}
 \end{lemma}
 
In the following Proposition, we compute the norm of the operator $\bb{K}_J$ on $L^{\infty}_s(J)$ with $s\in(-1,1)$, 
and the norm of the resolvent $\left(\bb{I}-\bb{K}_J\right)^{-1}$ on $L_s^\infty(J)$
with $s\in\left(-\frac{1+\alpha}{2\alpha},\frac{1+\alpha}{2\alpha}\right)$, where $\bb{I}$ is the identity operator.
\begin{proposition}\label{prop:Kzasym}
(1) Let $J=\bb{R}^+$ or an admissible interval, $s \in(-1,1)$ and $f:\bb{R}^+ \to \bb{R}$ be a locally bounded function.  We have that
\begin{equation}\label{eq:limitsrespect}
 \lim_{x \to +\infty}  x^{-s} f(x)= c \in \bb{R}^+ \quad\Longrightarrow\quad \lim_{x \to +\infty} x^{-s}\bb{K}_J[f](x)= \Phi(\alpha,s) c\;,
 \end{equation}
where $\Phi(\alpha,s)$ is as per (\ref{eq:completeintegral}).\\
(2) Let $I$ be an admissible interval, $s\in(-\infty,1)$ and $f \in L^{\infty}_{s}(I)$. Then,
\begin{equation}\label{eq:puntualestimate}
 \left| x^{-s} \bb{K}_I[f](x) \right| \leq \Phi\left(\alpha,s; \frac{\omega}{x}\right) \| f\|_{\infty,s}\;,
\end{equation}
where $\omega=\underset{x\in I}{\inf} \{x\}$ and $\Phi(\alpha,s; r)$ is as per (\ref{eq:incompleteintegralreal}).\\
(3) For all $s \in(-1,1)$, $\bb{K}_I$ is a continuous operator on the space $L^{\infty}_{s}(I)$, and
\begin{equation}\label{eq:normKI}
 \| \bb{K}_I \|_{\infty,s}= \Phi(\alpha,s) = \frac{\sin\left(\pi\frac{\alpha-1}{1+\alpha}s\right)}{\sin\left(\pi s\right)}\;.
\end{equation}
(4) For all $s\in \left(-\frac{\alpha+1}{2\alpha},\frac{\alpha+1}{2\alpha}\right)$,
the operator $\bb{I}-\bb{K}_I$ is invertible on $L^{\infty}_{s}(I)$ and 
\begin{equation}\label{eq:normResolvent}
  \| \left(\bb{I}_I-\bb{K}_I \right)^{-1} \|_{\infty,s} = \frac{1}{1-\Phi(\alpha,s)} =
  \frac{\sin\left(\pi\frac{\alpha-1}{1+\alpha}s\right)}
  {\sin\left(\pi s\right)- \sin\left(\pi\frac{\alpha-1}{1+\alpha}s\right)} \;.
\end{equation}
(5) Let $\omega >0$ and $g \in L^{\infty}_s([\omega,+\infty))$, with
$s \in  \left(-\frac{\alpha+1}{2\alpha},\frac{\alpha+1}{2\alpha}\right)$.
The general solution of the linear integral equation
\begin{equation}\label{eq:fKfg}
 f=\bb{K}_{\omega}[f]+g \;,\quad f \in L^{\infty}_{\frac{\alpha+1}{2\alpha}}\left([\omega,+\infty)\right)\;,
\end{equation}
is given by the formula
\begin{equation}\label{eq:affinesolution}
 f(x)= \left(\bb{I}-\bb{K}_{\omega}\right)^{-1} [g]+ c x\, \fun'(\omega x) \;,\quad c \in \bb{R}.
\end{equation}
In the above equation, $\fun$ is as per (\ref{eq:tauJ}) and  $\left(\bb{I}-\bb{K}_{\omega}\right)^{-1}$ is the resolvent operator on the space
$L^{\infty}_s\left([\omega,+\infty)\right)$.
Hence if $f$ is a solution of equation (\ref{eq:fKfg}), there exists a $c' \in \bb{R}$
such that
\begin{equation}\label{eq:subtracteds}
 f(x) - c'x^{\frac{1+\alpha}{2\alpha}} \in L^{\infty}_s\left([\omega,+\infty)\right)\;.
\end{equation}\\
(6) Let $\omega>0$ and $g:[\omega,+\infty) \to \bb{R}$ be a continuous function such that
\begin{equation*}
 \lim_{x \to + \infty} x^{-s} g(x) = d \in \R\;, \mbox{ for some } s \in  \left(-\frac{\alpha+1}{2\alpha},\frac{\alpha+1}{2\alpha}\right)\;.
\end{equation*}
Let $f \in L^{\infty}_{\frac{\alpha+1}{2\alpha}}\left([\omega,+\infty)\right)$ be a solution of (\ref{eq:fKfg}).
There exists a $c \in \bb{R}$
such that
\begin{equation}\label{eq:subtractedlimit}
 \lim_{x\to +\infty} x^{-s} \left( f(x) - c x^{\frac{1+\alpha}{2\alpha}}\right)= \frac{\sin\left(\pi\frac{\alpha-1}{1+\alpha}s\right)}
  {\sin\left(\pi s\right)- \sin\left(\pi\frac{\alpha-1}{1+\alpha}s\right)} d\;.
\end{equation}

 \begin{proof}
In this proof we use the short-hand notations:
$K:=K_{\alpha}$, $\kappa=\cos\left(\frac{2 \pi}{1+\alpha}\right)$ and $\sigma=\sin\left(\frac{2 \pi}{1+\alpha}\right)$.\\
(1) Since the kernel is bounded, it follows that if $[a,b]\subset{\bb{R}^+}$ then
\begin{equation}\label{eq:KL1}
 f \in L^{1}([a,b]) \quad\Longrightarrow\quad \int_a^b K\left(\frac{t}{x}\right) \frac{f(t)}{t}dt = O(x^{-1}) \mbox{ as } x \to + \infty \;,
\end{equation}
whence
\begin{equation*}
 \bb{K}_J[f](x)  - \bb{K}_{\bb{R}^+}[f](x)=O(x^{-1}) \mbox{ as } x\to+\infty\;.
\end{equation*}
Therefore, without loss in generality, we can restrict to the case $J=\bb{R}^+$.

We fix $\e\in(0,1)$ and write
 \begin{align}\label{eq:splitintas}
 x^{-s}\int_0^{\infty} K\left(\frac{y}{x}\right)\frac{f(y)}{y} dy = x^{-s}\int_{0}^{x^{\e}} K\left(\frac{y}{x}\right)\frac{f(y)}{y} dy +
 x^{-s} \int_{x^{\e}}^{+\infty}  K\left(\frac{y}{x}\right)\frac{f(y)}{y} dy\;.
  \end{align}
  We analyse the first term on the right-hand-side of (\ref{eq:splitintas}) and prove that it converges to $0$, when $\e$ is chosen appropriately.
  Due to the hypothesis on $f$, there exist $C,D>0$ such that $\sup_{u\in[0,y]}|f(u)| \leq C + D y^{\max\lbrace s,0\rbrace} $.
 Therefore, using (\ref{eq:maximumkernel}) we obtain
 \begin{multline*}
  x^{-s}\left|\int_{0}^{x^{\e}} K\left(\frac{y}{x}\right)\frac{f(y)}{y} dy \right| \leq x^{-s-1}\int_0^{x^{\e}}  \left( C + D y^{\max\lbrace s,0\rbrace}\right) dy \\ 
  \leq \frac{x^{-s+\e-1}}{\pi\sigma}
  \left( C+ D x^{\e \max\lbrace s,0\rbrace}\right) \;.
  \end{multline*}
  The above expression vanishes as $x \to 0$ for all $\e \in (0,1)$ if $s \in[0,1)$, and for all $\e \in (0,1+s)$ if $s\in(-1,0)$.
  
  We now analyse the second term in the right-hand-side of (\ref{eq:splitintas}).
  We have
  \begin{multline*}
 x^{-s}\int_{x^{\e}}^{\infty} K\left(\frac{y}{x}\right)\frac{f(y)}{y} dy = x^{-s}\int_{x^{-1+\e}}^{\infty} K(u)\frac{f(u x)}{u} du =
 \frac{\sigma}{\pi}\int_{x^{-1+\e}}^{\infty}  \frac{x^{-s} f(ux) }{1+u^2-2 \kappa u} du\\
 = \frac{\sigma}{\pi} \int_{x^{-1+\e}}^{\infty}  \frac{c\, u^{s}}{1+u^2-2 \kappa u} du - \frac{\sigma}{\pi} \int_{x^{-1+\e}}^{\infty}
 \frac{c\,u^{s}-x^{-s} f(ux) }{1+u^2-2 \kappa u} du \;.
  \end{multline*}
Since  $\frac{u^{s}}{1+u^2-2 \kappa u}$ is integrable on $\bb{R}^+$ for all $s\in(0,1)$, we have that
\begin{equation*}
 \lim_{x \to +\infty} \frac{\sigma}{\pi} \int_{x^{-1+\e}}^{\infty}  \frac{u^{s}}{1+u^2-2 \kappa u} du =
 \frac{\sigma}{\pi} \int_0^{\infty}  \frac{u^{s}}{1+u^2-2 \kappa u} du =\Phi(\alpha,s)\;,
\end{equation*}
Moreover, writing
\begin{align*}
& \frac{\sigma}{\pi} \int_{x^{-1+\e}}^{\infty}
 \frac{- c\, u^{s}+x^{-s} f(ux) }{1+u^2-2 \kappa u} du = \frac{\sigma}{\pi} \int_{x^{-1+\e}}^{\infty}
 \frac{u^{s}}{1+u^2-2 \kappa u}  \left( (ux)^{-s} f(ux)-c\right) du\;,
 \end{align*}
 we deduce that
 \begin{align*}
&  \left|\frac{\sigma}{\pi} \int_{x^{-1+\e}}^{\infty}
 \frac{- c\, u^{s}+x^{-s} f(ux) }{1+u^2-2 \kappa u} du\right| \leq \left(\sup_{u \in [x^{-1+\e},+\infty)} |(ux)^{-s} f(ux)-c|\right) \Phi(\alpha,s) \;.
\end{align*}
By hypothesis on $f$, we have that
$$\lim_{x \to + \infty} \sup_{u \in [x^{-1+\e},+\infty)} |(ux)^{-s} f(ux)-c|=
\lim_{x \to + \infty} \sup_{t \in [x^{\e},+\infty)} |t^{-s} f(t)-c|=0 \;,\quad \forall \e >0\;,$$
from which it follows that $
 \lim_{x \to +\infty}\frac{\sigma}{\pi} \int_{x^{-1+\e}}^{\infty}
 \frac{c\,u^{s}-x^{-s} f(ux) }{1+u^2-2 \kappa u} du =0$, for all $\e>0$. Therefore equation (\ref{eq:limitsrespect}) is proven.\\
(2) Since
\begin{multline*}
  \left| x^{-s} \bb{K}_I[f](x) \right| \leq x^{-s} \|f\|_{\infty,s} \int_{I} K_{\alpha}\left(\frac{x}{y}\right)y^{s-1} dy \\
  \leq x^{-s} \|f\|_{\infty,s} \int_{\omega}^{\infty} K_{\alpha}\left(\frac{x}{y}\right) y^{s-1} dy \leq   \|f\|_{\infty,s}
 \int_{\omega/x}^{\infty} K_{\alpha}(t)\,t^{s-1} dy\;,
\end{multline*}
the thesis follows directly from the definition of $\Phi(\alpha,s;r)$.\\
(3) Since $\Phi(\alpha,s;r)< \Phi(\alpha,s)$ for all $r>0$, the thesis follows directly from equations (\ref{eq:limitsrespect}) and (\ref{eq:puntualestimate}).\\
(4) The norm of $\bb{K}_I$ on $L^{\infty}_s(I)$ is $\Phi(\alpha,s)= \frac{\sin\left(\frac{\pi(\alpha-1)}{1+\alpha}s\right)}{\sin\left(\pi s\right)}$.
The latter function is smaller than one provided $s\in\left(-\frac{1+\alpha}{2\alpha},\frac{1+\alpha}{2\alpha}\right)$.
Therefore $\bb{I}-\bb{K}_I $ is invertible on $L^{\infty}_s(I)$ and the resolvent is
given by the Neumann series
$$\left(\bb{I}-\bb{K}_I \right)^{-1} = \sum_{n\geq0} \left(\bb{K}_I\right)^n\;.$$
The thesis follows immediately from the latter identity. \\
(5) The kernel of $\bb{I}-\bb{K}_{\omega}$
on $L^{\infty}_{\frac{\alpha+1}{2\alpha}}\left([\omega,+\infty)\right)$
has dimension one, see \cite[Section3]{baluzaII}.\footnote{The fact that the kernel of $\bb{I}_\omega-\bb{K}_\omega$
 has dimension one follows from the fact that the function $1-\Phi(\alpha,s)$ has a simple zero at $s=\frac{1+\alpha}{2\alpha}$,
 that is the Mellin transform of the resolvent has a simple pole at $s=-\frac{1+\alpha}{2\alpha}$.}
 After Proposition (\ref{prop:tauJ}) it is generated by the function $x\fun'(\omega x)$. Moreover,
 From Proposition (\ref{prop:tauJ}), we know that
$x\fun'(\omega x)= \frac{\alpha+1}{2\alpha} \omega^{\frac{\alpha-1}{2\alpha}} x^{\frac{1+\alpha}{2\alpha}} +
O\left(x^{-\frac{1+\alpha}{2\alpha}}\right)$ as $x\to+\infty$, from which equation (\ref{eq:subtracteds}) follows. \\
(6) By hypothesis $g \in L^{\infty}_{s}\left([\omega,+\infty)\right)$, hence $f$ is as per formula
(\ref{eq:affinesolution}). Reasoning as above, we deduce that
$f(x) = c x^{\frac{1+\alpha}{2\alpha}} +G + O\left(x^{-\frac{1+\alpha}{2\alpha}}\right)$ as $x \to +\infty$, for some $c>0$, with
$G=\left(\bb{I}-\bb{K}_\omega\right)^{-1}[g]$.
After (\ref{eq:limitsrespect}) we have that
$$\lim_{x \to +\infty} x^{-s}  \left(\bb{K}_\omega\right)^n(x)= \left(\Phi(\alpha,s)\right)^n d \;.$$
Since $G$ is explicitly given by the Neumann series
$\sum_{n\geq0} \left(\bb{K}_\omega\right)^n[g]$, the above computation yields the thesis.
 \end{proof}

\end{proposition}

\subsection*{The operator $\bb{K}^{(n)}_J$ with $n\geq 1$}

Here we establish the analogous of Proposition \ref{prop:Kzasym}
for the operators $\bb{K}_J^{(n)}$ with $n\geq 1$.
\begin{lemma}\label{lem:Knbounded}
(1) Let $I$ be an admissible set.
 For all $s \in(-1,1)$, $\bb{K}^{(n)}_I$ is a continuous operator on the space $L^{\infty}_{s}(I)$, and
\begin{equation} \label{eq:Psiasn}
 \| \bb{K}^{(n)}_I \|_{\infty,s}= \Psi(\alpha,s;n) := \int_{0}^{\infty}\left|K_{\alpha}(x;n)\right|x^{s-1} dx < \infty \;.
\end{equation}
(2) Let $J=\bb{R}^+$ or an admissible interval,  $s \in (-1,1)$ and $f:\bb{R}^+ \to \bb{R}$ be a locally bounded function. Then,
 \begin{equation}\label{eq:limitsrespectn}
 \lim_{x \to +\infty}  x^{-s} f(x)= c \in \bb{R}^+ \quad\Longrightarrow\quad \lim_{x \to +\infty} x^{-s}\bb{K}^{(n)}_J[f](x)=
 s^n \Phi(\alpha,s) c\;,
 \end{equation}
where $\Phi(\alpha,s)$ is as per (\ref{eq:completeintegral}).\\
(3) Let $\omega >0$. Assume that $ D^k g \in L^{\infty}_s([\omega,+\infty))$, with
$s \in  \left(-\frac{\alpha+1}{2\alpha},\frac{\alpha+1}{2\alpha}\right)$, for all $k\in\{0,\dots,n\}$.
Let $f$ be a solution of the linear integral equation
\begin{equation}\label{eq:fKfgn}
 f=\bb{K}_{\omega}[f]+g\;,\quad f \in L^{\infty}_{\frac{\alpha+1}{2\alpha}}\left([\omega,+\infty)\right)\;.
\end{equation}
There exists a $c \in \bb{R}$
such that
\begin{equation}\label{eq:subtractedlimitn}
 D_{x}^k\left( f(x) - c x^{\frac{1+\alpha}{2\alpha}}\right) \in L^{\infty}_s([\omega,+\infty))\;,\quad \forall k\in\{0,\dots,n\}\;.
\end{equation}

\begin{proof}
 (1) To prove (\ref{eq:Psiasn}), it is enough to notice that $K_{\alpha}(x;n)=O(x)$ as $x \to 0^+$ and $K_{\alpha}(x;n)=O(x^{-1})$ as $x \to +\infty$.\\
 (2) The proof of (\ref{eq:limitsrespectn}) follows the same steps of the proof of the case $n=0$, given in Proposition
 \ref{prop:Kzasym} (1). It is therefore omitted.\\
 (3) After Proposition \ref{prop:intzI} (5), the case $n=0$ follows. A simple computation shows that
 $D^k_x\left(\bb{K}_{\omega}[y^{\frac{1+\alpha}{2\alpha}}](x)-x^{\frac{1+\alpha}{2\alpha}}\right) = O(x^{-1})$
 for all $k\geq0$. Hence, defining $\hat{f}=f-c x^{\frac{1+\alpha}{2\alpha}}$, we have that
 $$
 D^k \hat{f}= \bb{K}^{(n)}_{\omega}[\hat{f}] + \hat{g}^{(k)}\;,\quad \hat{f},\hat{g}^{(k)} \in L^{\infty}_s\left([\omega,+\infty)\right)\;,
 $$
 where $\hat{g}^{(k)}(x)=D_x^k g(x)+ D^k_x\left(\bb{K}_{\omega}[y^{\frac{1+\alpha}{2\alpha}}](x)-x^{\frac{1+\alpha}{2\alpha}}\right) $.
 Since $\bb{K}_{\omega}$ is a bounded operator on $L^{\infty}_s\left([\omega,+\infty)\right)$, the thesis follows.
\end{proof}
\end{lemma}

As a Corollary of the above Lemma, we deduce
a simple but important a-priori estimate on the $L^{\infty}_s$ norms of the solutions
to the standard IBA equation (\ref{eq:IBAStGen}).

\begin{lemma}\label{lem:lemma1e2}\label{lem:apriorie}
Let $\{z,\underline{h},\omega\}$ be a solution of the standard IBA equation (\ref{eq:IBASt}).\\
(1) We have that
 \begin{align}
& \tilde{z}_n(x):=
D^n_x z(x)- \left(\frac{\alpha+1}{2\alpha}\right)^n x^{\frac{1+\alpha}{2\alpha}}
\in L^{\infty}\left([\omega,+\infty)\right)\;,\quad \forall n \geq 0\;,\label{eq:apriori1}
\end{align}
(2) Let $\la(\xi)=z(\omega\xi)-\bar{l}(\xi,\omega,p)$ with $\xi\in[1,+\infty)$ -- where $\bar{l}(\xi,\omega,p)$ is the unique normalised solution of the linear IBA equation (\ref{eq:IBAStL}), as per Theorem \ref{prop:fproperty}.

We have that
\begin{equation}\label{eq:apren}
\|D^n \la\|_{\infty} \lesssim_n 1 \;,\quad \forall n \geq 0\;.
\end{equation}
In particular,
  \begin{align}
  \label{eq:apre}
 & \|\la\|_{\infty} \leq \left(H+\frac{1}{4}\right)(\alpha-1)\;.
 \end{align}
\begin{proof}
(1) The standard IBA equation (\ref{eq:IBASt}) is an equation of the form 
\begin{align*}
    &f=\bb{K}_{\omega}[f]+g\;,\\
    &\mbox{with }g=-2p+\bb{K}_{\omega}[\afr{f}]- H F_\alpha\left(\frac{x}{\omega}\right)+ \sum_{k=1}^H F_\alpha\left(\frac{x}{h_k}\right)\;.
\end{align*}
The thesis follows from applying Lemma \ref{lem:Knbounded} (3). In this case,
$c=1$ by hypothesis on $z$, and $D^n g \in L^{\infty}\left([\omega,+\infty)\right)$ for all $n \geq 0$.
In fact, $D^n\bb{K}_{\omega}[\afr{z}]=
\bb{K}^{(n)}_{\omega}[\afr{z}]\in L^{\infty}\left([\omega,+\infty)\right)$
since $\afr{z}$ is bounded, and $D^n_x F_\alpha\left(\frac{x}{y}\right)=O(x^{-1})$ for all $y >0$. \\

(2) The same proof of (1) works.
The perturbed IBA equation (\ref{eq:IBASte}) is of the form
\begin{align*}
&\la= (\bb{I}-\bb{K}_{\omega})^{-1} [g]\;,\\
&\mbox{with } g= \bb{K}_{\omega}[\afr{z}]- H F_\alpha\left(\frac{x}{\omega}\right)+ \sum_{k=1}^H F_\alpha\left(\frac{x}{h_k}\right)\;. %\in L^{\infty},\; \la \in L^{\infty}.    
\end{align*}
We have that
$\|\bb{K}_{\omega}[\afr{z}]\|_{\infty}\leq\frac{1}{2} \| \bb{K}_{\omega} \|_{\infty}$ since $\afr{z}=\frac12$,
and $\|F_\alpha(x/y) \|_{\infty}=\frac{\alpha-1}{\alpha+1}$ as per (\ref{eq:Ffunza}).
Therefore we can apply Lemma \ref{lem:Knbounded}(3) to obtain (\ref{eq:apren}) --
notice that in this case
$c=0$, since by hypothesis $z$ and $l$ have the same normalisation at infinity.

Using the norm of the resolvent in the space $L^{\infty}$, as given by formula (\ref{eq:normResolvent}),
we obtain the estimate (\ref{eq:apre}).

\end{proof}

\end{lemma}

\section{Oscillatory integrals and a-priori estimates}
\label{sec:osc}
Since we study the IBA equation as a perturbation of its linearisation, we need to estimate
the magnitude of the perturbing nonlinear term $\bb{K}_1[\afr{\bar{l}+
\la}]$, where $\bar{l}$ is the solution of the linearised standard IBA (\ref{eq:IBAStL})
and $\la$ a putative solution of the perturbed IBA (\ref{eq:IBASteGen}).

To be more precise, there are two kinds of integrals that we need to study.
The first kind are integrals of the form
\begin{equation*}
 \int_{1}^{\infty}K_{\alpha}\left(\frac{x}{y};n\right)\afr{pf(y)}\frac{dy}{y}\;,\quad \alpha>1\;,
\end{equation*}
where $K_\alpha(\cdot;n)$ is as per (\ref{eq:Knal}), while $f:[1,+\infty)\to\R$ is such that
$f(x)=c \, x^{\frac{\alpha+1}{2\alpha}}+ O\left(x^{\frac{\alpha+1}{2\alpha}}\right)$ as $x \to +\infty$ for some $c>0$
(plus some regularity conditions). The second kind of integrals are of the form
\begin{equation*}
\int_1^{\infty}K_{\alpha}\left(\frac{x}{y};n\right)\left(\ceili{p f(y)+\e(y)-\frac12}-\ceili{p f(y)+\tilde{\e}(y)-\frac12}\right)
 \frac{dy}{y}\;,\quad \alpha>1\;,
\end{equation*}
where $f$ satisfies the same conditions as above, while $\e,\tilde{\e}$ are bounded perturbations.

The study of these integrals is the analytical cornerstone of our method of analysis and it will lead us
to prove the following important results:
\begin{itemize}
 \item any strictly monotone solution of the IBA equation has asymptotic
  behaviour $z(x)=x^{\frac{1+\alpha}{2\alpha}}- p(\alpha+1) + O\left(x^{-\frac{1+\alpha}{2\alpha}}\right)$ as $x \to +\infty$;
  \item the $L^{\infty}$ norm of solutions of the perturbed IBA equation and its derivatives are $O(p^{-1})$ as $p\to+\infty$;
\item the perturbed IBA equation is well-posed.
\end{itemize}
Not violating the principle no pain no gain,
both kinds of integrals poses some serious analytical challenges. In fact they are convolutions of functions which are both
discontinuous and highly-oscillatory.

In order to start with our investigation, we define a reasonable class of functions.
\begin{definition}\label{def:ZdA}
 Let $\delta\in\left(\frac12,1\right)$ and $c>0$. We define $Z_{\delta,c} \subset \mc{C}^2([1,+\infty))$ as
 the set of functions $f\in\mc{C}^2([1,+\infty))$ such that $f(1)=0$, $f'(x)>0$ for all $x \in [1,+\infty)$, and
  \begin{equation*}
D_x^n f= c\, \delta^n  x^{\delta}+ o\left(x^{\delta}\right) \mbox{ as } x \to +\infty\;,\quad n\in\{0,1,2\}\;.
\end{equation*} 
Denoting by $\varphi_f:\R^+ \to [1,+\infty)$ the right inverse of $f$, namely the unique function defined by the equation
$f(\varphi_f(t))=t$ for all $t\in\R^+$, we define
\begin{equation}\label{eq:NiGamma}
    N_i(f)=\left\|\varphi_f^{(i)}\right\|_{\infty,\frac{1}{\delta}-i}^+ \;,
    \quad \Gamma(f)=\inf_{t \in[1,+\infty)} \left|t^{-\frac{1}{\delta}}\varphi_f(t)\right|\;,\quad i\in\{1,2\}\;,
\end{equation}
where the norms $\| \cdot \|^{+}_{\infty,\gamma}$ are as per (\ref{eq:weightedsupR}).
\end{definition}

\begin{lemma}\label{lem:lemma3}
  If $f\in Z_{\delta,c}$, then $N_1(f)$ and $N_2(f)$ are finite.
  \begin{proof}
  Let $\varphi_f:\R^+ \to [1,+\infty)$ be the right inverse of $f$.
  By the inverse function Theorem we have that $\varphi_f \in \mc{C}^2(\R^+)$. Moreover, a simple computation yields
\begin{equation*}
 D_t^n\varphi_f^{(n)}(t)= c^{-\frac{1}{\delta}} \delta^{-n} t^{\frac{1}{\delta}}+o\left( t^{\frac{1}{\delta}}\right)
 \mbox{ as } t\to+\infty \;,\quad n\in\{0,1,2\}\;,
\end{equation*}  
from which the thesis follows.
  \end{proof}
\end{lemma}

Before stating and proving our estimates on the integrals mentioned above, we need a quite involved preparatory Lemma, which is the core of our method.

\begin{lemma}\label{lem:lemma4}
Fix $\delta\in\left(\frac12,1\right)$, $c>0$ and $p_0>0$.

With $k \in \bb{Z}$, let $t_k:=\frac{k-b}{p}$. The following estimates hold
\begin{multline}
\sum_{k \geq \ceili{b+\frac{1}{2}}}\sup_{t \in [t_{k-1},t_k]}
\left| \frac{d}{dt}\left[
 K_{\alpha}\left(\frac{x}{\varphi_f(t)};n\right)\frac{\varphi_f'(t)}{\varphi_f(t)}\right]
\right| \lesssim_n  
c\, p \, x^{-\delta} \left[ \left(\frac{N_1(f)}{\Gamma(f)}\right)^2 + \frac{N_2(f)}{\Gamma(f)} \right]\;,\\
\forall (f,p,b) \in Z_{\delta,c} \times [p_0,+\infty) \times \bb{R}\;.
\label{eq:seriessup}  
\end{multline}

\begin{proof}
We reduce to the case $c=1$, via the transformation $(f,p)\to \left(\frac{f}{c},p c\right)$,
which leaves the product $c\,p$ unchanged.
For sake of brevity, we write $K(\cdot)$, $\varphi$, $N_i$ and $\Gamma$ instead of $K_{\alpha}(\cdot,n)$, $\varphi_f$,
$N_i(f)$ and $\Gamma(f)$, respectively.
Moreover, we define
 \begin{equation}\label{eq:tkp1} \tag{a}
  \bar{k}=k-b\;,\quad k_x=\ceili{b+ p x^{\delta}+ \frac12}\;,\quad t_k=\frac{\bar{k}+\frac12}{p}\;.
 \end{equation}
(1) Differentiating, we get
\begin{align*} 
\frac{d}{dt} \left[
K\left(\frac{\varphi(t)}{x}\right)\frac{\varphi'(t)}{\varphi(t)} \right] =
K\left(\frac{\varphi(t)}{x}\right)\frac{\varphi''(t)}{\varphi(t)} + \left(\varphi'(t)\right)^2
\frac{d}{d\varphi} \left.\left[ K\left(\frac{\varphi}{x}\right)\frac{1}{\varphi}\right]\right|_{\varphi=\varphi(t)}\;.
\end{align*}
Therefore,
\begin{multline}
\sup_{t \in[t_{k-1},t_k]} \left|\frac{d}{dt} \left[
K\left(\frac{\varphi(t)}{x}\right)\frac{\varphi'(t)}{\varphi(t)} \right]\right| \leq
\left(w(t_{k-1})\right)^{\frac{1}{\delta}-2}N_2
\sup_{t \in[t_{k-1},t_k]} \left|K\left(\frac{\varphi(t)}{x}\right)\frac{1}{\varphi(t)}\right| \\ \label{eq:sumsplitp1} \tag{b}
+ \left(w(t_{k})\right)^{\frac{2}{\delta}-2} N_1^2 \sup_{t \in[t_{k-1},t_k]} \left|
\frac{d}{d\varphi} \left[ K\left(\frac{\varphi}{x}\right)\frac{1}{\varphi}\right]\right| \;.
\end{multline}
In the above inequality we have used the definition of the norms $\|\cdot\|^+_{\infty,\gamma}$, namely $\|g \|^+_{\infty,\gamma}=
\sup_{x \in\R^+} |w^{-\gamma}(x) g(x)| $, with
$w(x)=\chi_{[0,1)}(x)+ x \chi_{[1,+\infty)}(x)$.\\
%We analyse separately the two series on the right hand side of (\ref{eq:sumsplitp1}).\\
(2) We start by considering the first term on the right-hand-side of (\ref{eq:sumsplitp1}) and we show that
\begin{equation}\label{eq:oscsingproof1} \tag{c}
 N_2\sum_{k\geq k_0} \left(w(t_{k-1})\right)^{\frac{1-2\delta}{\delta}}
 \sup_{t \in[t_{k-1},t_k]} \left|K\left(\frac{\varphi(t)}{x}\right)\frac{1}{\varphi(t)}\right|\lesssim
 c\,p\, x^{-\delta} \frac{N_2}{\Gamma}\;.
\end{equation}
Notice that $w(t)=1$ for $t\leq 1$ and $w(t)=t$ for $t \geq 1$.
Since $t_k \leq 1$ for all $k\leq k_1$ and $t_k\geq 1$ for all $k\geq k_1$,
we split the series (\ref{eq:oscintproof1}) in two sub-series. The first one reads
\begin{multline}\label{eq:oscillatingproof1s1} \tag{d}
N_2  \sum_{k=k_0}^{k_1}  \left(w(t_{k-1})\right)^{\frac{1}{\delta}-2}
 \sup_{t \in[t_{k-1},t_k]} \left|K\left(\frac{\varphi(t)}{x}\right)\frac{1}{\varphi(t)}\right| \\ \leq N_2
 \sum_{k=k_0}^{k_1} 
 \sup_{t \in[t_{k-1},t_k]} \left|K\left(\frac{\varphi(t)}{x}\right)\frac{1}{\varphi(t)}\right| \lesssim N_2 x^{-1} \left( k_1 -k_0\right)
 \lesssim p\,x^{-1} N_2\;,
\end{multline}
where we have used (\ref{eq:maximumkernel}) to estimate the supremum. The second one reads
\begin{multline}
N_2  \sum_{k\geq k_1} \left(w(t_{k-1})\right)^{\frac{1}{\delta}-2}
 \sup_{t \in[t_{k-1},t_k]} \left|K\left(\frac{\varphi(t)}{x}\right)\frac{1}{\varphi(t)}\right| \\ \leq
 \label{eq:oscillatingproof1s2} \tag{e}
N_2 p^{2-\frac{1}{\delta}}
\sum_{k \geq k_1} \bar{k}^{\frac{1}{\delta}-2}  \sup_{t \in[t_{k-1},t_k]}  
\left| K\left(\frac{\varphi(t)}{x}\right)\frac{1}{\varphi(t)}\right|\;.
\end{multline}
Now, the series $\sum_{k\geq k_1} k^{\frac{1}{\delta}-2}$ does not converge since $\delta\in\left(\frac12,1\right)$. However,
we can overcome this difficulty by splitting the series 
using the intermediate summation limit $ k_x$.
Using again (\ref{eq:maximumkernel}) and comparing the sum $\sum_{k=k_1}^{k_x} \bar{k}^{\frac{1}{\delta}-2}$
with the integral $\int_{k_1}^{k_x} \bar{k}^{\frac{1}{\delta}-2} d \bar{k}$ we have
 \begin{multline}
 N_2 p^{2-\frac{1}{\delta}}
 \sum_{k =k_1}^{k_x}  \bar{k}^{\frac{1}{\delta}-2}  \sup_{t \in[t_{k-1},t_k]}  
\left| K\left(\frac{\varphi(t)}{x}\right)\frac{1}{\varphi(t)}\right| \lesssim
 N_2  p^{2-\frac{1}{\delta}} x^{-1} 
 \sum_{k =k_1}^{k_x} \bar{k}^{\frac{1}{\delta}-2}\\ \label{eq:oscillatingproof1s21} \tag{f}
  \lesssim  N_2 p^{2-\frac{1}{\delta}} x^{-1}
  \left[\left(p x^{\delta}+\frac12\right)^{\frac{1}{\delta}-1}-\left(p+\frac{1}{2}\right)^{\frac{1}{\delta}-1}\right]\lesssim 
 p \, x^{-\delta}N_2 \;,
 \end{multline}
and
 \begin{multline} \label{eq:oscillatingproof1s22} \tag{g}
 N_2p^{2-\frac{1}{\delta}}
 \sum_{k \geq k_x }  \bar{k}^{\frac{1}{\delta}-2}  \sup_{t \in[t_{k-1},t_k]}  
\left| K\left(\frac{\varphi(t)}{x}\right)\frac{1}{\varphi(t)}\right|  \\ 
 \leq N_2 p^{2-\frac{1}{\delta}}
 \sum_{k \geq k_x } \bar{k}^{\frac{1}{\delta}-2}  \sup_{t \in[t_{k-1},t_k]}  
\left| K\left(\frac{\varphi(t)}{x}\right)\right| \sup_{t \in[t_{k-1},t_k]}\left|\frac{1}{\varphi(t)}\right|\\ 
 \lesssim \frac{N_2}{\Gamma} p^{2-\frac{1}{\delta}}
 \sum_{k \geq k_x} \bar{k}^{\frac{1}{\delta}-2} (t_{k-1})^{-\frac{1}{\delta}}\lesssim  \frac{N_2}{\Gamma} p^{2}
 \sum_{k \geq k_x}\bar{k}^{-2}
 \lesssim p\,x^{-\delta}\frac{N_2}{\Gamma}\;.
 \end{multline}
 where we used the fact that $K$ is a bounded function, see (\ref{eq:maxkernel2}).
 The estimates (\ref{eq:oscillatingproof1s1},\ref{eq:oscillatingproof1s2},
 \ref{eq:oscillatingproof1s21},\ref{eq:oscillatingproof1s22})
 yields (\ref{eq:oscsingproof1}).\\
 (3) Now we prove the following estimate for the second series on the right-hand-side of (\ref{eq:sumsplitp1}). We shall prove that
\begin{equation}\label{eq:oscsingproof2} \tag{h}
 N_1^2 \sum_{k\geq k_0} \left(w(t_{k-1})\right)^{\frac{2}{\delta}-2}
 \sup_{t \in[t_{k-1},t_k]} \left|\frac{d}{d \varphi}K\left(\frac{\varphi(t)}{x}\right)\frac{1}{\varphi(t)}\right|\lesssim
 p\,x^{-\delta}\frac{N_1^2}{\Gamma^2}\;.
\end{equation}
As we did above, we split the series in three, introducing the intermediate summation limits $k_1,
k_x$. Using (\ref{eq:maxkernel3}), we get
\begin{multline}\label{eq:oscillatingproof2s1} \tag{i}
N_1^2 \sum_{k=k_0}^{k_1} \left(w(t_{k-1})\right)^{\frac{2}{\delta}-2} 
 \sup_{t \in[t_{k-1},t_k]} \left|\frac{d}{d\varphi }K\left(\frac{\varphi(t)}{x}\right)\frac{1}{\varphi(t)}\right| \\ 
 \leq N_1^2 \sum_{k=k_0}^{k_1}
 \sup_{t \in[t_{k-1},t_k]} \left|\frac{d}{d\varphi }K\left(\frac{\varphi(t)}{x}\right)\frac{1}{\varphi(t)}\right|
 \lesssim N_1^2 x^{-2}  \left( k_1-k_0\right)
 \lesssim p\,x^{-2}N_1^2\;.
\end{multline}
We are left to study
\begin{multline} \label{eq:oscillatingproof2s2} \tag{j}
N_1^2 \sum_{k\geq k_1} \left(w(t_{k-1})\right)^{\frac{2}{\delta}-2}
 \sup_{t \in[t_{k-1},t_k]} \left|\frac{d}{d\varphi }K\left(\frac{\varphi(t)}{x}\right)\frac{1}{\varphi(t)}\right|\\ 
\leq N_1^2  \sum_{k\geq k_1} t_k^{\frac{2}{\delta}-2}
 \sup_{t \in[t_{k-1},t_k]} \left|\frac{d}{d\varphi }K\left(\frac{\varphi(t)}{x}\right)\frac{1}{\varphi(t)}\right|\\
\leq N_1^2 p^{2-\frac{2}{\delta}} \sum_{k\geq k_1} \bar{k}^{\frac{2}{\delta}-2}
 \sup_{t \in[t_{k-1},t_k]} \left|\frac{d}{d\varphi }K\left(\frac{\varphi(t)}{x}\right)\frac{1}{\varphi(t)}\right|\;.
\end{multline}
Using (\ref{eq:maxkernel3}) and comparing the sum $\sum_{k=k_1}^{k_x} \bar{k}^{\frac{2}{\delta}-2}$
with the integral $\int_{k_1}^{k_x} \bar{k}^{\frac{2}{\delta}-2} d \bar{k}$, we get
\begin{multline}
N_1^2 p^{2-\frac{2}{\delta}} \sum_{k= k_1}^{k_x} \bar{k}^{\frac{2}{\delta}-2} 
 \sup_{t \in[t_{k-1},t_k]} \left| \frac{d}{d\varphi} \left( K\left(\frac{\varphi}{x}\right)\frac{1}{\varphi}\right) \right| \lesssim
 N_1^2 x^{-2} p^{2-\frac{2}{\delta}} \sum_{k= k_1}^{k_x} \bar{k}^{\frac{2}{\delta}-2}  \\
 \label{eq:oscillatingproof2s21} \tag{k}
 \lesssim  N_1^2 x^{-2} p^{2-\frac{2}{\delta}} \left[\left(p x^\delta+\frac{1}{2}\right)^{-1+\frac{2}{\delta}}
- \left(p+\frac12\right)^{-1+\frac{2}{\delta}}\right] \lesssim  p\,x^{-\delta} N_1^2  \;.
\end{multline}
Finally, using (\ref{eq:maxkernel4}) and comparing the series $\sum \bar{k}^{-2}$
with the integral $\int \bar{k}^{-2} d \bar{k}$, we have that
\begin{multline}
N_1^2 p^{2-\frac{2}{\delta}} \sum_{k\geq k_x} \bar{k}^{\frac{2}{\delta}-2}
 \sup_{t \in[t_{k-1},t_k]} \left| \frac{d}{d\varphi} \left( K\left(\frac{\varphi}{x}\right)\frac{1}{\varphi}\right) \right|\\ 
 \leq N_1^2 p^{2-\frac{2}{\delta}} \sum_{k\geq k_x} \bar{k}^{\frac{2}{\delta}-2}
 \sup_{t \in[t_{k-1},t_k]} \sup_{x\geq1} \left| \frac{d}{d\varphi} \left( K\left(\frac{\varphi}{x}\right)\frac{1}{\varphi}\right) \right| \\ 
 \lesssim N_1^2 p^{2-\frac{2}{\delta}} \sum_{k\geq k_x} \bar{k}^{\frac{2}{\delta}-2} \sup_{t \in [t_{k-1},t_k]}\frac{1}{\left(\varphi(t)\right)^2} 
\lesssim \frac{N_1^2}{\Gamma^2} p^{2-\frac{2}{\delta}} \sum_{k\geq k_x} \bar{k}^{\frac{2}{\delta}-2} t^{-2}_{k-1}  \\
\label{eq:oscillatingproof2s22} \tag{l}
 \lesssim \frac{N_1^2}{\Gamma^2} p^2 \sum_{k\geq k_x}\bar{k}^{-2} 
 \left(1-\frac{1}{2\bar{k}}\right)^{-\frac{2}{\delta}} \lesssim p\,x^{\delta} \frac{N_1^2}{\Gamma^2} \;.
\end{multline}
The estimates (\ref{eq:oscillatingproof2s1},\ref{eq:oscillatingproof2s2},\ref{eq:oscillatingproof2s21},\ref{eq:oscillatingproof2s22})
 yields (\ref{eq:oscsingproof2}).
 
 Finally, the estimates (\ref{eq:sumsplitp1},\ref{eq:oscsingproof1},\ref{eq:oscsingproof2})
 yields the thesis.
 
\end{proof}

\end{lemma}

We are now in the position of proving the following Proposition, on the integrals of the first type.
\begin{proposition}\label{prop:oscillatorysingle}
Fix $\delta\in\left(\frac12,1\right)$, $c>0$ and $p_0>0$.
For all $(f,p,b) \in Z_{\delta,c}\times[p_0,+\infty)\times\R$
consider the integral
 \begin{equation*}
 O_n(x;p):=\int_1^{\infty}K_{\alpha}\left(\frac{x}{y};n\right)\afr{b+pf(y)}\frac{dy}{y}\;,\quad x \geq 1 \;.
 \end{equation*}
The following estimates hold
\begin{multline}
\label{eq:oscillatorysingle}  
\left|O_n(x;p)\right| \lesssim \frac{1}{c\,p\,x^\delta}\left[ N_1(f)+ 
\left(\frac{N_1(f)}{\Gamma(f)}\right)^{2} + \frac{N_2(f)}{\Gamma(f)} \right]\;,\\
\forall (f,p,b) \in Z_{\delta,c} \times [p_0,+\infty) \times \bb{R}\;.
\end{multline}
where the constants $N_1(f)$, $N_2(f)$ and $\Gamma(f)$ are as per (\ref{eq:NiGamma}).

\begin{proof}
We reduce to the case $c=1$, via the transformation $(f,p)\to \left(\frac{f}{c},c\,p\right)$.
To simplify the notation, we write $K(\cdot)$, $\varphi$, $N_i$ and $\Gamma$ instead of $K_{\alpha}(\cdot;n)$, $\varphi_f$,
$N_i(f)$ and $\Gamma(f)$, respectively.
Moreover, we write
 \begin{equation}\label{eq:tk} \tag{a}
  \bar{k}=k-b\;,\quad k_x=\ceili{b+ p x^{\delta}+ \frac12}\;,\quad t_k=\frac{\bar{k}+\frac12}{p}\;,\quad u_k=\frac{\bar{k}}{p}\;.
 \end{equation}
(1) As a first step, we make the change of variable $y=\varphi(t)$ and we write
 \begin{equation}\label{eq:Onp1} \tag{b}
  O_n(x;p)=\int_{0}^{\infty} K\left(\frac{\varphi(t)}{x}\right)\frac{\varphi'(t)}{\varphi(t)} \afr{p t+b} dt \;.
 \end{equation}
The function $\afr{b+p t}$ is discontinuous at the points $t=t_k$ and 
\begin{equation*}
\afr{b+p t}= p t-\bar{k} \;,\quad \forall t\in[t_{k-1},t_k]\;.
\end{equation*}
Therefore we write 
\begin{multline}\tag{c}
O_n(x;p)=\int_{0}^{t_{k_0}} K\left(\frac{\varphi(t)}{x}\right)\frac{\varphi'(t)}{\varphi(t)} \afr{b+pt} dt \label{eq:oscintproof1}
  \\
  + \sum_{k\geq k_0+1} \int_{t_{k-1}}^{t_k}K\left(\frac{\varphi(t)}{x}\right)\frac{\varphi'(t)}{\varphi(t)} \left(p\,t-\bar{k}\right) dt\;.
\end{multline}
(2) Due to (\ref{eq:maximumkernel}), we have that
\begin{multline}
\sup_{t \in [0,t_{k_0}]}\left|  K\left(\frac{\varphi(t)}{x}\right)\frac{\varphi'(t)}{\varphi(t)} \right| \lesssim x^{-1} \sup_{t\in[0,t_{k_0}]}
 \left| \varphi' \right| \leq C_{\alpha} x^{-1} \sup_{t\in[0,1]}
 \left| \varphi' \right| 
 \Longrightarrow\\ 
\left|\int_{0}^{t_{k_0}} K\left(\frac{\varphi(t)}{x}\right)\frac{\varphi'(t)}{\varphi(t)} \afr{b+pt} dt\right| \lesssim
 N_1 x^{-1} t_{k_0}\lesssim N_1 x^{-1} p^{-1}.
\end{multline}
This term is already present in the estimate (\ref{eq:oscillatorysingle}), since $\delta <1$. Hence, we just need to verify
the estimate (\ref{eq:oscillatorysingle}) for the series in the right-hand-side of equation
(\ref{eq:oscintproof1}).\\
(3) Now we need to estimate the terms of that series. We notice that the function $t-\bar{k}$ is monotone on the interval on $[t_{k-1},t_k]$ and vanishes at $t=u_k$.
Therefore,
\begin{multline*}
  \int_{t_{k-1}}^{t_k} K\left(\frac{\varphi(t)}{x}\right)\frac{\varphi'(t)}{\varphi(t)} \left(p\,t-\bar{k}\right) dt = 
 a_k(x,p) \int_{t_{k-1}}^{u_k}  \left(p\,t-\bar{k}\right) dt \\
 + b_k(x,p)\int_{u_k}^{t_{k}} \left( p\,t-\bar{k} \right) dt\;,
 \end{multline*}
 where
\begin{align}\nonumber
& a_k(x,p)= K\left(\frac{\varphi(t)}{x}\right)\frac{\varphi'(t)}{\varphi(t)}\;, \mbox{   for some } t \in [t_{k-1},u_k] \;, \\
& b_k(x,p) = K\left(\frac{\varphi(t)}{x}\right)\frac{\varphi'(t)}{\varphi(t)}\;, \mbox{   for some } t \in [u_k,t_k] \;. \label{eq:Onp2} \tag{d} 
\end{align}
Since $\int_{t_{k-1}}^{u_k} \left(p\,t-\bar{k}\right) dt=-\int_{u_{k}}^{t_k} \left(p\,t-\bar{k}\right) dt=-\frac{1}{8 p}$, then
\begin{equation}\label{eq:akmkbknk} \tag{e}
 \sum_{k\geq k_0+1}\int_{t_{k-1}}^{t_k}K\left(\frac{\varphi(t)}{x}\right)\frac{\varphi'(t)}{\varphi(t)} \left(p t-\bar{k}\right) dt=
 \sum_{k\geq k_0+1}\frac{1}{8 p}\left[b_k(x,p) - a_k(x,p)\right]\;.
\end{equation}
We do not know the exact value of the $a_k$'s and $b_k$'s, but we can estimate their difference.
In fact, using
the mean value Theorem we get
\begin{align} \nonumber
& \left| b_k(x,p)-a_k(x,p) \right| \leq \sup_{t \in[t_{k-1},t_k]}\left| \frac{d}{dt}\left[
K\left(\frac{\varphi(t)}{x}\right)\frac{\varphi'(t)}{\varphi(t)}\right] \right| \left|t_{k}-t_{k-1}\right| \Longrightarrow \\ \label{eq:Onp3} \tag{f} 
& \left|\sum_{k\geq k_0+1}\frac{1}{8 p}\left[b_k(x,p) - a_k(x,p)\right]\right|
\leq \frac{1}{8 p^2}\sum_{k\geq k_0+1}\sup_{t \in[t_{k-1},t_k]}\left| \frac{d}{dt} \left[
K\left(\frac{\varphi(t)}{x}\right)\frac{\varphi'(t)}{\varphi(t)}\right]\right|\;.
\end{align}
Combining (\ref{eq:Onp2},\ref{eq:Onp3}) with Lemma \ref{lem:lemma4}, the thesis follows.
\end{proof}

\end{proposition}

We turn our attention to the study of integrals of the second kind. These depend on a fixed reference function $f$, which
is assumed to belong to the space $Z_{\delta,c}$ introduced above, and to a perturbation $\e$ which is assumed to be
small with respect to $z$, in the sense of the following definition.

\begin{definition}\label{def:Xzdp}
Let $f\in Z_{\delta,c}$ and $p_0 >0$, we define
\begin{equation}
\label{eq:Xfp0}
    X_{f,p_0}=\left\lbrace \e\in\mc{C}^1\left([1,+\infty)\right)\,|\,\|\e\|_\infty<\frac12
    \mbox{ and } p_0 f'(x)+\e'(x)>0\,,\,\forall x\geq 1\right\rbrace\;.
\end{equation}
\end{definition}

\begin{proposition}\label{prop:differencenonlinear}
Fix $\delta\in\left(\frac12,1\right)$, $c>0$ and $p_0 >0$.
For any $(f,\e,p,b)\in Z_{\delta,c}\times X_{f,p_0}\times[p_0,+\infty)\times\bb{R}$,
consider the integral
\begin{equation}
 Q_n[\e](x;p):=\bb{K}^{(n)}\left[\ceili{b+ p f+ \e-\frac12}\right](x)\;,\quad x\geq 1\;.
\end{equation}

The following estimates hold 
\begin{multline}\label{eq:differencenonlinear}
 \left| \|Q_n[\e_1](\cdot;p)- Q_n[\e_2](\cdot;p) \|_{\infty}
 - \Psi(\alpha,0;n) \| \e_1-\e_2 \|_{\infty} \right| \lesssim_n  \frac{\| \e_1-\e_2 \|_{\infty}}{c\,p} \times \\   \left[\left(\frac{N_1(f)}{\Gamma(f)}\right)^2 + \frac{N_2(f)}{\Gamma(f)} \right] \;,\quad
\forall  (f,\e_1,\e_2,p,b) \in  Z_{\delta,c} \times X_{f,p_0} \times X_{f,p_0}
\times [p_0,+\infty)\times \bb{R} \;,
%  & \left| \| \bb{K}'[\ceili{b+ p f+ \e-\frac12}]-\bb{K}'[\ceili{b+ p f-\frac12}] \|_{\infty,s}
%  - \Phi(\alpha,s) \| \e \|_{\infty,s} \right|\leq
%  C_{\rho,\gamma,s} p^{-1}  \| \e \|_{\infty,s}
\end{multline}
where $\Psi(\alpha,0;n)$ is as per (\ref{eq:Psiasn}),\footnote{Recall that
$\Psi(\alpha,0;0)=\Phi(\alpha,0)=\frac{\alpha-1}{\alpha+1}$.} and the
constants $N_1$, $N_2$ and $\Gamma$ are as per (\ref{eq:NiGamma}).
\begin{proof}
We reduce to the case $c=1$, via the transformation $(f,p)\to \left(\frac{f}{c},c\,p\right)$.
To simplify the notation, we write $K(\cdot)$, $\varphi$, $N_i$ and $\Gamma$ instead of $K_{\alpha}(\cdot;n)$, $\varphi_f$,
$N_i(f)$ and $\Gamma(f)$, respectively.
Moreover, we write
 \begin{equation}\label{eq:tkp3} \tag{a}
  \bar{k}=k-b\;,\quad k_x=\ceili{b+ p x^{\delta}+ \frac12}\;,\quad t_k=\frac{\bar{k}+\frac12}{p}\;,\quad u_k=\frac{\bar{k}}{p}\;.
 \end{equation}
After the change of variable $y=\varphi(t)$, we have
\begin{multline}\label{eq:oscp31} \tag{b}
Q_n[\e_1](x;p)-Q_n[\e_2](x;p) = \int_{0}^{\infty} K\left(\frac{\varphi(t)}{x}\right)\frac{\varphi'(t)}{\varphi(t)} \times \\
\times \left(\ceili{b+ p t+ \e_1(\varphi(t))-\frac12} - \ceili{b+ p t+ \e_2(\varphi(t))-\frac12}\right) dt\;.
\end{multline}
By construction, the functions $b+p t + \e_i(\varphi(t))$ with $i\in\{1,2\}$ are strictly increasing for all $t\in\R^+$. Hence, for every $k$ there exists a unique point
$t_{k,i}$ such that
\begin{equation*}
p\, t_{k,i}+\e_i (\varphi(t_{k,i}))=\bar{k}+\frac12 \;,\quad i\in\{1,2\}\;.
\end{equation*}
Moreover 
\begin{equation*}
  t_{k,i} \in [u_{k-1},u_k]\;,\quad i\in\{1,2\}\;,
\end{equation*}
and
\begin{equation}\label{eq:tkemtk} \tag{c}
 \left| t_{k,2}-t_{k,1} \right| \leq  p^{-1}\|\e_1-\e_2 \|_{\infty} =
 \|\e_1-\e_2 \|_{\infty} 
 (u_{k+1}-u_{k})\;.
\end{equation}
Since
\begin{equation*}
\ceili{b+ p t + \e_i(\varphi(t))-\frac{1}{2}} =k \;,\quad \forall t \in (t_{k-1,i}, t_{k,i})  \;,\quad i\in\{1,2\} \;,
\end{equation*}
it follows that
\begin{equation*}
Q_n[\e_1](x;p)-Q_n[\e_2](x;p) =
\sum_{k\geq k_0} \int_{t_{k,1}}^{t_{k,2}}  K\left(\frac{\varphi(t)}{x}\right)\frac{\varphi'(t)}{\varphi(t)} dt \;.
\end{equation*}
Therefore, there exist $t^{*}_k \in [ t_{k,1},t_{k,2}] \subset [u_{k-1},u_k]$ such that
\begin{equation*}
Q_n[\e_1](x;p)-Q_n[\e_2](x;p)  =
\sum_{k\geq k_0}  K\left(\frac{\varphi(t^*_k)}{x}\right)  \frac{\varphi'(t^{*}_k)}{\varphi(t^*_k)}  \left(t_{k,2}-t_{k,1}\right)\;.
\end{equation*}
Using (\ref{eq:tkemtk}), from the latter estimate we get
\begin{multline}
\left| Q_n[\e_1](x,p)-Q_n[\e_2](x,p) \right| \leq
\sum_{k\geq k_0} \left|K\left(\frac{\varphi(t^*_k)}{x}\right)\right| \frac{\varphi'(t^{*}_k)}{\varphi(t^*_k)}  \left| t_{k,2}-t_{k,1} \right|\\ 
\leq  \|\e_1-\e_2 \|_{\infty} \sum_{k\geq k_0} \left|K\left(\frac{\varphi(t^*_k)}{x}\right)\right| \frac{\varphi'(t^{*}_k)}{\varphi(t^*_k)}\left(u_{k}-u_{k-1}\right) \\
\underset{R.S.}{\sim}  \|\e_1-\e_2 \|_{\infty} \sum_{k\geq k_0} 
\int_{u_{k-1}}^{u_{k}} \left| K\left(\frac{\varphi(t)}{x}\right)\right| \frac{\varphi'(t)}{\varphi(t)} dt \\ \label{eq:Qnestimateproof} \tag{d}
 = \|\e_1-\e_2 \|_{\infty} \int_{1}^{\infty} \left| K\left(\frac{x}{y}\right) \right| \frac{dy}{y} \leq \Psi(\alpha,0;n)  \|\e_1-\e_2 \|_{\infty}\;,
\end{multline}
where we have denoted by $\underset{R.S.}{\sim}$ the interchange of an integral with its Riemann sum.
The cost of such an interchange is small. Indeed, using the standard inequality
\begin{equation*}
\left|\int_{a}^b f(x) dx  - f(c) (b-a) \right|\leq \left(\sup_{x \in [a,b]} |f'(x)|\right) (b-a)^2\;,\quad \forall c \in [a,b]\;,
\end{equation*}
together with  Lemma \ref{lem:lemma4}, we obtain
\begin{multline}
\label{eq:RStobeestimated} \tag{e}
\sum_{k\geq k_0}  \left|K\left(\frac{\varphi(t^{*}_k)}{x}\right)\right| \frac{\varphi'(t^{*}_k)}{\varphi(t^*_k)}\left(u_{k}-u_{k-1}\right) -
 \sum_{k\geq k_0} \int_{u_{k-1}}^{u_{k}} \left| K\left(\frac{\varphi(t)}{x}\right)\right| \frac{\varphi'(t)}{\varphi(t)} dt \\
\leq \sum_{k\geq k_0} \sup_{t \in [u_{k+1},u_{k}]} \left| \partial_t  K\left(\frac{\varphi(t)}{x}\right) \frac{\varphi'(t)}{\varphi(t)}\right|  \left(u_{k+1}-u_k \right)^2 \\
\leq  p^{-2}\sum_{k\geq k_0} \sup_{t \in [u_{k+1},u_{k}]} \left| \partial_t  K\left(\frac{\varphi(t)}{x}\right)
\frac{\varphi'(t)}{\varphi(t)} \right|  \lesssim p^{-1} x^{-\delta}
\left(\frac{N_1^2}{\Gamma^{2}} + \frac{N_2}{\Gamma} \right)\;.
 \end{multline}
Combining (\ref{eq:Qnestimateproof}) and (\ref{eq:RStobeestimated}), we deduce that
\begin{align*}
\left| Q_n[\e_1](x;p)-Q_n[\e_2](x;p) - \| \e_1-\e_2\|_{\infty}  \Psi(\alpha,0;n) \right|
    \lesssim p^{-1} x^{-\delta}
\left(\frac{N_1^2}{\Gamma^{2}} + \frac{N_2}{\Gamma} \right) \| \e_1-\e_2\|_{\infty}.
\end{align*}
Taking the $\| \cdot \|_{\infty}$ norm of the above expression, we obtain the thesis. 

\end{proof}

\end{proposition}
\begin{remark}
A few comments on the above Proposition are needed.

The definition of the space $X_{z,p_0}$ is tailored to enforce the contractiveness,
with respect to the norm $L^{\infty}$, of the nonlinear operator $Q_n[\e]$,
when $p$ is sufficiently large. In fact,
\begin{itemize}
 \item The integrand $\ceili{b+pz+\e-\frac12}$ is piece-wise constant. The condition $\| \e \|_{\infty} < \frac12$
 ensures that $\e-\tilde{\e}<1$ which in turns allows us to locate the discontinuity of
 $\ceili{b+pz+\e-\frac12}-\ceili{b+pz+\tilde{\e}-\frac12}$ in terms of $\|\e-\tilde{\e}\|_{\infty}$.
 This condition could in principle be relaxed
 at the cost of a much more complicated result. However there is no point in pursuing this road, when 
 studying the well-posedness of the perturbed IBA in the large momentum limit, since we will be able to
 prove the any solution $\la$ of the perturbed IBA satisfies the a-priori estimate
 $\|\la\|_{\infty,s}\lesssim p^{-1}$ for all $ s \in \left(-\frac{1+\alpha}{2\alpha},0\right]$.
 \item We can prove the above Proposition substituting the norm $\| \cdot \|_{\infty}$ with the norm $\| \cdot \|_{\infty,s}$
for any $ s \in \left[-\frac{1+\alpha}{2\alpha},0\right]$. This makes the proof much more complicated and it is also unnecessary
for the same reason stated above.
 \item The operator $Q_n[\e]$ is discontinuous at $\e$ if
$p f'(x^*)+\e'(x^*)\leq 0$ for some $x^*$ such that $p f(x^*)+\e(x^*)+b-\frac12 \in \bb{Z}$.
 Therefore the condition $p f'(x)+\e'(x)> 0$ cannot be relaxed, if we want the operator to be contractive (hence continuous).
\end{itemize}

\end{remark}

\subsection{Functions belonging to $Z_{\delta,c}$}
We have shown that any solution of the linearised IBA equation (\ref{eq:IBAStL})
can be expressed in terms of the function $\fun$, defined in (\ref{eq:omegaground}).

After Theorem \ref{prop:fproperty}, we know that $\fun\in Z_{\frac{1+\alpha}{2\alpha},c}$
with $c=A^{\frac{1+\alpha}{2\alpha}}$, where $A$ is as per (\ref{eq:Adef}).
Here, we show that $ \fun+ p^{-1}\left(\e-\e(1)\right)\in Z_{\frac{1+\alpha}{2\alpha},c}$,
if the perturbation $\e:[1,+\infty)\to\R^+$, and its first two derivatives, are bounded with respect to
some weighted $L^{\infty}$ norm.

\begin{lemma}\label{lem:varphiperturbed} \label{cor:smallpertubation}
Let  $C>0$ and $s\leq\frac{1+\alpha}{2\alpha}$. We define the set
\begin{align} \nonumber
& \bb{D}_{s,n}(C)=\left\lbrace f\in\mc{C}^n\left([1,+\infty)\right) \,|\,
\|D^k f\|_{\infty,s}\leq C\, \mbox{ and } \right. \\ 
& \left. \exists d \in \bb{R} \,|\,
\lim_{x \to +\infty}\xi^{-\frac{1+\alpha}{2\alpha}} D^k_{\xi} f(\xi)= d \left( \frac{1+\alpha}{2\alpha} \right)^n \,,\,
\forall k\in\{0,\dots,n\}\right\rbrace\;.
\end{align}
Notice that the second condition is automatically satisfied if $s<\frac{1+\alpha}{2\alpha}$, in which case 
$d=0$.\\
(1) Let $\fun$ be the function defined by equation (\ref{eq:tauJ}).
For any $C>0$ and $s\leq\frac{1+\alpha}{2\alpha}$, there exists a $p_{s,C}>0$ such that
\begin{align}\label{eq:fepZ}
& f_{\e,p}:=\fun+p^{-1}\left(\e-\e(1)\right) \in Z_{\delta,c} \mbox{ with }
c=A^{\frac{1+\alpha}{2\alpha}}+ p^{-1} d\;,\\ \label{eq:xaDxbounded}
& \underset{\xi \in[1,+\infty)}{\inf}
    \left(\xi^{-\frac{1+\alpha}{2\alpha}}D_\xi f_{\e,p}(\xi)\right) \gtrsim 1 \;, \\
& \label{eq:phiebounds}
   N_{1,2}(f_{\e,p})\lesssim 1\;,\quad 
\Gamma(f_{\e,p}) \gtrsim 1 \;, 
\end{align}
for all $(p,\e) \in [p_{s,C},+\infty)\times \bb{D}_{s,2}(C)$.
 In the above statements $A$ is as per (\ref{eq:Adef}), and $N_1$, $N_2$ and $\Gamma$ as per (\ref{eq:NiGamma}).\\
(2) Fix furthermore $H\geq0$ and $\widehat{C} >0$, and assume that $p\geq \frac{H+\widehat{C}}{\sqrt{2\left(1+\alpha\right)}}$. Let $\bar{l}(\xi):=\bar{l}(\xi;\omega,p)$ be
 the rescaled solution of the linearised standard IBA equation, defined as per (\ref{eq:lomegaground2}),
 and $\Omega_{H,\widehat{C}}$ be as per (\ref{eq:omegaground}).\footnote{$\Omega_{H,\widehat{C}}$ is the interval of $\omega$'s such that $\bar{l}(1,\omega,p)+H \in [-C,C]$.}
 
 For any $C,\widehat{C}>0$, there exists a $p_{C,\widehat{C}}$ such that
 \begin{align}
&     \bar{f}_{\e,p}:=p^{-1}\left(\bar{l}+\e-\bar{l}(1)-\e(1)\right)\in Z_{\frac{1+\alpha}{2\alpha},c} \mbox{ with }
c=A^{\frac{1+\alpha}{2\alpha}}\left(1+ \frac{\bar{l}(1)}{p (2\alpha+2)^{\frac12}}\right)\;,
     \label{eq:perturbedZaA} \\
& \label{eq:N1N2Gbounded}
N_{1,2}\left(\bar{f}_{\e,p}\right)\lesssim 1 \;,\quad \Gamma\left(\bar{f}_{\e,p}\right) \gtrsim 1\;,
 \end{align}
for all $(\omega,p,\e)\in\Omega_{H,\widehat{C}}\times[p_{C,\widehat{C}},+\infty)\times \bb{D}_{0,2}(C)$, and
 \begin{equation}\label{eq:perturbedXpl}
 \e \in X_{\bar{f}_0,p_{C,\widehat{C}}} \;, 
 \end{equation}
 for all $(\omega,\e)\in\Omega_{H,\widehat{C}}\times\left\lbrace f\in\bb{D}_{0,1}(C)\,|\, \|f\|_\infty < \frac12 \right\rbrace$.
\end{lemma}
\begin{proof}
(1) Follows directly from the properties of the function $\fun$ and the definition of the sets $\bb{D}_{n,s}(C)$.
In particular (\ref{eq:xaDxbounded}) follows from (\ref{eq:infderf}).\\
(2) After Lemma \ref{lem:linearisedground}, $\bar{l}-\fun \in D_{2,\frac{1+\alpha}{2\alpha}}(C)$ with
$d=A^{\frac{1+\alpha}{2\alpha}} \frac{\bar{l}(1)}{ (2\alpha+2)^{\frac12}}$. The thesis then follows from (1).
\end{proof}

\subsection{Large $x$ asymptotics}
Here we use the analysis of the oscillatory integral,
to prove the following a-priori large $x$ estimate for $p$ fixed.
\begin{theorem}\label{thm:apriori}
Let $Q$ be a purely real and normalised solution of the BAE. Denote by $z$ its associated function, by $\bb{H}$ its set of hole-numbers and by
$\sec_{\bb{H}}$ the corresponding sector -- as defined per (\ref{eq:secH}).\\
(1) For all $\delta>0$
\begin{equation}\label{eq:aprioristrong}
 z(x)=x^{\frac{\alpha+1}{2\alpha}}-p \left(1+\alpha\right)-\frac{\sec_{\bb{H}}}{2}\left(\alpha-1\right)
 + O\left(x^{-\frac{\alpha+1}{2\alpha}+\delta}\right) \mbox{ as } x\to+\infty\;.
\end{equation}
(2) Assume $\sec_{\bb{H}}=0$. Let $\omega \in \bb{R}^+$ such that $\ceili{z(\omega)-\frac12}=-H$,
$l(x;\omega,p)$ be the solution of the linearised standard IBA equation (\ref{eq:IBAStL}),
and $\la(\xi)= z(\xi \omega)-l(\xi\omega;\omega,p)$. Then, for every $n\in\bb{N}$ 
$$D^n\la \in L^{\infty}_s([1,+\infty)) \;,\quad \forall s\in\left(-\frac{1+\alpha}{2\alpha},0\right]\;.$$
% In particular
%  \begin{equation}\label{eq:aprzle}
%   \| \la \|_{\infty} \leq  \left(H+\frac{1}{4}\right)\left(\alpha-1\right)\;.
%  \end{equation}
\end{theorem}
Before proving Theorem \ref{thm:apriori}, we need a preparatory Lemma
\begin{lemma}\label{prop:zetaas}
Let $z$ be a solution of the standard IBA equation, then
\begin{equation}\label{eq:Kzetaas}
\left|\bb{K}^{(n)}_\omega[\afr{z}](x)\right| \lesssim x^{-\frac{\alpha+1}{2\alpha}}.
\end{equation}
\begin{proof} 
After Lemma \ref{lem:lemma1e2}, $D_x z=\frac{\alpha+1}{2\alpha} x^{\frac{\alpha+1}{2\alpha}}+ O(x^{\e})$ as $x \to +\infty$,
for any $\e>0$. Hence, there exists $X_M\geq 0$ such that
$\underset{x\in[X_M,+\infty)}{\inf}\{x^{-\frac{\alpha+1}{2\alpha}} z'(x)\}>0$. Thus, we write
\begin{multline*}
 \bb{K}_\omega[\afr{z}](x) = \int_{\omega}^{X_M}
  K_\alpha\left(\frac{x}{y}\right) \afr{z(y)} \frac{dy}{y} + \int_{X_M}^{\infty} K_\alpha\left(\frac{x}{y}\right) \afr{z(y)} \frac{dy}{y} \notag \\
 =\int_{\omega}^{X_M}
  K_\alpha\left(\frac{x}{y}\right) \afr{z(y)} \frac{dy}{y}+ \int_{1}^{\infty} K_\alpha\left(\frac{x}{yX_M}\right) \afr{\tilde{z}(y)} \frac{dy}{y}\;,
\end{multline*}
where $\tilde{z}(\xi)=z(\xi X_M)$ for all $\xi \in [1,+\infty)$. Due to (\ref{eq:KL1}), we have
\begin{equation*}
    \int_{\omega}^{X_M}
  K_\alpha\left(\frac{x}{y}\right) \afr{z(y)} \frac{dy}{y}=O(x^{-1}) \mbox{ as } x\to+\infty \;,
\end{equation*}
while the integral
\begin{equation*}
\int_{1}^{\infty} K_\alpha\left(\frac{x}{yX_M}\right) \afr{\tilde{z}(y)} \frac{dy}{y}\;,
\end{equation*}
fulfils the hypothesis of Proposition \ref{prop:oscillatorysingle} with $p=1$ and $b=0$, since $\tilde{z}\in Z_{\frac{1+\alpha}{2\alpha}}$, whence the thesis.
\end{proof}
\end{lemma}

\begin{proof}[Proof of Theorem \ref{thm:apriori}]
(1) We notice that it is sufficient to prove the thesis assuming that the sector vanishes. In fact, if $S:= \sec_{\bb{H}}$
is the sector of $z$, then
$z-S$ is a solution of the logarithmic BAE of vanishing sector and with momentum $p^*=p-\frac{S}{2}$.
Since by assumption
$z-S=x^{\frac{\alpha+1}{2\alpha}}-(p-\frac{S}{2}) \left(1+\alpha\right)
 + O\left(x^{-\frac{\alpha+1}{2\alpha}+\e}\right)$ then
 $z=x^{\frac{\alpha+1}{2\alpha}}-p \left(1+\alpha\right)-\frac{S}{2}\left(\alpha-1\right)
 + O\left(x^{-\frac{\alpha+1}{2\alpha}+\e}\right)$.
 
 The standard IBA equation (\ref{eq:IBASt}) is an equation of the form $f=\bb{K}_{\omega}[f]+g$, for an $f
 \in L^{\infty}_{\frac{1+\alpha}{2\alpha}}$
 with
 $$g(x)=-2p+\bb{K}_{\omega}[\afr{z}](x)- H F_\alpha\left(\frac{x}{\omega}\right)+ \sum_{k=1}^H F_\alpha\left(\frac{x}{h_k}\right)\;.$$ Using Lemma \ref{prop:zetaas} and the explicit expression
 for $F_\alpha$ (\ref{eq:Ffunza}), we easily obtain that
$ g(x) + 2p \in L^{\infty}_{-\frac{1+\alpha}{2\alpha}}$. Hence,
applying Proposition \ref{prop:Kzasym}(5,6) to the above equation, we obtain the thesis. \\
(2) In this case, we reason as above but we apply the estimate of Lemma \ref{lem:Knbounded}(3) to the perturbed IBA
equation (\ref{eq:IBASte}).\\
\end{proof}

\section{Proof of the Main Theorem}
\label{sec:mainthm}
We recall that in Proposition \ref{prop:strongz} we have established a bijection between
\begin{itemize}
\item Solutions of the BAE (\ref{eq:BAEQ})
whose set of holes-numbers $\bb{H}$ has sector $\sec_{\bb{H}}=0$, namely $\bb{H}=\bb{H}_{\snu}$
for a possibly empty partition $[\nu]$ as per formula (\ref{eq:partitiontoholes}).
 \item Equivalence classes of
 strictly monotone solutions of the IBA equation (\ref{eq:IBAStGen})
satisfying the constraint (\ref{eq:IBAStka}).
\end{itemize}
We have also established a dictionary between the data of the set of hole-numbers $\bb{H}_{\snu}$
and the data that define the standard IBA equation, namely a non-negative integer $H$ and
a strictly increasing function $\sigma:\lbrace 1,\dots, H \rbrace \to \bb{Z}$ such that $\sigma(1)>-H$.

In this Section we prove our main Theorem by showing that, fixed $H$ and $\sigma$,
if $p$ is large enough there exists a unique -- up to equivalence -- strictly solution of the standard IBA
equation (\ref{eq:IBAStGen}) satisfying the constraint (\ref{eq:IBAStka}).

The proof is divided in two steps:\\
(1) \textit{Existence}: We prove that, if $p$ is large enough, for each
$\omega \in \Omega_{H,\frac14}$ -- with $\Omega_{H,C}$ as per (\ref{eq:omegaground}) --
the perturbed IBA equation (\ref{eq:IBASteGen}) admits a unique solution $(\la^*(\cdot;\omega,p),\vecmu^*(\omega,p))$
such that $\| \la^*(\cdot;\omega,p) \|_{\infty}<\frac14$, which we verify to be 
strictly monotone.
Moreover, we show that the corresponding solutions of
of the IBA equation $z(\cdot;\omega,p)=l(\cdot;\omega,p)+\la^*\left(\frac{\cdot}{\omega};\omega,p\right)$, with $\omega \in
\Omega_{H,\frac14}$ are all equivalent. \\
(2) \textit{Uniqueness}: Assume that $Q$ is a purely real and normalised solution of the BAE
whose set of hole numbers $\bb{H}_Q$ coincide with the one parameterised by the data $H$ and $\sigma$. Let
$z$ be its associated function and $\omega^*$ the point such that $z(\omega^*;\omega^*,p)=-H$.
In Proposition \ref{prop:uniqueness}, we prove that
$\omega^* \in \Omega_{H,Cp^{-1}}$, for some $C>0$, and $z$ coincides with the solution of the standard IBA constructed at step (1).

\subsection{Existence}
We denote, as customary, by $l(x;\omega,p)$, the unique normalised solution of the linearised standard IBA
(\ref{eq:IBAStL}), studied in Theorem \ref{prop:fproperty}, and $\bar{l}(\xi):=\bar{l}(\xi;p,\omega)$ 
the function $\bar{l}(\xi)=l(\xi \omega;\omega,p)$.
Recall that $\Omega_{H,\frac14}$ is the interval of $\omega$'s
such that $\left|\bar{l}(1)+H\right| \leq \frac14 $. Recall moreover that,
assuming that $p$ is large enough, $\bar{l}(\xi)$ is strictly monotone for all $\omega \in \Omega_{H,\frac14}$,
as guaranteed by Lemma \ref{cor:smallpertubation}.

Here we prove the existence of a solution of the perturbed IBA equation (\ref{eq:IBASteGen}) \footnote{or, more conveniently,
system (\ref{eq:IBASte},\ref{eq:IBAStmu2}) since
$\bar{l}(\xi)$ is strictly monotone.}. To this aim, we consider the perturbed IBA equation
as the equation for the fixed points of the following $(\omega,p)-$family of nonlinear maps
$\mc{N}=\left( \mc{N}_0, \mc{N}_1,\dots, \mc{N}_H\right)$
\begin{align}\nonumber 
& \mc{N} : \mc{C}^0([1,+\infty))\oplus \bb{R}^H 
\to  \mc{C}^0([1,+\infty)) \oplus \bb{R}^H \\ \label{eq:Ndef}
& X:=(\la,\vecmu)\mapsto \left( \mc{N}_0[X], \mc{N}_1[X],\dots, \mc{N}_H[X] \right),\\ \label{eq:N0def}
& \mc{N}_0[X](\xi):= \bb{K}_1[\la](\xi) - \bb{K}_1[\afr{\bar{l}+\la}](\xi) -
  \sum_{k=1}^{H} \left[ F_\alpha\left(\frac{\xi}{1+\mu_k}\right)- F_\alpha(\xi)\right] \;,\\ \label{eq:Nkdef}
& \mc{N}_k[X]:= \left(\frac{\bar{l}(1+\mu_k)-\bar{l}(1)}{\mu_k} \right)^{-1} \left(  
 \sigma(k)+\frac12- \bar{l}(1)- \la(1+\mu_k) \right).
\end{align}
In equation (\ref{eq:Nkdef}), $\left(\frac{\bar{l}(1+\mu_k)-\bar{l}(1)}{\mu_k} \right):= \bar{l}'(1)$ if $\mu_k=0$.

We want to prove that the above map is contractive in a \textit{neighborhood} of zero. Let us introduce a convenient metric to work with.
\begin{definition}
We denote by $X:=(\la,\vecmu)$ a point in $\mc{C}^0([1,+\infty))\oplus \bb{R}^H$. For every $\alpha>1$, we define the following norm on
  $\mc{C}^0([1,+\infty))\oplus \bb{R}^H$ 
  \begin{align}\label{eq:normC0RH}
   & \|X\|:= \| \la \|_{\infty} +2 G_{\alpha} \sum_{k=1}^H \left| \mu_k \right|\;,\\
    \label{eq:Gal}
 & G_{\alpha}:=\sup_{x\in\R^+}\left|F'_{\alpha}(x)\right|=
\begin{cases}
 \frac{1}{\pi}\sin\left(\frac{2 \pi}{1+\alpha}\right)\;,\quad 1<\alpha\leq 3 \\
\frac{1}{\pi\sin\left(\frac{2 \pi}{1+\alpha}\right)}\;,\quad \alpha>3
\end{cases}\;.
  \end{align}
 \end{definition}

 \begin{definition}\label{def:Brho}
 For any $\rho>0$, let $B_{\rho} \subset \mc{C}^0([1,+\infty))\oplus \bb{R}^H$ be the set of all
 $X:=(\la,\vecmu)$ of norm less than $\rho$, such that
  \begin{align*}
  & \mu_k \geq 0\;,\quad \forall k\in\{1,\dots,H\}\;, \\
& \la\in\mc{C}^1\left([1,+\infty)\right) \mbox{ and } \| D\la\|_{\infty} < \tilde{\rho}:= 2 \left( \Psi(\alpha,0,1)  +  \frac{ L_{\alpha}}{G_{\alpha}} \right)\rho\;,
  \end{align*}
where $\Psi(\alpha,s,n)$ is as per (\ref{eq:Psiasn}) and
\begin{align}\label{eq:Hal}
 L_{\alpha} :=\sup_{x\in\R^+}\left| K'_{\alpha}(x) \right|= 
\begin{cases}
 \frac{1}{\pi}\sin\left(\frac{2 \pi}{1+\alpha}\right)\;,\quad 1<\alpha\leq 3\\
\displaystyle{\frac{-\sin \left(\frac{2 \pi }{\alpha+1}\right) \left(2 \cos\left(\frac{2 \pi  (\alpha+3)}{3 (\alpha+1)}\right)+1\right)}
{\pi\left[2 \cos \left(\frac{2 \pi  (\alpha+3)}{3 (\alpha+1)}\right)+3-4 \cos \left(\frac{2 \pi }{\alpha+1}\right) \cos \left(\frac{\pi  (\alpha+3)}{3 (\alpha+1)}\right)\right]^2}} \;,\quad \alpha>3
\end{cases}\;.
\end{align}
 \end{definition}

 \begin{theorem}\label{thm:contraction}
  There exist positive constants $p_0,C>0$ such that\\
  (1) for all
  $(\omega,p,\rho) \in \Omega_{H,\frac14} \times [p_0,+\infty)\times \left[C p^{-1},\frac14\right]$, the map 
  $\mc{N}$ sends $B_{\rho}$ into itself and it is contractive, so that
  the map $\mc{N}$ has a unique fixed point $X^*:=\left(\la^*(\cdot;\omega,p),\vecmu^*(\omega,p)\right)$.\footnote{If
  we choose $\Omega_{H,t}$ with $0\leq t<\frac12$ instead of 
  $\Omega_{H,\frac14}$, we obtain that the map is contractive on
  $(p,\omega,\rho) \in [p_0,\infty[\times \Omega_{H,t} \times [C p^{-1},\frac12-t]$ }\\
  (2) The unique fixed point $X^*:=\left(\la^*(\cdot;\omega,p),\vecmu^*(\omega,p)\right)$
  fulfils the following estimates
  \begin{align}\label{eq:esol}
   & \|D^n\la^*\|_{\infty} \lesssim_s p^{-1} \;,\quad \forall s  \in \left(-\frac{1+\alpha}{2\alpha},0\right]\;,\quad n\in\{0,1\}\;,\\ \label{eq:musol}
   & \left| \mu_k^*-\hat{\mu}_k\right| \lesssim p^{-2} \;,\quad \forall k= 1 \dots H, \\\label{eq:mu*}
   & \hat{\mu}_k:=\frac{\sqrt{2}\alpha }{p\left(1+\alpha\right)^{\frac32}}\left(\sigma(k)+\frac12-\bar{l}(1)\right).
  \end{align}
  (3) Let 
  \begin{equation}\label{eq:zomegathmcont}
   z(\cdot;\omega,p)=\bar{l}\left(\frac{\cdot}{\omega};\omega,p\right)+
   \la^*\left(\frac{\cdot}{\omega};\omega,p\right)\;, \quad h_k(\omega,p) = \omega \left(1 +\mu_k^*(\omega,p)\right).
  \end{equation}
  The triple $\left(z(\cdot;\omega,p), \underline{h}(\omega,p),\omega\right)$ is a strictly monotone solution of the
  standard IBA equation (\ref{eq:IBAStGen}).
  
  Moreover, for every $\omega_1,\omega_2 \in \Omega_{H,\frac14}$, we have $z(x;\omega_1,p)=z(x;\omega_2,p)$ for
   all $x \geq \max\{\omega_1,\omega_2\}$.  In other words, the solutions $z(x,\omega,p)$ for $\omega \in  \Omega_{H,\frac14}$
   are all pairwise equivalent.\\
 (4) Let $x_k$ with $k \in \bb{Z}_p \setminus \bb{H}$, denote a root of the (unique up to equivalence) solution $z$.
 The following
  asymptotic holds
  \begin{align}\label{eq:xkexpexcited}
\left| x_k(p)\,p^{-\frac{2\alpha}{1+\alpha}} -
A\,\left[ 1+
\left(\frac{\sqrt{2}\alpha}{(1+\alpha)^{\frac32}}\left(k+\frac12\right)p^{-1} \right)\right] \right| \lesssim_k
p^{-2}\;,\quad \forall k\in\bb{Z}_{p}\setminus\bb{H}\;,
\end{align}
 where $A$ is as per (\ref{eq:Adef}).
 \end{theorem}

 In order to prove the Theorem we need 3 preparatory Lemmas.
\begin{lemma}\label{lem:Fa-Fb}
 Let $F_\alpha\in L^\infty\left([1,+\infty)\right)$ defined as per (\ref{eq:Ffunza}). For all $a,b \in \bb{R}^+$\\
 \begin{align}\label{eq:Fa-Fb}
  & \left\|F_{\alpha}\left(\frac{x}{a}\right)-F_{\alpha}\left(\frac{x}{b}\right)\right\|_{\infty,-1} \leq \left| a-b \right| G_{\alpha}\;,\\ \label{eq:Ka-Kb}
  & \left\|D_x F_{\alpha}\left(\frac{x}{a}\right)-D_xF_{\alpha}\left(\frac{x}{b}\right)\right\|_{\infty,-1} \leq \left| a-b \right| L_{\alpha}\;,
 \end{align}
 where $G_{\alpha},L_{\alpha}$ are as per (\ref{eq:Gal},\ref{eq:Hal}) and $\|\cdot\|_{\infty,-1}$ is the norm of $L^\infty_{-1}\left([1,+\infty)\right)$.
\begin{proof}

From the definition (\ref{eq:Ffunza}) of $F_{\alpha}$, it is easy to show that
\begin{equation*}
 F_{\alpha}\left(\frac{x}{a}\right)-F_{\alpha}\left(\frac{x}{b}\right) =  F_{\alpha}\left(\frac{b}{x}\right)-F_{\alpha}\left(\frac{a}{x}\right)\;.
\end{equation*}
Consequently, using the mean value Theorem, we get
\begin{multline*}
x \left| F_{\alpha}\left(\frac{x}{a}\right)-F_{\alpha}\left(\frac{x}{b}\right) \right| = x \left| F_{\alpha}\left(\frac{a}{x}\right)-F_{\alpha}\left(\frac{b}{x}\right) \right| \leq x \left| \frac{a}{x} - \frac{b}{x} \right|
\sup_{t\in\left[\frac{a}{x},\frac{b}{x}\right]} \left|F'_{\alpha}(t)\right| \\\leq \left| a-b \right|\sup_{t\in\R^+} \left|F'_{\alpha}(t)\right|=\left|a-b\right|G_\alpha\;.
\end{multline*}
An analogous computation yields
\begin{equation*}
 \left|D_x F_{\alpha}\left(\frac{x}{a}\right)-D_xF_{\alpha}\left(\frac{x}{b}\right)\right| \leq |a-b| \sup_{x\in\R^+} | K'_{\alpha}(x)|=|a-b|L_{\alpha}\;.
\end{equation*}
\end{proof}

\end{lemma}

\begin{lemma}\label{lem:incrration}
 Let $\vecmu,\vectmu\in\bb{R}^H$ such that $\mu_k,\tilde{\mu}_k\in[0,c]$ for all $k\in\{1,\dots,H\}$, for some $c>0$.
 There exists a $p_0>0$ such that for all $(\omega,p)\in\Omega_{H,\frac14}\times[p_0,+\infty)$ 
 \begin{align} \nonumber
& \left|\left(\frac{\bar{l}(1+\tilde\mu_k)-\bar{l}(1)}{\tilde\mu_k} \right)^{-1}  -
\left(\frac{\bar{l}(1+\mu_k)-\bar{l}(1)}{\mu_k} \right)^{-1} \right| \lesssim p^{-1} |\tilde\mu_k-\mu_k|\;, \\\label{eq:incrration1}
&  \left|\left(\frac{\bar{l}(1+\mu_k)-\bar{l}(1)}{\mu_k} \right)^{-1}\right| \lesssim p^{-1}\;,\quad \forall k\in\{1,\dots,H\}\;.
\end{align}
\begin{proof}
Let
$a=\left(\inf_{\xi \in [1,1+c]} | \bar{l}'(\xi) |\right)^{-1}$ and $b=\sup_{\xi \in [1,1+c]} | \bar{l}''(\xi) |$.

The mean-value Theorem yields the following estimates
\begin{align*}
& \left|\left(\frac{\bar{l}(1+\mu_k)-\bar{l}(1)}{\mu_k} \right)^{-1}\right| \leq a\;,\\
& \left|\left(\frac{\bar{l}(1+\tilde\mu_k)-\bar{l}(1)}{\tilde\mu_k} \right)^{-1}  -
\left(\frac{\bar{l}(1+\mu_k)-\bar{l}(1)}{\mu_k} \right)^{-1} \right| \leq   
a \left(1-\left(1+ ab |\mu_k-\tilde\mu_k|\right)^{-1}\right)\;.
\end{align*}

After Lemma \ref{cor:smallpertubation}, we have that $a\lesssim p^{-1}$ and
$b\lesssim p$, whence the thesis is proven.
\end{proof}

\end{lemma}

\begin{lemma}\label{lem:nlestimates}
Let $X,\widetilde{X}\in\mc{C}^0([1,+\infty))\oplus \R^H$ with
$X:=(\la,\vecmu)$ and $\widetilde{X}:=(\tilde\la,\vectmu)$. There exist $p_0,C >0$ such that for any
$(\omega,p,\rho)\in\Omega_{H,\frac14}\times[p_0,+\infty)\times \left[C p^{-1},\frac14\right)$ we have
\begin{align}\label{eq:nltobep1}
 &\text{(1)} \quad\|\mc{N}_0[(0,\underline{0})]\|_\infty \lesssim p^{-1}\;, \\ \nonumber
 &\text{(2)}\quad \left| \left\|\mc{N}_0[X]-\mc{N}_0[\widetilde{X}]\right\|_{\infty}-
 \frac{\alpha-1}{\alpha+1} \| \la-\tilde{\la} \|_{\infty} - G_{\alpha} \sum_{k=1}^H | \mu_k-\tilde\mu_k | \right|
 \lesssim p^{-1} \| \la-\tilde{\la} \|_{\infty}  \;,\\ \label{eq:nltobep2}
 &\\ \nonumber
 &\text{(3)}\quad \left|\mc{N}_k[X]-\mc{N}_k[\widetilde{X}]\right| \lesssim p^{-1} \left[
   \| \la-\tilde{\la}\|_{\infty} +\left(\rho +\tilde{\rho}+ \sigma(H)+H+\frac14\right)|\mu_k -\tilde\mu_k| \right]\;,\\
   \label{eq:nltobep2bis}
 &\\
   \label{eq:Nk0mu*}
&\text{(4)} \quad \left| \mc{N}_k[(0,\underline{\hat\mu})]-\hat\mu_k \right| \lesssim p^{-2}\;,\quad \forall k\in\{1,\dots,H\}\;,
\\ \label{eq:nltobep3}
 &\text{(5)} \quad  \| D\mc{N}_0[X] \|_{\infty} < \left( 2 \Psi(\alpha,0,1) + 
 \frac{L_{\alpha}}{G_{\alpha}}\right) \rho :=\tilde{\rho},
 \end{align}
 where $\Psi(\alpha,0,n)$ is as per (\ref{eq:Psiasn}).\\
\begin{proof}
This Lemma is a direct application of Propositions \ref{prop:oscillatorysingle} and \ref{prop:differencenonlinear}, and of Lemma
\ref{cor:smallpertubation}, i.e. of all the theory
that we have developed in this paper.
Recall that due to Lemma \ref{cor:smallpertubation}, there exists a $p_0>0$ such that
for all $(\omega,p)\in \Omega_{H,\frac14}\times [p_0,+\infty)$ we have
$p^{-1}\left(\bar{l}-\bar{l}(1)\right) \in Z_{\frac{1+\alpha}{2\alpha},c}$ with $c=A^{\frac{1+\alpha}{2\alpha}}\left(1+\frac{\bar{l}(1)}{p\sqrt{2\left(1+\alpha\right)}}\right)$ and $A$ as per
(\ref{eq:Adef}). Moreover, the quantities $N_1$, $N_2$ and $\Gamma$ defined as per (\ref{eq:NiGamma}) can be chosen independently of $(\omega,p)$, namely
 \begin{equation*}
N_1\left(p^{-1}\left(\bar{l}-\bar{l}(1)\right)\right) \lesssim 1\;,\quad N_2\left(p^{-1}\left(\bar{l}-\bar{l}(1)\right)\right) \lesssim 1 \;,\quad \Gamma\left(p^{-1}\left(\bar{l}-\bar{l}(1)\right)\right) \gtrsim 1 \;.
 \end{equation*}
(1) Equation (\ref{eq:nltobep1}) follows directly from Proposition \ref{prop:oscillatorysingle}.\\
(2) Simple arithmetic yields 
\begin{multline*}
\mc{N}_0[X](\xi)-\mc{N}_0[\widetilde{X}](\xi) =\bb{K}_1\left[\ceili{\bar{l}+\la}\right](\xi)-\bb{K}_1\left[\ceili{\bar{l}+\tilde{\la}}\right](\xi) \\
+ \sum_{k=1}^H \left[ F_{\alpha}\left(\frac{\xi}{1+\tilde\mu_k}\right) - F_{\alpha}\left(\frac{\xi}{1+\mu_k}\right)\right]\;.
\end{multline*}
The norm of the right hand-side of the above identity can be estimated using Proposition \ref{prop:differencenonlinear} and (\ref{eq:Fa-Fb}).
In fact, due to Proposition \ref{prop:differencenonlinear},
\begin{equation*}
 \left\|\bb{K}_1\left[\ceili{\bar{l}+\frac12+\la}\right]-\bb{K}_1\left[\ceili{\bar{l}+\frac12+\tilde{\la}}\right] \right\|_{\infty}
 \leq \left(\Phi(\alpha,0) +\frac{C}{p}\right)\| \la-\tilde{\la} \|_{\infty}\;.
\end{equation*}
Moreover, due to (\ref{eq:Fa-Fb}),
\begin{equation*}
 \left\|\sum_{k=1}^H \left[F_{\alpha}\left(\frac{\cdot}{1+\tilde\mu_k}\right) - F_{\alpha}\left(\frac{\cdot}{1+\mu_k}\right) \right]\right\|_{\infty}
 \leq G_{\alpha}\sum_{k=1}^H|\mu_k-\tilde\mu_k|.
\end{equation*}
Combining the latter two estimates and using the fact that $\Phi(\alpha,0)=\frac{\alpha-1}{\alpha+1}$ (see (\ref{eq:completeintegral})) we obtain (\ref{eq:nltobep2}).\\
(3) We begin with the decomposition
  \begin{multline*}
    \mc{N}_k[X]-\mc{N}_k[\widetilde{X}]=\left(\frac{\bar{l}(1+\tilde\mu_k)-\bar{l}(1)}{\tilde\mu_k} \right)^{-1}
   \left( \tilde{\la}(1+\tilde\mu_k)-\la(1+\tilde\mu_k)\right)
   \\ 
    +\left(\frac{\bar{l}(1+\tilde\mu_k)-\bar{l}(1)}{\tilde\mu_k} \right)^{-1}  \left(\la(1+\tilde\mu_k)-\la(1+\mu_k)\right)  \\
    + \left[\left(\frac{\bar{l}(1+\tilde\mu_k)-\bar{l}(1)}{\tilde\mu_k} \right)^{-1}  -
\left(\frac{\bar{l}(1+\mu_k)-\bar{l}(1)}{\mu_k} \right)^{-1} \right] \left( \la(1+\mu_k) +\bar{l}(1) - \sigma(k)-\frac{1}{2}\right).
  \end{multline*}
  We also notice that
  \begin{equation}\nonumber
   \left|\la(1+\tilde\mu_k)-\la(1+\mu_k)\right|\leq  \left|\mu_k-\tilde{\mu}_k\right| \sup_{\xi\in[1,+\infty)}\la'(\xi)\leq
  \tilde{\rho} \left|\mu_k-\tilde{\mu}_k\right|
  \end{equation}
Combining the two latter estimates and using Lemma \ref{lem:incrration}, we obtain equation (\ref{eq:nltobep2bis}).\\
(4)
Equation (\ref{eq:Nk0mu*}) follows directly from (\ref{eq:1+deltap}) of Lemma \ref{lem:linearisedground}.\\
(5) 
Simple arithmetic yields
\begin{multline*}
D_\xi\mc{N}_0[X](\xi)-D_\xi\mc{N}_0[(0,\underline{0})](\xi) =\bb{K}_1^{(1)}\left[\ceili{\bar{l}+\la}\right](\xi)-\bb{K}_1^{(1)}\left[\ceili{\bar{l}}\right](\xi) \\
+ \sum_{k=1}^H \left[ F_{\alpha}\left(\xi\right) - F_{\alpha}\left(\frac{\xi}{1+\mu_k}\right)\right]\;,
\end{multline*}
whence, after Proposition \ref{prop:differencenonlinear} and Lemma \ref{lem:Fa-Fb},
it follows that there exist $C',p_0>0$ such that, for all $p\geq p_0'$
\begin{multline*}
\|D\mc{N}_0[X] -D\mc{N}_0[(0,\underline{0})]\|_{\infty}\leq \sum_{k=1}^H\left\|D F_\alpha\left(\frac{\cdot}{1+\mu_k}\right)-D F_\alpha\right\|_\infty\\
+\left\|\bb{K}_1^{(1)}\left[\ceili{\bar{l}+\la}\right]-
\bb{K}_1^{(1)}\left[\ceili{\bar{l}}\right]\right\|_\infty
\leq \left(\Psi(\alpha,0,1) + \frac{C'}{p}\right)\| \la \|_{\infty} + L_\alpha\sum_{k=1}^H|\mu_k|\;.
\end{multline*}
Moreover, from Proposition \ref{prop:oscillatorysingle}, it follows that 
 \begin{equation*}
\|D \mc{N}_0[(0,\underline{0})]\|_{\infty} \lesssim p^{-1} 
\end{equation*}
Hence, using the triangle inequality, we deduce that there exists a $C'>0$ such that
\begin{equation*}
   \|D\mc{N}_0[X]\|_{\infty} < \rho \left(  \Psi(\alpha,0,1) +   \frac{L_{\alpha}}{G_{\alpha}}\right)+
   C'\,p^{-1},
\end{equation*}
from which the thesis follows.

\end{proof}

\end{lemma}

The proof of Theorem \ref{thm:contraction} is essentially a corollary of Lemma \ref{lem:nlestimates}.
 \begin{proof}[Proof of Theorem \ref{thm:contraction}]
 (1) After (\ref{eq:nltobep2},\ref{eq:nltobep2bis}), we  have that there exists a $\widetilde{C}>0$ such that
 \begin{equation*}
  \left\|\mc{N}[X]-\mc{N}[\widetilde{X}]\right\|\leq  \left(\frac{\alpha-1}{\alpha+1} + \frac{\widetilde{C}}{p}\right)\| \la-\tilde{\la} \|_{\infty}  + \frac{\widetilde{C}}{p} \left(\rho +\tilde{\rho}\right)
   |  \mu_k -\tilde\mu_k|\;.
 \end{equation*}
 Since $\rho,\tilde{\rho}$ are bounded, it follows that there exists another $\widetilde{C}>0$ such that for all $\rho<\frac12$,
 \begin{equation}\label{eq:NXYth}
\left\|\mc{N}[X]-\mc{N}[\widetilde{X}]\right\|\leq \left( \max\left\{\frac{\alpha-1}{\alpha+1},\frac12\right\}+ \frac{C}{p} \right) \left\|X-\widetilde{X}\right\| \;,\quad\forall X,\widetilde{X}\in B_\rho \;.
 \end{equation}
 Since $\frac{\alpha-1}{\alpha+1}<1$, the map $\mc{N}$ is contractive on $B_{\rho}$ for every $\rho <\frac12$.
 
 We now have to prove that there exists a $C>0$ such that $\mc{N}\left[B_{\rho}\right] \subset B_{\rho}$ holds
 for all $\rho$ such that $Cp^{-1}\leq\rho<\frac14$. The proof goes as follows:
 \begin{itemize}
 \item Using (\ref{eq:nltobep1},\ref{eq:NXYth}), if $\rho<\frac14$ we have
\begin{equation*}
 \|\mc{N}[X]\|\leq \frac{\widetilde{C}}{p}+ \left( \max\left\{\frac{\alpha-1}{\alpha+1},\frac14\right\}+\frac{\widetilde{C}}{p} \right) \| X\| \;.
\end{equation*}
Hence $\|\mc{N}[X]\|<\rho$ for all $X \in B_{\rho}$ if
\begin{equation*}
  \frac{\widetilde{C}}{p}+ \left( \max\left\{\frac{\alpha-1}{\alpha+1},\frac14\right\}+\frac{\widetilde{C}}{p} \right) \rho <\rho \;,
\end{equation*}
It follows that there exists another $C'>0$ such that the above inequality is satisfied whenever $\rho \geq C p^{-1}$.
 \item By construction, we have that $\sigma(k)+\frac12+H \geq \frac12$.
 Moreover, $\bar{l}(1)+H \leq\frac14 $, by definition of $\Omega_{H,\frac{1}{4}}$.
 Finally $\la(1+\mu_k)\leq-\frac14$, by definition of $B_{\frac14}$. Hence 
 $\sigma(k)+\frac12- \bar{l}(1)- \la(1+\mu_k) \geq 0$. Finally $\left(\frac{\bar{l}(1+\mu_k)-\bar{l}(1)}{\mu_k} \right)^{-1}>0$
 for the hypothesis on $p$. It follows that $\mc{N}_k[X]\geq 0$ for all $X \in B_{\rho}$.
  \item Due to (\ref{eq:nltobep3}),
 the inequality $\|D\mc{N}_0[X] \|_{\infty} <\tilde{\rho}$ is satisfied if $\rho \in [C p^{-1},\frac14[$
 for some $C>0$. 
 \end{itemize}
(2) Because of (1) the map $\mc{N}$ has a unique fixed point $X^*=(\la^*(\cdot;\omega,p),\vecmu^*(\omega,p))$ for all
$X\in B_{\rho}$, with $ Cp^{-1}\leq \rho <\frac12$, and this belongs to the intersection of all these sets which coincide
with $B_{Cp^{-1}}$.

It means that $\|\la^*(\cdot;\omega,p)\|_{\infty}$, $\| D\la^*(\cdot;\omega,p)\|_{\infty}$ and $|\mu_k^*|$ are $O(p^{-1})$.

In order to prove equation (\ref{eq:esol}), we notice that 
$\la(\xi,\omega)$ satisfies the equation (\ref{eq:IBASte})
\begin{align}
 &\la(\xi)-\bb{K}_1[\la](\xi)= g(\xi)\;,\notag\\
 &g(\xi)=- \bb{K}_1\left[\afr{\bar{l}+\la}\right](\xi) -
  \sum_{k=1}^{H} \left( F_\alpha\left(\frac{\xi}{1+\mu_k}\right)- F_\alpha(\xi)\right) .
  \label{eq:pertIBAproof}
\end{align}
Using the estimate on $\|\bb{K}_1[\bar{l}+\la]\|_{\infty,-\frac{1+\alpha}{2\alpha}}$ in Proposition
\ref{prop:oscillatorysingle} and the estimate on $\left\|F_\alpha(\cdot)-F_\alpha\left(\frac{\cdot}{1+\mu_k}\right)\right\|_{\infty,-1}$
in Lemma \ref{lem:Fa-Fb}, we have that the forcing term $g$ satisfies the estimate
\begin{equation*}
    \|g\|_{\infty,-\frac{1+\alpha}{2\alpha}} \lesssim p^{-1}\;,
\end{equation*}
Now (\ref{eq:esol}) with $n=0$ follows from the fact that 
the operator $\bb{I}-\bb{K}_1$ is invertible on $L^{\infty}_{s}\left([1,+\infty)\right)$ with $-\frac{1+\alpha}{2\alpha}<s\leq 0$, as per Proposition
\ref{prop:Kzasym} (4). The case $n=1$ is proven applying the Euler operator $D$ to both side of the equality (\ref{eq:pertIBAproof}) and
following the same considerations.

To prove (\ref{eq:musol}),
we use (\ref{eq:nltobep2bis},\ref{eq:Nk0mu*}):
 \begin{multline*}
   \left|\mu_k^*-\hat\mu_k\right|\leq
   \left|\mc{N}_k[(\la,\vecmu^*)]-\mc{N}_k[(\la,\hat\vecmu)]\right|+ 
   \left|\mc{N}_k[(\la,\hat\vecmu)-\mc{N}_k[(0,\hat\vecmu)]\right|\\ +\left|\mc{N}_k[(0,\hat\vecmu)]-\hat\mu_k \right| \lesssim p^{-2}.
  \end{multline*}

(3) For every $\omega \in \Omega_{H,\frac14}$, $\la^*$ satisfies
 the constraints (\ref{eq:IBAekappa},\ref{eq:IBAeine}). In fact, by hypothesis on $\omega$, $|\bar{l}(1)+H|\leq \frac14$,
the constraint (\ref{eq:IBAekappa}) follows from estimate (\ref{eq:esol}) with $n=0$;
the constraint (\ref{eq:IBAeine}) follows from the estimate (\ref{eq:esol}) with $n=1$, together with Lemma \ref{cor:smallpertubation}.

Since $\la^*$ satisfies
 the constraints (\ref{eq:IBAekappa},\ref{eq:IBAeine}),
 then the triple $\left(z(\cdot;\omega,p), \underline{h}(\omega,p),\omega\right)$ is a strictly monotone solution of the
standard IBA equation (\ref{eq:IBAStGen}) which satisfies the constraint (\ref{eq:IBAStka}). 

Now let $\omega_1,\omega_2 \in \Omega_{H,\frac14}$ such that $\omega_2\geq \omega_1$ and define
$\tilde{X}(\omega_2):=(\tilde{\la}(\xi),\tilde{\vecmu})$, where
$\tilde{\la}(\xi)=z(\xi \omega_2,\omega_1)-\bar{l}(\xi,\omega_2,p)$
and $\tilde{\mu}_k=\mu_k(\omega_1)+\frac{\omega_1}{\omega_2}-1$. 
By construction $\tilde{X}(\omega_2)$ is a fixed point of the map $\mc{N}$, with $\omega=\omega_2$.
Moreover a simple computation yields that $\|\tilde{X}_{\omega_2}\|_{\infty} \lesssim p^{-1}$.
Since, after 1., the map has a unique small fixed point, then
$\tilde{\la}=\la(\xi,\omega_2,p)$. The thesis follows.\\
(4) Equation (\ref{eq:xkexpexcited}) follows from the expression of $\omega_H$
(\ref{eq:omegagroundlim}) and the equation (\ref{eq:1+deltap}).
 
\end{proof}

\subsection{Uniqueness}

With the same notation of Theorem \ref{thm:contraction}.
\begin{proposition}
\label{prop:uniqueness}
Fix a non-negative integer $H$ and
a strictly increasing function $\sigma:\lbrace 1,\dots, H \rbrace \to \bb{Z}$ such that $\sigma(1)>-H$.

Assume that $(z,\underline{h},\omega^*)$ is a strictly monotone
solution of the standard IBA equation (\ref{eq:IBAStGen}) such that $z(\omega^*)=-H$.
Let  $l(x,\omega^*,p)$ be the unique normalised solution of the linearised IBA equation (\ref{eq:IBAStL}),
studied in Theorem \ref{prop:fproperty}.
Define $ \la(\xi)=z(\xi \omega^*)-l(\xi \omega^*,\omega^*,p)$,
and $\mu_k= \frac{h_k}{\omega^*}-1$.

There exist $p^*,C>0$ such that for all $p\geq p^*$
\begin{equation}
 \|(\la,\vecmu)\| \in B_{C p^{-1}} \mbox{ and } \omega^* \in \Omega_{H,C p^{-1}},
\end{equation}
where the norm $\|\cdot \|$ is as per (\ref{eq:normC0RH}), and $\Omega_{H,C}$ as per 
(\ref{eq:omegaground}).
Whence,
$z$ coincides with the solution of the standard IBA
equation obtained via the contraction Theorem \ref{thm:contraction}.
\begin{proof}
We use the notation $\bar{l}(\xi):=\bar{l}(\xi,\omega^*,p)$.

Since $\la$ satisfies the perturbed IBA equation
(\ref{eq:IBASte}), after Lemma (\ref{lem:apriorie}) it fulfils the a-priori estimate
\begin{align}\label{eq:eDeproof} \tag{a}
 \|D^n\la\|_{\infty} \lesssim_n 1 \;,\quad \forall n\in\bb{N}\;.
\end{align}
Restricting to the case $n=0$ and using the fact that $z(\omega^*)=-H$, the latter estimate
implies that
\begin{equation*}
 \left| \bar{l}(1)+ H \right| \lesssim 1 \;.
\end{equation*}
Therefore by Lemma \ref{lem:linearisedground},
$\omega^* \in \Omega_{H,C}$ for some $C>0$ (independent on $p$).
Hence, after Lemma \ref{cor:smallpertubation} and equation (\ref{eq:eDeproof}),
we have that 
\begin{equation*}
   \bar{f}_{\la,p}:= p^{-1}\left(\bar{l}-\bar{l}(1)+\la-\la(1)\right) \in Z_{\frac{1+\alpha}{2\alpha},c} \;,
   \mbox{ with } \left|c A^{-\frac{1+\alpha}{2\alpha}}-1 \right| \lesssim p^{-1}\;,
\end{equation*}
and 
\begin{equation*}
    N_1\left(\bar{f}_{\la,p}\right),N_2\left(\bar{f}_{\la,p}\right)\lesssim 1\;,\quad \Gamma\left(\bar{f}_{\la,p}\right)\gtrsim 1 \;,
\end{equation*}
In the above equations $Z_{\frac{1+\alpha}{2\alpha},c}$ is as per Definition \ref{def:ZdA}, $A$ as per (\ref{eq:Adef}), and
$N_1$, $N_2$ and $\Gamma$ as per (\ref{eq:NiGamma}). Therefore, after
Lemma \ref{cor:smallpertubation} and Proposition \ref{prop:oscillatorysingle}, we have that
 \begin{align}\label{eq:infderproof} \tag{b}
 & \left|\left(\frac{\bar{l}(1+\mu_k)-\bar{l}(1)}{\mu_k} \right)^{-1}\right|\lesssim p^{-1}\;,\quad \forall k\in\{1,\dots,H\}\;,\\
 \label{eq:oscproof} \tag{c}
 &\| \bb{K}_1[\afr{\bar{l}+\la}] \|_{\infty} \lesssim p^{-1}\;,\quad \| \bb{K}_1^{(1)}[\afr{\bar{l}+\la}]\|_{\infty} \lesssim p^{-1} \;.
\end{align}
Using the equation (\ref{eq:Nkdef}) for $\mu_k$, namely $\mu_k=\mc{N}_k[(\la,\vecmu)]$, and the
estimate (\ref{eq:infderproof}), we obtain that
\begin{equation}\label{eq:muklastproof} \tag{d}
\left|\mu_{k}\right| \lesssim p^{-1} \;,\quad \forall k\in\{1,\dots,H\}\;.
\end{equation}
 The latter estimate, combined with Lemma \ref{lem:Fa-Fb}, implies that
\begin{align}\label{eq:FamFa1proof} \tag{e}
& \left\| F_{\alpha}\left(\frac{\cdot}{1+\mu_k}\right)-
 F_{\alpha}\left(\cdot\right)\right\|_{\infty} \lesssim p^{-1}\;,\\ \label{eq:KamKa1proof} \tag{f}
 & \left\| K_{\alpha}\left(\frac{\cdot}{1+\mu_k}\right)-
 K_{\alpha}\left(\cdot\right)\right\|_{\infty} \lesssim p^{-1} .
\end{align}
Finally, using (\ref{eq:IBASte}), estimates (\ref{eq:oscproof},\ref{eq:FamFa1proof},\ref{eq:KamKa1proof})
 imply that
\begin{equation}\label{eq:edelastproof} \tag{g}
 \|\la\|_{\infty} \lesssim p^{-1}\;,\quad
\|D\la\|_{\infty} \lesssim p^{-1}\;.
\end{equation}
The estimate (\ref{eq:edelastproof}), combined with the relation
$z(\omega^*)=-H$, implies that
$$|\bar{l}(1,\omega^*,p)+H| \lesssim p^{-1} \Longrightarrow \omega^* \in \Omega_{H,C p^{-1}} \mbox{ for some } C>0 ,$$ 
with $C$ independent on $p$.\\
The estimates (\ref{eq:muklastproof}) and (\ref{eq:edelastproof}) imply that
$(\la,\vecmu)\in B_{Cp^{-1}}$ for some $C>0$ (independent on $p$).

The thesis is proven.
\end{proof}
\end{proposition}

The main Theorem of our paper is now a corollary of Theorem \ref{thm:contraction} and Proposition \ref{prop:uniqueness}.
\begin{theorem}\label{cor:main}
Let $N$ be a non-negative integer number and $[\nu]$ a partition of $N$.
 If $p$ is large enough there exists a unique purely real and normalised solution $Q$ of the BAE such that $\bb{H}_Q=\bb{H}_{\snu}$.
 
 Let $\omega_H$ be as per (\ref{eq:omegagroundlim}), $l(x):=l(x;\omega_H,p)$ be the corresponding
 solution of the linearised IBA equation (\ref{eq:IBAStL}) and
 $\hat{x}_k=l^{-1}(k+\frac12)$, with $k\in\bb{Z}_p\setminus\bb{H}_{\snu}$.
  The roots $\{x_k\}_{k\in\bb{Z}_p\setminus\bb{H}_{\snu}}$ of $Q$ fulfil the following uniform estimate
  \begin{align}\label{eq:xkfinalest}
\left| \frac{x_k(p)}{ \hat{x}_k(p)}-1 \right| \lesssim C(k/p)
p^{-2}\;,\quad \forall k\in\bb{Z}_{p}\setminus\bb{}H_{\snu}\;,
\end{align}
where $C:\R^+\to \R^+$ is a bounded function such that $C(x)=O\big(x^{-2+\e}\big)$ as $x \to +\infty$, for all
$\e>0$.
Moreover, fixed $k$, the following non-uniform estimate holds
 \begin{align}\label{eq:xkfinal}
\left| x_k(p) \, p^{-\frac{2\alpha}{1+\alpha}} -
A\left[ 1+
\left(\frac{\sqrt{2}\alpha}{(1+\alpha)^{\frac32}}\left(k+\frac12\right)p^{-1} \right)\right] \right| \lesssim_k
p^{-2}\;,\quad \forall k\in\bb{Z}_{p}\setminus\bb{H}\;.
\end{align}
\end{theorem}
\begin{proof}
(1) After Proposition \ref{prop:strongz}, solutions of the BAE (\ref{eq:BAEQ}) whose set of holes number has sector $0$ are in bijection with
strictly monotone solutions of the standard IBA equation satisfying the constraint (\ref{eq:IBAStka}).
The existence of a solution of the latter problem was established in Theorem \ref{thm:contraction}, its uniqueness
in Proposition \ref{prop:uniqueness}.
(2) Now we prove the estimate (\ref{eq:xkfinalest}).
We use the notation $\hat{\xi}_k(p)=\frac{\hat{x}_k(p)}{ \omega_H}$, $\xi_k(p)=\hat{\xi}_k(p)\big(1+\delta_k(p)\big)$.
We let $\bar{l}(\xi):=l(\xi \omega_H;\omega_H,p)$ and $\la\left(\xi\right):=z(\xi \omega_H)-\bar{l}(\xi)$.
After Lemma \ref{lem:linearisedground} (2)(i) we have that
\begin{equation}\label{eq:barlpproof} \tag{a}
 \inf_{\xi\geq 1} \left(\xi^{1-\frac{\alpha+1}{2\alpha}}\bar{l}'(\xi)\right) \gtrsim p\;,
\end{equation}
and after Theorem \ref{thm:contraction} (2), we have that
\begin{equation}\label{eq:lasproof} \tag{b}
 \|\la\|_{\infty,-\frac{1+\alpha}{2\alpha}+\e} \lesssim_\e p^{-1}\;,\quad \forall  \e \in \left(0,\frac{1+\alpha}{\alpha}\right].
\end{equation}
Using the mean value Theorem, together with (\ref{eq:lasproof}) and (\ref{eq:barlpproof}), we obtain that
\begin{equation*}
     \left|\delta_k(p)\right|=\left|\frac{\la\left(\hat{\xi}_k (1+\delta_k) \right)}{\xi_k \bar{l}' \left(\hat{\xi}_k (1+\delta'_k )\right)}\right|
     \lesssim_{\e} p^{-2} \hat{\xi}_k^{-\frac{1+\alpha}{\alpha}+\e} \;,
\end{equation*}
where in the above equation $\delta'_k$ is a real number such that $|\delta'_k|\leq |\delta_k|$.
Since $\xi_k\geq 1$ for every $k$ and, when $k/p$ is large $\hat{\xi}_k^{\frac{1+\alpha}{2\alpha}} \sim \frac{k+\frac12}{p}$,
we obtain the thesis.\\
(3) Estimate (\ref{eq:xkfinal}) was already proven in Theorem \ref{thm:contraction} (5).
\end{proof}

\section{On the ODE/IM correspondence}
\label{sec:ODEIMcorr}
In this Section we briefly introduce the ODE/IM correspondence for the Quantum KdV/Monster potentials; the interested reader
can read more details in \cite{doreyreview}, which is the standard reference
for early results in this theory, or in \cite{BLZ04} or in \cite{coma20}.

\subsection{Quantum KdV}
The Quantum KdV model was constructed by Bazhanov-Lukyanov-Zamolodchikov in \cite{baluzaI,baluzaII}
as an example of an infinite dimensional quantum theory integrable by the Bethe Ansatz.
In the free-field representations,
we start with the Fock representation $\mc{F}_{p}$
of the Heisenberg algebra with quasi-momentum $p$ and Planck constant
$\beta^2/2$ ($p\geq0$ and $\beta>0$). This induces
a (generically irreducible) Virasoro module $\mc{V}_{c,\Delta}$ with central charge $c:=c(\beta,p)$
and highest weight $\Delta:=\Delta(\beta,p)$, given by the expression
\begin{equation}\label{eq:cdelta}
 c=13-6\left(\beta^{2}+\beta^{-2}\right)\;,\quad \Delta=\frac12+\frac{p^2}{\beta^2}-\frac{1}{4}\left(\beta^{2}+\beta^{-2}\right)\;.
\end{equation}
The Hilbert space has further decomposition in spaces of level $N$
\begin{equation}
 \mc{V}_{c,\Delta}=\underset{N\geq 0}{\oplus}\mc{V}_{c,\Delta}^{(N)}\;,
\end{equation}
where  $\dim \mc{V}_{c,\Delta}^{(N)}$ is the number of integer partitions of $N$.

The level subspaces support the action of the operator-valued function (Q-transfer matrix) $\mc{Q}_+(x)$
\begin{equation}
 \mc{Q}_+(x): \mc{V}_{c,\Delta}^{(N)} \to \mc{V}_{c,\Delta}^{(N)}\;,
\end{equation}
whose eigenvalues $Q(x)$ satisfy the BAE (\ref{eq:BAEQ}),
provided
\begin{equation}\label{eq:alphabeta}
 \alpha=\beta^{-2}-1\;.
\end{equation}
Mathematically speaking, not much is known about
the eigenvalues of the operator-valued-function $\mc{Q}_+(x)$, however, as discussed in
\cite{BLZ04}, they should satisfy the following properties:
\begin{itemize}
 \item $Q(x)$ is a real entire function of $x$, of order $\frac{1+\alpha}{2\alpha}$, whose
 zeroes are almost all simple, real and positive, and it satisfies the normalisation (see \cite{BLZ04} eq. (A.5));
 \begin{align*}
&\lim_{x \to +\infty} x^{-\frac{1+\alpha}{2\alpha}}n_Q(x)=1,
 \end{align*}
 where by $n_Q$ we denote the counting function of $Q$, as per (\ref{eq:counting});
\item The zeroes of $Q$ are simple, real and positive if $2p\geq N+\frac12$.
\end{itemize}

\subsection{Monster potentials}
Let us now consider the Schr\"odinger equation
\begin{equation}\label{eq:schrintro}
\psi''(t)=\left(V(t)-E\right)\psi(t) \;,\quad E \in \C \;.
\end{equation}
We say that $V$ is a monster potential if it has the following form
\begin{equation}\label{eq:V}
V(t)=\frac{\ell(\ell+1)}{t^2}+t^{2\alpha}-2 \frac{d^2}{d t^2} \sum_{k=1}^N\,  \log\left(t^{2\alpha+2}-z_k\right) \;,
\end{equation}
In the above formula, $\{z_1,\dots,z_N\}$ are non-zero and pairwise distinct complex numbers chosen in such a way that
the monodromy about each value of $t$ solution of $t^{2\alpha+2}=z_k$ is trivial for every value of $E \in \bb{C}$
(in other words, all these singularities are apparent).
The latter requirement is equivalent to the following system of $N$ algebraic equations for the $N$ complex
unknowns $\{z_1,\dots,z_N\}$,
\begin{multline}\label{eq:algblz}
\sum_{j \neq k} \frac{z_k\left( z_k^2 +(3+\alpha)(1+2\alpha)z_k z_j+
\alpha (1+2 \alpha)z_j^2 \right)}{ (z_k-z_j)^3}- \frac{\alpha z_k}{ 4 (1+\alpha)} \\ 
= - \frac{4\ell\left(\ell+1\right)+1-4\alpha^2}{16\left(\alpha+1\right)} \;,\quad k\in\{1,\dots,N\}\;.
\end{multline}
The monster potential, as well as the algebraic system (\ref{eq:algblz}),
were introduced by Bazhanov-Lukyanov-Zamolodchikov in \cite{BLZ04} (see also \cite{fioravanti05})
to generalise the ODE/IM conjecture of Dorey-Tateo (corresponding to the case $N=\ell=0$) \cite{dorey98}.
In fact, the spectral determinant $D(E)$ of the central connection problem for the equation (\ref{eq:schrintro}) satisfies
the following properties:
\begin{itemize}
 \item $D(E)$ is an entire function of $E$, of order $\frac{1+\alpha}{2\alpha}$
 satisfying the normalisation (see \cite{BLZ04} eq. (22)) 
  \begin{align}
&\lim_{E \to +\infty} E^{-\frac{1+\alpha}{2\alpha}}n_D(E)=
\left(\frac{2\sqrt{\pi}\,\Gamma\left(\frac32+\frac{1}{2\alpha}\right)}
{\Gamma\left(1+\frac{1}{2\alpha}\right)}\right)^{\frac{2\alpha}{1+\alpha}} : =  \eta\;,\label{eq:etadef}
 \end{align}
 where $n_D$ is the counting function of $D$.
 \item $D(E)$ satisfies the BAE (\ref{eq:BAEQ}), provided
 \begin{equation}
\label{eq:ell}
p=\frac{2\ell+1}{\alpha+1}\;.
\end{equation}
 \end{itemize}
In our previous paper \cite{coma20} we studied the monster potentials in the large $\ell$ limit and we obtained a complete
classification -- up to the technical
Conjecture 5.9, on the existence of a certain Puiseaux series, that we were not able to prove in full generality.
To illustrate our results, we must introduce a nice mathematical object known as Wronskian of Hermite polynomials, which is associated
to any partition $[\nu]=(\nu_1,\dots,\nu_H)$ of $N$. It is the following polynomial of degree $N$
\begin{equation}\label{eq:Pnu}
P^{\snu}(t)=\mbox{Wr}[H_{\nu_H}(t),H_{\nu_{H-1}+1}(t),\dots,H_{\nu_1+H-1}(t)] \;,
\end{equation}
where $\displaystyle{H_{n}(t)=(-1)^ne^{\frac{t^2}{2}}\frac{d^n}{dt^n}e^{-\frac{t^2}{2}}}$
is the $n-$th Hermite polynomial and $\mbox{Wr}$ is the Wronskian of $H$ functions. We also denote by
\begin{equation}\label{eq:vnu}
 \underline{v}^{\snu}=(v_1^{\snu},\dots,v_N^{\snu}) \in \C^N \;,
\end{equation}
the, not necessarily pairwise distinct, roots of $P^{\snu}(t)$.

The main result of \cite{coma20} (transcribed in the notation of the present paper) can be condensed in the following statements.
Assuming that the partition $[\nu]$ satisfies the Conjecture 5.9 in \cite{coma20},
if $\ell$ is sufficiently large there exists a unique (up to symmetries of the $v_k^{\snu}$'s)
solution of the algebraic system (\ref{eq:algblz}) with asymptotics
\begin{align}
& \label{eq:zkexpansion}
z_k^{\snu}=\frac{\ell^2}{\alpha}+ \frac{(2\alpha+2)^{\frac34}}{\alpha}  v_k^{\snu} 
\ell^{\frac32}+ o\left(\ell^{-\frac32}\right) , \quad k\in\{1,\dots,N\}\;.
\end{align}
Moreover, denoting by $D^{\snu}(E;\ell)$ the corresponding spectral determinant (which we assume to be normalised
so that $D^{\snu}(0;\ell)=1$) then its zeroes (i.e. the energy levels) are simple and positive, and
have the following -- non-uniform -- asymptotic behaviour as $\ell\to+\infty$
\begin{multline}
\label{eq:Ekexpfinal}
E_k^{\snu}(\ell) =
\frac{\left(1+\alpha\right)}{\alpha^{\frac{\alpha}{1+\alpha}}} \ell^{\frac{2\alpha}{1+\alpha}}\left[1+\frac{\alpha}{1+\alpha}
\left(1+2\left(2+2\alpha\right)^{\frac12}\left(k+\frac12\right)\right)\ell^{-1}+O(\ell^{-2})\right]\\
\forall k\in\bb{Z}_p\setminus \bb{H}_{\snu}.
\end{multline}

\begin{remark}
The reason beyond the asymptotics (\ref{eq:zkexpansion}) and (\ref{eq:Ekexpfinal}) lie in
the following discovery of the authors \cite{coma20}: in the large momentum limit
the monster potentials converge, when appropriately scaled, to monodromy-free rational extensions of the
harmonic oscillators. The latter are potentials of the form $t^2-2 \frac{d^2}{dt^2} \ln P(t)$,
where $P$ is a polynomial such that the corresponding stationary Schr\"odinger equation has trivial monodromy about
every pole for every value of the energy
\footnote{The problem of studying these potentials, or the potentials isospectral with them, is very interesting in itself,
and there is a vast literature dedicate to it, see e.g. \cite{mckean82} and \cite{felder12}.}.
Every such a potential, as proven in \cite{oblomkov99}, can be obtained via a chain of Darboux transformations. Hence it is necessarily of the form
\begin{equation}\label{eq:oblomkov}
t^2-2 \frac{d^2}{dt^2} \ln P^{[\nu]}(t)\;,
\end{equation}
for some partition $\nu \vdash N$; moreover, see \cite{coma20}, the energy levels -- when imposing that the eigenfuctions decay at $\pm \infty$ -- for the above potential are
\begin{equation}\label{eq:comasp}
 \mathcal{E}_k^{[\nu]}= 2k+1 \;,\quad \forall k\in\bb{Z}_{\frac{N}{2}}\setminus \bb{H}_{[\nu]}\;.
\end{equation}
As it is proven in \cite{coma20}, (\ref{eq:zkexpansion}) follows from (\ref{eq:oblomkov}) and (\ref{eq:Ekexpfinal}) follows from (\ref{eq:comasp}).
\end{remark}

\subsection{The bijection between solutions of the BAE and monster potentials}
Let $Q^{\snu}(x;p)$ be the (unique, for $p$ sufficiently large)
solution of the BAE (\ref{eq:BAEQ}) with set of hole-numbers $\bb{H}_{\snu}$.
After Theorem \ref{cor:main}, its roots have the asymptotic expansion
\begin{align}
&x_k(p)=
A\,p^{\frac{2\alpha}{1+\alpha}}\left[1+\frac{\sqrt{2}\alpha}{\left(1+\alpha\right)^{\frac32}}\left(k+\frac12\right)p^{-1} +O(p^{-2})\right] \;,\quad k\in\bb{Z}_{p}\backslash\bb{H}_{\snu}\;,\notag \\
&\mbox{with } A=(1+\alpha) \left(\frac{1}{\sqrt{\pi\alpha}}
 \frac{\Gamma \left(\frac{1}{2 \alpha }\right)}{\Gamma \left(\frac{1+\alpha}{2 \alpha }\right)}\right)^{\frac{2 \alpha }{\alpha +1}} \;.
\label{eq:xkfinalm}
\end{align}
 Therefore comparing (\ref{eq:xkfinalm}) with (\ref{eq:Ekexpfinal}) using the relation (\ref{eq:ell}), we obtain the following immediate
 result 
 \begin{equation}\label{eq:DvsQ}
  D^{\snu}\left(E;\ell\right)=Q^{\snu}\left(%\frac{\tilde{A}}{A} 
  \frac{E}{\eta}; \frac{2\ell+1}{\alpha+1}\right)\;,
 \end{equation}
 with $\eta$ as per (\ref{eq:etadef}).
The latter equation establishes the bijection between normalised solutions of the BAE (\ref{eq:BAEQ}) and monster potentials,
if $\alpha>1$ and $p$ large.

\begin{remark}
 The present paper only deals with the case $\alpha>1$ but we expect that the exact ODE/IM correspondence (6.17) holds for all $\alpha>0$ provided
 $p$ is sufficiently large (depending on $\alpha$).
 In particular, the case $\alpha=1$ can be checked explicitly. In fact,
 according to \cite[Equation (5.24)]{coma20}, the exact energy levels of the monster potentials with $\alpha=1$ are
\begin{equation}\label{eq:specharmonic}
 E^{[\nu]}_k(\ell)= 2 \ell + 3 + 4 k \;,\quad \forall k\in\bb{Z}_p\setminus \bb{H}_{[\nu]}.
\end{equation}
These numbers coincide with the Bethe roots for the solutions of the BAE of Quantum KdV with $\beta^2=\frac12$ (that is $\alpha=1$),
as computed in \cite[Equations (4.30,4.33)]{baluzaII}, provided
the multiplicative constant, discussed in \cite[Footnote 1 and Equation (15)]{BLZ04}, relating the quantum mechanical energy
$E$ with the Quantum KdV spectral parameter $\la^2$ is taken into
account.
\end{remark}

\begin{remark}
We can now explain the origin of the formula for the solution of the linearised IBA equation in terms of WKB integrals,
see Theorem \ref{prop:fproperty}.
Let us in fact consider the anharmonic oscillator
\begin{equation}\label{eq:anharmonic}
\psi''(t)=\left(t^{2\alpha}+ \frac{\ell(\ell+1)}{t^2}-E\right)\psi(t) \;,\quad E \in \C \;,
\end{equation}
which corresponds to the case of an empty partition.
A standard WKB analysis leads us to the following quantisation condition in the large $\ell$ limit
\begin{equation*}
 I(E_k;\ell)\sim 2k+1\;,
\end{equation*}
where
\begin{equation*}
 I(E;\ell)= \frac{2}{\pi} \int_{t_1(E,\ell)}^{t_2(E,\ell)} \sqrt{E t^2-t^{2\alpha+2}-\left(\ell+\frac12\right)^2} \frac{dt}{t}\;,
\end{equation*}
with $0<t_1(E,\ell)<t_2(E,\ell)$ the two roots of %$t^{2\alpha+2}- E t^2 + \ell(\ell+1)$
$E t^2-t^{2\alpha+2}-\left(\ell+\frac12\right)^2$ on the positive real axis.
A simple computation yields
\begin{align*}
I(E;\ell)= \left(2\ell+1\right) S\left(E \left(\ell+\frac12\right)^{-\frac{2\alpha}{\alpha+1}}\right),
\end{align*}
where the function $S$ is as per (\ref{eq:Jfun}). Hence the WKB formula for the eigenvalues eventually reads
\begin{equation}\label{eq:WKBquant}
 S\left(E \left(\ell+\frac12\right)^{-\frac{2\alpha}{\alpha+1}}\right) \sim \frac{2 k+1}{2\ell+1} \;.
\end{equation}
In Theorem \ref{cor:main}, we have found that the same approximation formula holds -- provided $H=0$, $p=\frac{2\ell+1}{\alpha+1}$ and
$x_k= E_k \eta $ with $\eta$ as per (\ref{eq:etadef}) -- in terms of a 
particular solution of the linearised IBA equation (\ref{eq:IBAStL}). The two formulae must coincide and in fact
$S$ coincides with a particular solution of the linearised IBA equation, as we have proven in Proposition \ref{prop:tauJ}.
In other words, under the ODE/IM correspondence, in the large momentum limit
the linearisation of the IBA equation corresponds
to the WKB approximation!
A similar phenomenon was already observed in \cite[Section 3.2]{dorey20}.
 \end{remark}
\begin{remark}
 The ODE/IM between the Quantum KdV model and the monster potentials have been vastly generalised,
 see \cite{doreyreview,dorey20} and references therein. Of these generalisations
 there are families which contain the original one as a particular case. These correspond to
 \begin{enumerate}
  \item Massive deformations
 of Quantum KdV \cite{lukyanov10,dorey20};
 \item Higher rank generalisations of Quantum KdV that are associated to any affine algebra, see \cite{marava15,marava17,fh16} inspired by \cite{FF11}.
 \item Fields theories corresponding to the thermodynamic limit
 of the inhomogeneous XXZ chain \cite{baz21}.
 \end{enumerate}
In the cases (2) and (3) the nature of the BAE and of the Hilbert space of the theory is very similar to the one of Quantum KdV.
Moreover the analogue of the DDV equations are known \cite{zinn98}
and the analogue of monster potentials were defined \cite{FF11,fh16}
(even though its explicit realisation may look really \textit{monstrous}, see
\cite{mara18}). It is not too optimistic to think that a study of the ODE/IM correspondence
can be developed along the same lines as in \cite{coma20} and the present paper, and that an explicit
correspondence such as (\ref{eq:DvsQ}) can be established.
\end{remark}

\def\cprime{$'$} \def\cprime{$'$} \def\cprime{$'$} \def\cprime{$'$}
  \def\cprime{$'$} \def\cprime{$'$} \def\cprime{$'$} \def\cprime{$'$}
  \def\cprime{$'$} \def\cprime{$'$} \def\cydot{\leavevmode\raise.4ex\hbox{.}}
  \def\cprime{$'$} \def\cprime{$'$} \def\cprime{$'$}

\end{document}